\documentclass[
reprint,
superscriptaddress,
nofootinbib,
aps, prx
]{revtex4-1}

\usepackage{custom}
\usepackage{mathtools}

\usepackage{enumitem}
\newtheorem{Lemma}{Lemma}
\newtheorem{corol}{Corollary}

\usepackage{multirow}

\def\be{\begin{equation}}
\def\ee{\end{equation}}
\def\ba{\begin{eqnarray}}
\def\ea{\end{eqnarray}}

\newcommand\q{\quad}
\def\Nl{{\mathchoice
{\setbox0=\hbox{$\displaystyle\rm N$}\hbox{\hbox to0pt
{\kern0.4\wd0\vrule height0.9\ht0\hss}\box0}}
{\setbox0=\hbox{$\textstyle\rm N$}\hbox{\hbox to0pt
{\kern0.4\wd0\vrule height0.9\ht0\hss}\box0}}
{\setbox0=\hbox{$\scriptstyle\rm N$}\hbox{\hbox to0pt
{\kern0.4\wd0\vrule height0.9\ht0\hss}\box0}}
{\setbox0=\hbox{$\scriptscriptstyle\rm N$}\hbox{\hbox to0pt
{\kern0.4\wd0\vrule height0.9\ht0\hss}\box0}}}}
\def\Zl{{\mathchoice
{\setbox0=\hbox{$\displaystyle\rm Z$}\hbox{\hbox to0pt
{\kern0.4\wd0\vrule height0.9\ht0\hss}\box0}}
{\setbox0=\hbox{$\textstyle\rm Z$}\hbox{\hbox to0pt
{\kern0.4\wd0\vrule height0.9\ht0\hss}\box0}}
{\setbox0=\hbox{$\scriptstyle\rm Z$}\hbox{\hbox to0pt
{\kern0.4\wd0\vrule height0.9\ht0\hss}\box0}}
{\setbox0=\hbox{$\scriptscriptstyle\rm Z$}\hbox{\hbox to0pt
{\kern0.4\wd0\vrule height0.9\ht0\hss}\box0}}}}
\def\Ql{{\mathchoice
{\setbox0=\hbox{$\displaystyle\rm Q$}\hbox{\hbox to0pt
{\kern0.4\wd0\vrule height0.9\ht0\hss}\box0}}
{\setbox0=\hbox{$\textstyle\rm Q$}\hbox{\hbox to0pt
{\kern0.4\wd0\vrule height0.9\ht0\hss}\box0}}
{\setbox0=\hbox{$\scriptstyle\rm Q$}\hbox{\hbox to0pt
{\kern0.4\wd0\vrule height0.9\ht0\hss}\box0}}
{\setbox0=\hbox{$\scriptscriptstyle\rm Q$}\hbox{\hbox to0pt
{\kern0.4\wd0\vrule height0.9\ht0\hss}\box0}}}}
\def\Rl{{\mathchoice
{\setbox0=\hbox{$\displaystyle\rm R$}\hbox{\hbox to0pt
{\kern0.4\wd0\vrule height0.9\ht0\hss}\box0}}
{\setbox0=\hbox{$\textstyle\rm R$}\hbox{\hbox to0pt
{\kern0.4\wd0\vrule height0.9\ht0\hss}\box0}}
{\setbox0=\hbox{$\scriptstyle\rm R$}\hbox{\hbox to0pt
{\kern0.4\wd0\vrule height0.9\ht0\hss}\box0}}
{\setbox0=\hbox{$\scriptscriptstyle\rm R$}\hbox{\hbox to0pt
{\kern0.4\wd0\vrule height0.9\ht0\hss}\box0}}}}
\def\Cl{{\mathchoice
{\setbox0=\hbox{$\displaystyle\rm C$}\hbox{\hbox to0pt
{\kern0.4\wd0\vrule height0.9\ht0\hss}\box0}}
{\setbox0=\hbox{$\textstyle\rm C$}\hbox{\hbox to0pt
{\kern0.4\wd0\vrule height0.9\ht0\hss}\box0}}
{\setbox0=\hbox{$\scriptstyle\rm C$}\hbox{\hbox to0pt
{\kern0.4\wd0\vrule height0.9\ht0\hss}\box0}}
{\setbox0=\hbox{$\scriptscriptstyle\rm C$}\hbox{\hbox to0pt
{\kern0.4\wd0\vrule height0.9\ht0\hss}\box0}}}}
\def\Hl{{\mathchoice
{\setbox0=\hbox{$\displaystyle\rm H$}\hbox{\hbox to0pt
{\kern0.4\wd0\vrule height0.9\ht0\hss}\box0}}
{\setbox0=\hbox{$\textstyle\rm H$}\hbox{\hbox to0pt
{\kern0.4\wd0\vrule height0.9\ht0\hss}\box0}}
{\setbox0=\hbox{$\scriptstyle\rm H$}\hbox{\hbox to0pt
{\kern0.4\wd0\vrule height0.9\ht0\hss}\box0}}
{\setbox0=\hbox{$\scriptscriptstyle\rm H$}\hbox{\hbox to0pt
{\kern0.4\wd0\vrule height0.9\ht0\hss}\box0}}}}
\def\Ol{{\mathchoice
{\setbox0=\hbox{$\displaystyle\rm O$}\hbox{\hbox to0pt
{\kern0.4\wd0\vrule height0.9\ht0\hss}\box0}}
{\setbox0=\hbox{$\textstyle\rm O$}\hbox{\hbox to0pt
{\kern0.4\wd0\vrule height0.9\ht0\hss}\box0}}
{\setbox0=\hbox{$\scriptstyle\rm O$}\hbox{\hbox to0pt
{\kern0.4\wd0\vrule height0.9\ht0\hss}\box0}}
{\setbox0=\hbox{$\scriptscriptstyle\rm O$}\hbox{\hbox to0pt
{\kern0.4\wd0\vrule height0.9\ht0\hss}\box0}}}}
\newcommand{\cc}{\mathcal C}
\newcommand{\cd}{\mathcal D}

\newcommand{\cg}{\mathcal G}
\newcommand{\ch}{\mathcal H}

\newcommand{\cl}{\mathcal L}

\newcommand{\cp}{\mathcal P}

\newcommand{\calr}{\mathcal R}
\newcommand{\cs}{\mathcal S}
\newcommand{\ct}{\mathcal T}

\newcommand{\intsum}{\mathclap{\displaystyle\int}\mathclap{\textstyle\sum}}

\newcommand{\varep}{\varepsilon}

\def\nn{\nonumber}

\newcommand{\eqa}{\begin{eqnarray}}
\newcommand{\neqa}{\end{eqnarray}}

\def\f{\frac}

\usepackage{bbm}

\def\q{{\quad}}

\begin{document}

\title{
Equivalence of approaches to relational quantum dynamics in relativistic settings
}

\author{Philipp A. H\"{o}hn}
\email[]{philipp.hoehn@oist.jp}
\affiliation{Okinawa Institute of Science and Technology Graduate University, Onna, Okinawa 904 0495, Japan}\affiliation{Department of Physics and Astronomy, University College London, London, United Kingdom}

\author{Alexander R. H. Smith}
\email[]{alexander.r.smith@dartmouth.edu}
\affiliation{Department of Physics and Astronomy, Dartmouth College, Hanover, New Hampshire 03755, USA}

\author{Maximilian P. E. Lock}
\email[]{maximilian.lock@univie.ac.at}
\affiliation{Institute for Quantum Optics and Quantum Information (IQOQI), Austrian Academy of Sciences, A-1090 Vienna, Austria}

\date{\today}

\begin{abstract}
We have previously shown \cite{Hoehn:2019owq} that three approaches to relational quantum dynamics\,---\,relational Dirac observables, the Page-Wootters formalism and quantum deparametrizations\,---\,are equivalent. Here we show that this `trinity' of relational quantum dynamics holds in relativistic settings per frequency superselection sector. We ascribe the time according to the clock subsystem to a POVM which is covariant with respect to its (quadratic) Hamiltonian. This differs from the usual choice of a self-adjoint clock observable conjugate to the clock momentum. It also  resolves Kucha\v{r}'s criticism that the Page-Wootters formalism yields incorrect localization probabilities for the relativistic particle when conditioning on a Minkowski time operator. We show that conditioning instead on the covariant clock POVM results in a Newton-Wigner type localization probability commonly used in relativistic quantum mechanics. By establishing the equivalence mentioned above, we also assign a consistent conditional-probability interpretation to relational observables and deparametrizations. Finally, we expand a recent method of changing temporal reference frames, and show how to transform states and observables frequency-sector-wise. We use this method to discuss an indirect clock self-reference effect and explore the state and temporal frame-dependence of the task of comparing and synchronizing different quantum clocks.
\end{abstract}

\maketitle

\tableofcontents
\section{Introduction}\label{sec_intro}

In general relativity, time plays a different role than in classical and quantum mechanics, or quantum field theory on a Minkowski background. General covariance dispenses with a preferred choice of time and introduces instead a dynamical notion of time which depends on solutions to the Einstein field equations. {In the} canonical approach to quantum gravity this leads to the infamous \emph{problem of time} \cite{kucharTimeInterpretationsQuantum2011a,Isham1993,andersonProblemTime2017}. Its most well-known facet is that, due to the constraints of the theory, quantum states of spacetime (and any matter contained in it) do not at first sight appear to undergo any time evolution, in seeming contradiction with everyday experience.

{The resolution comes from one of the key insights of general relativity: any physical notion of time is relational, the degrees of freedom of the {U}niverse evolve relative to one another~\cite{rovelliQuantumGravity2004,smolinCaseBackgroundIndependence2006,smolinTemporalRelationalism2018}.} This insight has led to three main relational approaches to the problem of time, each of which seeks to extract a notion of time from within the quantum degrees of freedom, relative to which the others evolve: 
\begin{itemize}
\item[(i)] a Dirac quantization scheme, wherein relational observables are constructed that encode correlations between evolving and clock degrees of freedom \cite{dewittQuantumTheoryGravity1967,rovelliQuantumGravity2004,rovelliThereIncompatibilityWays1991,ashtekarLecturesNonPerturbativeCanonical1991,thiemannModernCanonicalQuantum2008,Rovelli:1990jm, Rovelli:1989jn,Rovelli:1990ph,Rovelli:1990pi,rovelliQuantumGravity2004, kucharTimeInterpretationsQuantum2011a,Isham1993,Marolf:1994nz,Marolf:1994wh,Gambini:2000ht,Tambornino:2011vg,Giddings:2005id,dittrichPartialCompleteObservables2007,Dittrich:2005kc, Dittrich:2006ee,  Dittrich:2007jx,Gambini:2008ke, hoehnHowSwitchRelational2018,Hoehn:2018whn,Hoehn:2019owq,Bojowald:2010xp, Bojowald:2010qw, Hohn:2011us,Dittrich:2016hvj,Dittrich:2015vfa,Chataignier:2019kof,chataignier2020relational}, 
\item[(ii)] the Page-Wootters formalism, which defines a relational dynamics in terms of conditional probabilities for clock and evolving degrees of freedom \cite{pageEvolutionEvolutionDynamics1983, woottersTimeReplacedQuantum1984,Gambini:2006ph, Gambini:2006yj,Gambini:2008ke, giovannettiQuantumTime2015,Smith:2017pwx,Smith:2019imm,Dolby:2004ak, castro-ruizTimeReferenceFrames2019,Boette:2018uix,Diaz:2019bad,Diaz:2019xie,Singh:2020kdu,leonPauliObjection2017, marlettoEvolutionEvolutionAmbiguities2017,Nikolova:2017huj, baumann2019generalized,Hoehn:2019owq,favalli2020hermitian,foti2020time}, and 
\item[(iii)] classical or quantum deparametrizations, which result in a reduced quantum theory that only treats the evolving degrees of freedom {as quantum}~\cite{ashtekarLecturesNonPerturbativeCanonical1991,kucharTimeInterpretationsQuantum2011a,Isham1993,hoehnHowSwitchRelational2018,Hoehn:2018whn,Hoehn:2019owq,Thiemann:2004wk,Bojowald:2019mas}. 
\end{itemize}

These three approaches have been pursued largely independently with the relation between them previously unknown. They have also not been without criticism, especially the Page-Wootters formalism. For example, Kucha\v{r}~\cite{kucharTimeInterpretationsQuantum2011a} raised three fundamental criticisms against this approach, namely that it:
\begin{itemize}
\item[(a)] leads to wrong localization probabilities in relativistic settings, 
\item[(b)]  is in conflict with the constraints of the theory, and 
\item[(c)] yields wrong propagators. 
\end{itemize}
Concern has also been voiced that there is an inherent ambiguity in terms of which clock degrees of freedom one should choose, also known as the \emph{multiple choice problem}~\cite{kucharTimeInterpretationsQuantum2011a,Isham1993,andersonProblemTime2017,Albrecht:2007mm,Albrecht:2012wsd}. Indeed, in generic general relativistic systems there is no preferred choice of relational time variable and different choices may lead to \emph{a priori} different quantum theories.

In our recent work \cite{Hoehn:2019owq} we addressed the relation between these three approaches (i)--(iii) to relational quantum dynamics, demonstrating that they are, in fact, equivalent when the clock Hamiltonian features a continuous and non-degenerate spectrum. Specifically, we constructed the explicit transformations mapping each formulation of relational quantum dynamics into the others. These maps revealed the Page-Wootters formalism (ii) and quantum deparametrizations (iii) as \emph{quantum symmetry reductions} of the manifestly gauge-invariant formulation~(i). In other words, the Page-Wootters formalism (ii) and quantum deparametrizations (iii) can be regarded as quantum analogs of gauge-fixed formulations of gauge-invariant quantities (i). Conversely, the formulation in terms of relational Dirac observables (i) constitutes the quantum analog of a gauge-invariant extension of {the} gauge-fixed formulations (ii) and (iii). More physically, these transformations establish (i) as a clock-choice-neutral (in a sense explained below), (ii) as a relational Schr\"odinger, and (iii) as a relational Heisenberg picture of the dynamics. Constituting three faces of the same quantum dynamics, we called the equivalence of (i)--(iii) the \emph{trinity of relational quantum dynamics}. 

This equivalence not only provides relational Dirac observables with a consistent conditional probability interpretation, but also resolves Kucha\v{r}'s criticism (b) that the Page-Wootters formalism would be in conflict with the quantum constraints. Furthermore, the trinity resolves Kucha\v{r}'s criticism (c) that the Page-Wootters formalism would yield wrong propagators, by showing that the correct propagators always follow from manifestly gauge-invariant conditional probabilities on the physical Hilbert space \cite{Hoehn:2019owq}. This resolution of criticism (c) differs from previous resolution proposals which relied on ideal clocks \cite{Gambini:2008ke,corbinSemiClassicalLimitMinimum2009, giovannettiQuantumTime2015} and auxiliary ancilla systems \cite{giovannettiQuantumTime2015} and can be viewed as an extension of \cite{Dolby:2004ak}.

The transformations between (i)--(iii) of the trinity also allowed us to address the multiple choice problem in \cite{Hoehn:2019owq} by extending a previous method for changing temporal reference frames, i.e.\ clocks, in the quantum theory  \cite{hoehnHowSwitchRelational2018,Hoehn:2018whn,castro-ruizTimeReferenceFrames2019} (see also \cite{Bojowald:2010xp,Bojowald:2010qw,Hohn:2011us,Bojowald:2016fac}). {The resolution to the problem lies in part in realizing that a solution to the Wheeler-DeWitt equation encodes the relations between all subsystems, including the relations between subsystems employed as clocks to track the dynamics of other subsystems; there are multiple choices of clocks, each of which can be used to define dynamics}. Our proposal is thus to turn the multiple choice problem into a \emph{feature} by having a multitude of quantum time choices at our disposal, which we are able to connect through quantum temporal frame transformations. This is in line with developing a genuine quantum implementation of general covariance \cite{giacominiQuantumMechanicsCovariance2019,Vanrietvelde:2018pgb,Vanrietvelde:2018dit,hoehnHowSwitchRelational2018,Hoehn:2018whn,Hoehn:2019owq,giacominiRelativisticQuantumReference2019,hamette2020quantum,chataignier2020relational}. This proposal is part of current efforts to develop a general framework of quantum reference frame transformations (and study their physical consequences \cite{loveridgeSymmetryReferenceFrames2018a,yang2020switching,Gielen:2020abd,savi2020quantum,le2020blurred, tuziemski2020decoherence,Hardy:2018kbp,Hardy:2019cef,Guerin:2018fja,Zych:2018nao,Barbado:2020snx}), and should be contrasted with other attempts at resolving the multiple choice problem by identifying a preferred choice of clock~\cite{marlettoEvolutionEvolutionAmbiguities2017} (see~\cite{Hoehn:2019owq} for further discussion of this proposal).

We did not address Kucha\v{r}'s criticism (a) that the Page-Wootters formalism yields the wrong localization probabilities for relativistic models in  \cite{Hoehn:2019owq} as they feature clock Hamiltonians which are quadratic in momenta and thus generally have a degenerate spectrum, splitting into positive and negative frequency sectors. This degeneracy is not covered by our previous construction.
{While} quadratic clock Hamiltonians are standard in the literature on relational observables (approach (i)) and deparametrizations (approach (iii)), see e.g.\ \cite{rovelliQuantumGravity2004,ashtekarLecturesNonPerturbativeCanonical1991,thiemannModernCanonicalQuantum2008,Tambornino:2011vg,Hoehn:2018whn},  relativistic particle models  have {only} recently been studied in the Page-Wootters formalism (approach (ii)) \cite{Diaz:2019bad,Smith:2019imm,Diaz:2019xie,Singh:2020kdu}. {However, Kucha\v{r}'s criticism~{(a)} that the Page-Wootters approach yields incorrect localization probabilities in relativistic settings has yet to be addressed.}
{Since the Page-Wootters formalism encounters challenges in relativistic settings, given the equivalence of relational approaches implied by the trinity, one might worry about relational observables and deparametrizations too.}

In this article, we show that these challenges can be overcome, and a consistent interpretation of the relational dynamics can be provided. To this end, we extend the trinity to quadratic clock Hamiltonians, thus encompassing many relativistic settings; we show that all the results of \cite{Hoehn:2019owq} hold per frequency sector associated to the clock due to a superselection rule induced by the Hamiltonian constraint. Frequency-sector-wise, the relational dynamics encoded in (i) relational observables, (ii) the Page-Wootters formalism, and (iii) quantum deparametrizations are thus also fully equivalent. 

The key to our construction, as in \cite{Hoehn:2019owq}, is the use of a Positive-Operator Valued Measure (POVM) which here transforms covariantly~with respect to the quadratic clock Hamiltonian \cite{holevoProbabilisticStatisticalAspects1982,buschOperationalQuantumPhysics,busch1994time,braunsteinGeneralizedUncertaintyRelations1996} as a time observable. This contrasts with the usual approach of employing an operator conjugate to the clock momentum (i.e.\ the Minkowski time operator in the case of a relativistic particle). This covariant clock POVM is instrumental in our resolution of Kucha\v{r}'s criticism (a) that the Page-Wootters formalism yields wrong localization probabilities for relativistic systems. We show that when conditioning on this covariant clock POVM rather than Minkowski time, one obtains a Newton-Wigner type localization probability \cite{haag2012local,Newton:1949cq}. While a Newton-Wigner type localization is approximate and not fully Lorentz covariant, due to the relativistic localization no-go theorems of Perez-Wilde \cite{FernandoPerez:1976ib} and Malament \cite{Malament1996} (see also \cite{Yngvason:2014oia,Papageorgiou:2019ezr}), it is generally accepted as the best possible localization in relativistic quantum mechanics. {(In quantum field theory localization is a different matter} \cite{haag2012local,Yngvason:2014oia}.) This demonstrates the advantage of using covariant clock POVMs in relational quantum dynamics \cite{Brunetti:2009eq,Smith:2017pwx,Smith:2019imm, Hoehn:2019owq,Loveridge:2019phw}. The trinity also extends the probabilistic interpretation of relational observables: a Dirac observable describing the relation between a position operator and the covariant clock POVM corresponds to a Newton-Wigner type localization in relativistic settings. 

Finally, we again use the equivalence  between (i)--(iii) to construct temporal frame changes in the quantum theory. On account of superselection rules across frequency sectors, temporal frame changes can only map information contained in the overlap of two frequency sectors, one associated to each clock, from one clock `perspective' to another. We apply these temporal frame change maps to explore {an indirect} clock self-reference and the temporal frame and state dependence of comparing and  synchronizing readings of different quantum clocks.

While completing this manuscript, we became aware of~\cite{chataignier2020relational}, which independently extends some results of~\cite{Hoehn:2019owq} on the conditional probability interpretation of relational observables and their equivalence with the Page-Wootters formalism into a more general setting. However, a different formalism \cite{Chataignier:2019kof} is used in \cite{chataignier2020relational}, which does not employ covariant clock POVMs and therefore the two works complement one another.

Throughout this article we work in units where $\hbar=1$.


\section{Clock-neutral formulation of classical and quantum mechanics}\label{sec_cRDOs2}

Colloquially, general covariance posits that the laws of physics are the same in every reference frame. This is usually interpreted as implying that physical laws should take the form of tensor equations. Tensors can be viewed as reference-frame-neutral objects: they define a description of physics prior to choosing a reference frame. They thereby encode the physics as `seen' by \emph{all} reference frames at once. If one wants to know the numbers which a measurement of the tensor in a particular reference frame would yield, one must contract the tensor with the vectors corresponding to that choice of frame. In this way, the description of the same tensor looks different relative to different frames, but the tensor per se, as a multilinear map, is reference-frame-neutral. It is this reference-frame-neutrality of tensors which results in the frame-independence of physical laws.

The notion of reference frame as a vector frame is usually taken to define the orientation of a local laboratory of some observer. In practice, one often implicitly identifies the local lab (i.e.\ the reference system relative to which the remaining physics is described) with the reference frame.
This is an idealization which ignores the lab's back-reaction on spacetime, interaction with other physical systems and possible internal dynamics, while at the same time assuming it to be sufficiently classical so that superpositions of orientations can be ignored. Such an idealization is appropriate in general relativity where the aim is to describe the large-scale structure of spacetime. However, {in quantum gravity, where the goal is to describe the micro-structure of spacetime,} this may no longer be appropriate. More generally, we may ask about the fate of general covariance when we take seriously the fact that {physically meaningful} reference frames are in practice always associated with physical systems, and as such are comprised of dynamical degrees of freedom that {may} couple with other systems, undergo their own dynamics and will ultimately be subject to the laws of quantum theory. What are then the reference-frame-neutral structures?

In regard to this question, we note that the {classical} notion of general covariance for  reference frames {associated to idealized local labs} is deeply intertwined with invariance under general coordinate transformations, i.e.\ passive diffeomorphisms. In moving towards non-idealized reference frames  (or rather systems), we shift focus from coordinate descriptions to dynamical reference degrees of freedom, relative to which  the remaining physics will be described. In line with this, we shift the focus from passive to active diffeomorphisms, which directly act on the dynamical degrees of freedom. This is advantageous for quantum gravity, where classical spacetime coordinates are \emph{a priori} absent. A quantum version of general covariance should be formulated in terms of dynamical reference degrees of freedom \cite{giacominiQuantumMechanicsCovariance2019,Vanrietvelde:2018pgb,Vanrietvelde:2018dit,hoehnHowSwitchRelational2018,Hoehn:2018whn,Hoehn:2019owq,giacominiRelativisticQuantumReference2019,hamette2020quantum,chataignier2020relational}.

The active symmetries imply a redundancy in the description of the physics. \emph{A priori} all degrees of freedom stand on an equal footing, giving rise to a freedom in choosing which of them shall be treated as the redundant ones. The key idea is to identify this choice with the choice of reference degrees of freedom, i.e.\ those relative to which the remaining degrees of freedom will be described.\footnote{{Indeed, we do not want to describe the reference degrees of freedom directly relative to themselves in order to avoid the self-reference problem \cite{dalla1977logical,breuer1995impossibility}. Nevertheless, through the perspective-neutral structure it is possible to construct indirect self-reference effects of quantum clocks through temporal frame changes, see \cite{Hoehn:2019owq} and Sec.~\ref{sec_selfref}.}} Accordingly, choosing a dynamical reference system amounts to removing redundancy from the description.  As such, we may interpret the   redundancy-containing description (in both the classical and quantum theory) as a \emph{perspective-neutral} description of physics, i.e.\ as a global description of physics prior to having chosen a reference system, from whose perspective  the remaining degrees of freedom are to be described \cite{Vanrietvelde:2018pgb,Vanrietvelde:2018dit,hoehnHowSwitchRelational2018,Hoehn:2018whn,Hoehn:2019owq}. This perspective-neutral structure is  thus proposed as the reference-frame-neutral structure for dynamical (i.e.\ non-idealized) reference systems.

In this article we focus purely on temporal diffeomorphisms and thus on temporal reference frames/systems, or simply clocks. In this case, we refer to the perspective-neutral structure as a \emph{clock-(choice-)neutral} structure \cite{hoehnHowSwitchRelational2018,Hoehn:2018whn,Hoehn:2019owq}, which we briefly review here in both the classical and quantum theory. It is a description of the physics, prior to having chosen a temporal reference system relative to which the dynamics of the remaining degrees of freedom are to be described.

\subsection{Clock-neutral classical theory}\label{sec_clneutral}

Consider a classical theory described by an  action 
$\mathcal{S} = \int_\mathbb{R} du \, L(q^a, dq^a/du )$, where $q^a$ denotes a collection of  configuration variables indexed by $a$. Such a theory exhibits \emph{temporal diffeomorphism invariance} if the action $\mathcal{S}$ is reparametrization invariant; that is, $L(q^a, dq^a/du )\mapsto L(q^a, dq^a/du' ) du'/du$ transforms as a scalar density under $u \mapsto u'(u)$. The Hamiltonian of such a theory is of the form $H = N(u) \,C_H$, where $N(u)$ is an arbitrary lapse function and 
\ba
C_H = \sum_a\,q^a\,p_a-L\approx0 , \label{Hconstriant}
\ea 
the so-called Hamiltonian constraint, is a consequence of the temporal diffeomorphism symmetry. This equation defines the constraint surface $\mathcal{C}$ inside the kinematical phase space $\mathcal{P}_{\rm kin}$, which is parametrized by the canonical coordinates $q^a,p_b$.  {The} $\approx$ denotes a weak equality, i.e.\ one which only holds on $\cc$ \cite{diracLecturesQuantumMechanics1964,Henneaux:1992ig}.

The Hamiltonian generates a dynamical flow on $\cc$, which  transforms an arbitrary phase space function $f$ according to
\begin{align}
\frac{df}{du} \ce \{f, C_H \}
\label{gaugeFlow}
\end{align}
and integrates to a finite transformation $\alpha_{C_H}^u\cdot f$, where for simplicity the lapse function has been chosen to be unity, $N(u) =1$. {Owing to the reparametrization invariance}, this  flow  should be interpreted as a gauge transformation rather than true evolution~\cite{rovelliQuantumGravity2004,thiemannModernCanonicalQuantum2008}, and thus for an observable $F$ to be physical, it must be invariant under such a transformation, i.e. 
\begin{align}
\{F, C_H \} \approx 0.
\label{DiracObservable}
\end{align}
Observables satisfying Eq.~\eqref{DiracObservable} are known as \emph{Dirac observables}.

In order to obtain a gauge-invariant dynamics, we have to choose a dynamical temporal reference system, i.e.\ a clock function $T(q^a,p_a)$, to parametrize the dynamical flow, Eq.~\eqref{gaugeFlow}, generated by the constraint. We can then describe the evolution of the remaining degrees of freedom relative to $T(q^a,p_a)$. This gives rise to so-called relational Dirac observables (a.k.a.\ evolving constants of motion) which encode {the answer to} the question ``what is the value of  the function $f$ along the flow generated by $C_H$ on $\cc$ when the clock $T$ reads $\tau$?''~\cite{Rovelli:1990jm, Rovelli:1989jn,Rovelli:1990ph, Rovelli:1990pi,rovelliQuantumGravity2004,dittrichPartialCompleteObservables2007,Dittrich:2005kc, Dittrich:2006ee,  Dittrich:2007jx, Tambornino:2011vg,thiemannModernCanonicalQuantum2008,hoehnHowSwitchRelational2018,Hoehn:2018whn,Hoehn:2019owq}. We will denote such an observable by $F_{f,T}(\tau)$. As shown in \cite{dittrichPartialCompleteObservables2007,Dittrich:2005kc, Dittrich:2006ee,  Dittrich:2007jx}, these observables can be constructed by solving $\alpha^u_{C_H}\cdot T=\tau$ for $u=u_T(\tau)$ and setting 
\begin{align}
F_{f,T}(\tau) &\ce \alpha_{C_H}^{u}\cdot f\, \Big|_{u=u_T(\tau)} \nn \\
&\approx\sum_{n=0}^{\infty}\, \frac{\left(\tau-T\right)^n}{n!}  \left\{f,\frac{C_H}{\{T,C_H\}}\right\}_n,
\label{RelationalDiracObservable1}
\end{align}
where $\{f,g\}_n :=\{\{f,g\}_{n-1},g\}$ is the $n^\text{th}$-nested Poisson bracket subject to $\{f,g\}_0:=f$. The $F_{f,T}(\tau)$ satisfy Eq.~\eqref{DiracObservable} and thus constitute a family of Dirac observables parametrized by $\tau$. Such relational observables are so-called gauge-invariant extensions of gauge-fixed quantities \cite{dittrichPartialCompleteObservables2007,Dittrich:2005kc, Dittrich:2006ee,  Dittrich:2007jx,Henneaux:1992ig,Hoehn:2019owq,Chataignier:2019kof}. 

In generic models there is no preferred choice for the clock function $T$ among the degrees of freedom on $\cp_{\rm kin}$, which is sometimes referred to as the \emph{multiple choice problem} \cite{kucharTimeInterpretationsQuantum2011a,Isham1993}. Different choices of $T$ will lead to different relational Dirac observables, as can be seen in Eq.~\eqref{RelationalDiracObservable1}. All these different choices are encoded in the constraint surface $\cc$ and stand \emph{a priori} on an equal footing.

This gives rise to the interpretation of $\cc$ as a clock-neutral structure. The temporal diffeomorphism symmetry leads to a redundancy in the description of $\cc$: thanks to the Hamiltonian constraint the kinematical canonical degrees of freedom are not independent and due to its gauge flow there will only be $\dim\cp_{\rm kin}-2$ independent physical phase space degrees of freedom. In particular, relative to any choice of clock function $T$ one can construct $\dim\cp_{\rm kin}-2$ independent relational Dirac observables using Eq.~\eqref{RelationalDiracObservable1} \cite{diracLecturesQuantumMechanics1964,Henneaux:1992ig}. Hence, the relational Dirac observables relative to any other clock choice $T'$ can be constructed from them. Consequently, there is redundancy among the relational Dirac observables relative to different clock choices.  
Thus $\cc$ yields a description of the physics prior to choosing and fixing a clock relative to which the gauge-invariant dynamics of the remaining degrees of freedom can be described. Specifically, no choice has been made as to which of the kinematical and physical degrees of freedom are to be considered as redundant.
In analogy to the tensor case, $\cc$ still contains the information about all clock choices and their associated relational dynamics at once; it yields a clock-neutral description.

Being of odd dimension $\dim\cp_{\rm kin}-1$, $\cc$ is also not a phase space. A proper phase space description can be obtained, e.g.\ through phase space reduction by gauge-fixing~\cite{hoehnHowSwitchRelational2018,Hoehn:2018whn,Hoehn:2019owq,Thiemann:2004wk,chataignier2020relational}. Given a choice of clock function $T$, we may consider the gauge-fixing condition $T=const$, which may be valid only locally on $\cc$. Since $F_{f,T}(\tau)$ is constant along each orbit generated by $C_H$ for each value of $\tau$, we do not lose any information about the relational dynamics by restricting to $T=const$ and leaving $\tau$ free. By restricting to the relational observables $F_{f,T}(\tau)$ relative to clock $T$ and by solving the two conditions $T=const$, $C_H=0$, we remove {the} redundancy from among both the kinematical and physical degrees of freedom. The surviving reduced phase space description, which no longer contains the clock degrees of freedom as dynamical variables, can be interpreted as the description of the dynamics relative to the temporal reference system defined by the clock function $T$. {But now we keep track of time evolution not in terms of the dynamical $T$, but in terms of the parameter $\tau$ representing its `clock readings'.} In particular, the temporal reference system is not described relative to itself, e.g.\ one finds the tautology $F_{T,T}(\tau)\approx\tau$. Accordingly, choosing the `perspective' of a clock means choosing the corresponding clock degrees of freedom as the redundant ones and removing them. {The theory is then deparametrized: it no longer contains a gauge-parameter $u$, nor a constraint, nor dynamical clock variables\,---\,only true evolving degrees of freedom.}

\subsection{Clock-neutral quantum theory}\label{sec_cnqt}

Following the Dirac prescription for quantizing constrained systems~\cite{diracLecturesQuantumMechanics1964, Henneaux:1992ig,ashtekarLecturesNonPerturbativeCanonical1991,thiemannModernCanonicalQuantum2008}, one first promotes the canonical coordinates of $\cp_{\rm kin}$ to canonical position and momentum operators $\hat{q}^a$ and $\hat{p}_a$ acting on a \emph{kinematical Hilbert space} $\mathcal{H}_{\rm kin}$. The Hamiltonian constraint in Eq.~\eqref{Hconstriant} is {then} imposed by demanding that physical states of the quantum theory are annihilated by the quantization of the constraint function
\begin{align}
\hat{C}_H \ket{\psi_{\rm phys}} = 0.
\label{Wheeler-DeWitt}
\end{align}
Solutions to this Wheeler-DeWitt-like equation may be constructed from kinematical states $\ket{\psi_{\rm kin }} \in\mathcal{H}_{\rm kin}$ via a group averaging operation~\cite{marolfRefinedAlgebraicQuantization1995,Hartle:1997dc,Giulini:1998rk,Giulini:1998kf,Marolf:2000iq, thiemannModernCanonicalQuantum2008}\footnote{In contrast to \cite{marolfRefinedAlgebraicQuantization1995,Hartle:1997dc,Giulini:1998rk,Giulini:1998kf,Marolf:2000iq, thiemannModernCanonicalQuantum2008} and for notational simplicity, we refrain from using the more rigorous formulation in terms of Gel'fand triples and algebraic duals of (dense subsets of) Hilbert spaces. However, the remainder of this article could be put into such a more precise formulation.}
\begin{align}
\ket{\psi_{\rm phys}} = \delta(\hat{C}_H) \ket{\psi_{\rm kin }} = \frac{1}{2\pi}\int_G du\, e^{-i \hat{C}_H u} \ket{\psi_{\rm kin }},
\label{projector}
\end{align}
where $G$ parametrizes the group generated by $\hat{C}_H$. Physical states are not normalizable in $\ch_{\rm kin}$ if they are improper eigenstates of $\hat C_H$ (i.e.\ if zero lies in the continuous part of its spectrum). However, they are normalized with respect to the so-called \emph{physical inner product} 
\begin{align}
\braket{\psi_{\rm phys}|\phi_{\rm phys}}_{\rm phys}&\ce\langle\psi_{\rm kin}|\delta(\hat C_H)|\phi_{\rm kin}\rangle_{\rm kin}\label{PIP22}
\end{align}
where $\braket{\cdot|\cdot}_{\rm kin}$ is the {kinematical }inner product  and $|\psi_{\rm kin}\rangle, \ket{\phi_{\rm kin}} \in\ch_{\rm kin}$ reside in the equivalence class of states mapped to the same $\ket{\psi_{\rm phys}}, \ket{\phi_{\rm phys}}$ under the projection in Eq.~\eqref{projector}. Equipped with this inner product, the space of solutions to the Wheeler-DeWitt equation  in Eq.~\eqref{Wheeler-DeWitt} can usually be Cauchy completed to form the so-called \emph{physical Hilbert space} $\ch_{\rm phys}$ \cite{marolfRefinedAlgebraicQuantization1995,Hartle:1997dc,Giulini:1998rk,Giulini:1998kf,Marolf:2000iq, thiemannModernCanonicalQuantum2008}.

A gauge-invariant (i.e.\ physical) observable  $\hat{F}$ acting on $\mathcal{H}_{\rm phys}$ must satisfy  the quantization of Eq.~\eqref{gaugeFlow}
\begin{align}
\left[ \hat{C}_H, \hat{F} \right]\,\ket{\psi_{\rm phys}} = 0.
\end{align}
Such an observable $\hat{F}$ is a \emph{quantum Dirac observable}.

Clearly, $\exp(-i\,u\,\hat C_H)\,\ket{\psi_{\rm phys}}=\ket{\psi_{\rm phys}}$, i.e.\ physical states do not evolve under the dynamical flow generated by the Hamiltonian constraint. This is the basis of the so-called \emph{problem of time} in quantum gravity \cite{kucharTimeInterpretationsQuantum2011a,Isham1993,andersonProblemTime2017}, and of statements that a quantum theory defined by a Hamiltonian constraint is timeless. However, such a theory is only `background-timeless', i.e.\ physical states do not evolve with respect to the `external' gauge parameter $u$ parametrizing the group  generated by the Hamiltonian constraint. Instead, it is more appropriate to regard the quantum theory on $\ch_{\rm phys}$ as a \emph{clock-neutral} quantum theory: it is a global description of the physics prior to choosing an internal clock relative to which to describe the dynamics of the remaining degrees of freedom, as argued in~\cite{hoehnHowSwitchRelational2018,Hoehn:2018whn,Hoehn:2019owq}. Just as in the classical case, there will in general be many possible clock choices and the `quantum constraint surface' $\ch_{\rm phys}$ contains the information about all these choices at once; it is thus by no means `internally timeless'.

The goal is to suitably quantize the relational Dirac observables in Eq.~\eqref{RelationalDiracObservable1}, promoting them to families of operators $\hat F_{f,T}(\tau)$ on $\ch_{\rm phys}$. This involves a quantization of the temporal reference system $T$ and it is clear that in the quantum theory different choices of $T$ will also lead to different quantum relational Dirac observables. This will give rise to a multitude of gauge-invariant, relational quantum dynamics, {each} expressed with respect to the evolution parameter $\tau$, which corresponds to the readings of the chosen quantum clock (and is thus \emph{not} a gauge parameter). The quantization of relational observables is non-trivial, especially because Eq.~\eqref{RelationalDiracObservable1} may not be globally defined on $\cc$, and depends very much on the properties of the chosen clock. Steps towards systematically quantizing relational Dirac observables have been undertaken e.g.\ in  \cite{Hoehn:2019owq,Marolf:1994wh,Giddings:2005id,Chataignier:2019kof,chataignier2020relational} and part of this article is devoted to further developing them for a class of relativistic models.

In analogy to the classical case, the clock-neutral description on the `quantum constraint surface' $\ch_{\rm phys}$ is redundant: since the constraint is satisfied, not all the degrees of freedom are independent. In particular, the sets of quantum relational Dirac observables relative to different clock choices\,---\,and thus different relational quantum dynamics\,---\,will be interdependent. The proposal is once more to associate the choice of clock with the choice of redundant degrees of freedom; moving to the `perspective' of a given clock means considering the quantum relational observables relative to it as the independent ones, and removing the (now redundant) dynamical clock degrees of freedom altogether. This works through a quantum symmetry reduction procedure, i.e.\ the quantum analog of phase space reduction, which is tantamount to a \emph{quantum deparametrization} and has been developed in \cite{hoehnHowSwitchRelational2018,Hoehn:2018whn,Hoehn:2019owq} and will be further developed in Sec.~\ref{sec_trinity}. In particular, this  procedure is at the heart of changing from a description relative to one quantum clock to  one relative to another {clock}, which we elaborate on in Sec.~\ref{sec_cqc}. As such, quantum symmetry reduction is the key element of a proposal for exploring a quantum version of general covariance \cite{Vanrietvelde:2018pgb,Vanrietvelde:2018dit,hoehnHowSwitchRelational2018,Hoehn:2018whn,Hoehn:2019owq,giacominiQuantumMechanicsCovariance2019, castro-ruizTimeReferenceFrames2019} and thereby also addressing the multiple choice problem in quantum gravity and cosmology \cite{kucharTimeInterpretationsQuantum2011a,Isham1993} {(see also \cite{Bojowald:2010xp,Bojowald:2010qw,Hohn:2011us,Bojowald:2016fac,Gielen:2020abd})}.

\section{Quadratic clock Hamiltonians}
\label{Quadratic clock Hamiltonians}

Building upon the clock-neutral discussion, we now assume that the kinematical degrees of freedom described by $\cp_{\rm kin}$ and $\ch_{\rm kin}$ split into a clock $C$ and an ``evolving'' system $S$, which do not interact. This will permit us to choose a temporal reference system in the next section, and thence define a relational dynamics in both the classical and quantum theories.

\subsubsection{Classical theory}

Suppose the classical theory describes a clock $C$ associated with the phase space $\mathcal{P}_C \simeq T^* \mathbb{R} \simeq \mathbb{R}^2$, and some system of interest $S$ associated with a phase space $\mathcal{P}_S$,  so that the kinematical phase space decomposes as $\cp_{\rm kin} {\ce} \cp_C\oplus\cp_S$. We assume $\cp_C$ to be parametrized by the canonical pair $(t,p_t)$, but will not need to be specific about the structure of $\cp_S$ (other than assuming it to be a finite dimensional symplectic manifold). Further suppose that the clock and system are not coupled, leading to a Hamiltonian constraint function that is a sum of their respective Hamiltonians\footnote{The assumption that clock and system do not interact does not hold in generic general relativistic systems. However, it is satisfied  in some commonly used examples (see \cite{Hoehn:2019owq} for a discussion and Table~\ref{Table:Examples} for some examples).}
\begin{align}
C_H = H_C + H_S \approx 0,
\label{firstConstraint}
\end{align}
where $H_C$ is a function  on $\mathcal{P}_C$ and $H_S$ is a function on~$\mathcal{P}_S$.

This article concerns clock Hamiltonians that are quadratic in the clock momentum, $H_C = s\, p_t^2/2$, where $s\in\{-1,+1\}$, so that the Hamiltonian constraint becomes
\begin{align}
C_H = s\,\f{p_t^2}{2}+ H_S\approx0.\label{constraint2}
\end{align}
This class of clock Hamiltonians appears in a wide number of (special and general) relativistic and non-relativistic models\,---\,see Table~\ref{Table:Examples} for examples. They are doubly degenerate; every value of $H_C$ has two solutions in terms of $p_t$, except on the line defined by $p_t=0$. Note that $p_t$ is a Dirac observable.

The constraint in Eq.~\eqref{constraint2} can be factored into two constraints, each linear in $p_t$, and which in the $s=-1$ case define the positive and negative frequency modes in the quantum theory \cite{hoehnHowSwitchRelational2018,Hoehn:2018whn,Bojowald:2010qw}: 
\ba
\!\!\!\!\!C_H=s\,C_+\cdot C_-\,, \,\q\text{for} \q C_\sigma \ce \f{p_t}{\sqrt{2}}+\sigma\,\sqrt{-s\,H_S}\,,\label{factorize}
\ea
where we have introduced the degeneracy label $\sigma=\pm1$.  Note that Eq.~\eqref{constraint2} forces $s\,H_S$ to take non-positive values on $\cc$.
For simplicity, we shall henceforth refer to $\sigma=+1$ as positive and $\sigma=-1$ as negative frequency modes for \emph{both} $s=\pm1$.\footnote{We emphasize that $\sigma=-1$ denotes negative frequency modes and {\it not} that momenta take values $p_t\leq0$. Indeed, for the $\sigma$-modes, momenta satisfy $\sigma\,p_t\leq0$, which follows from setting $C_\sigma=0$.} It follows that we can decompose the constraint surface into a positive and a negative frequency sector  \cite{hoehnHowSwitchRelational2018,Hoehn:2018whn}
 \ba
\cc= \cc_+\cup\,\cc_-\,,\label{Cdecomp}
\ea
where $\cc_\sigma$ is the set of solutions to $C_\sigma=0$ in $\cp_{\rm kin}$. The intersection $\cc_+\cap\,\cc_-$ is defined by $p_t = H_S=0$ (see Fig.~\ref{fig_cc} for an illustration).

\begin{table}[t]

\begin{tabular}{c}
{Examples} of {constraints of the form in Eq.~\eqref{constraint2}}
\\ \ \vspace{-5pt} \\ 
\hline\hline \vspace{-5pt}
\\
{Non-relativistic particle {and arbitrary system}} \\
${C}_H =\frac{{\bm p}^2}{2m} + H_S$
\\   \ \vspace{-5pt} \\ 
\hline \vspace{-5pt}
\\
{Relativistic particle in inertial coordinates}\\
${C}_H = -{p}_{t}^2 +\bm p^2 +m^2$
\\ \ \vspace{-5pt} \\ 
\hline \vspace{-5pt} 
\\
Isotropic cosmology with massless scalar field \\
$C_H=p_\phi^2-p_\alpha^2-4k\,\exp(4\alpha)$
\\ \ \vspace{-5pt} \\ 
\hline \vspace{-5pt} 
\\
Homogeneous cosmology (vacuum Bianchi models)\\
$
{C}_H = -\frac{1}{2} \bar{p}_0^2+k_0\,\exp(2\sqrt{2}\bar{\beta}^0)+\f{1}{2}p_+^2+k_+\,\exp(-4\sqrt{3}\bar{\beta}^+)$\\ \ \vspace*{.5pt}
$+\f{k_-}{2}p_-^2
$
\\ \ \vspace{-5pt} \\ 
\hline \vspace{-5pt} 
\\
\end{tabular}
\label{Table:Examples}
\caption{Some examples of constraints of the form of Eq.~\eqref{constraint2}, i.e.\ with clock Hamiltonians quadratic in an appropriate canonical momentum. The last three (relativistic) examples each contain  both cases $s=\pm1$, depending on which degree of freedom is used to define the clock $C$. In the example of the Friedman-Lema\^itre-Robertson-Walker model with homogeneous massless scalar field we have used $\alpha:=\ln a$, where $a$ is the scale factor, and $k$ is the spatial curvature constant \cite{Blyth:1975is,Hawking:1983hn,Hajicek:1986ky,Kiefer:1988tr,Ashtekar:2011ni,Ashtekar:2007em,bojobuch} (here a choice of lapse function $N=e^{3\alpha}$ has been made and included in the definition of $C_H$). The shape of the Hamiltonian constraint for vacuum Bianchi models can be found, e.g., in~\cite{Ashtekar:1993wb} and holds for types I, II, III, VIII, IX and the Kantowski-Sachs models. Here $\bar\beta^0,\bar\beta^+,\bar\beta^-$ are linear combinations of the Misner anisotropy parameters and $k_0,k_+,k_-$ are constants, each of which may be zero, depending on the model. }
\end{table}

\begin{figure}[t]
\includegraphics[width= 245pt]{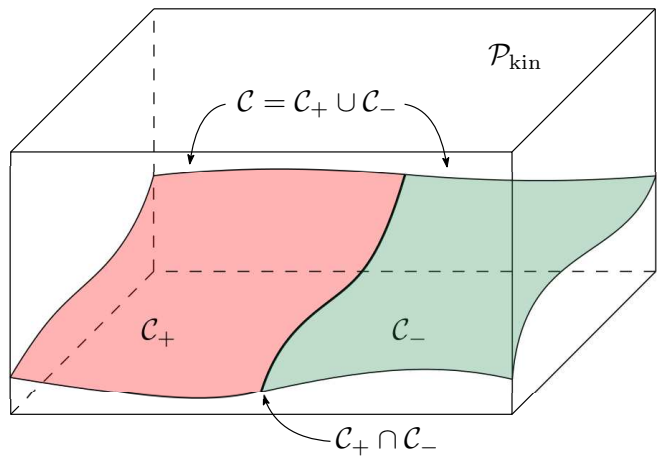}
\caption{
Depicted are the surfaces $\mathcal{C}_+$ (red) and $\mathcal{C}_-$ (green) defined by $C_+ = 0$ and $C_- = 0$, respectively. The union of these surfaces is the constraint surface $\mathcal{C} = \mathcal{C}_+ \cup \mathcal{C}_- \subset \mathcal{P}_{\rm kin}$, while their intersection $\mathcal{C}_+ \cap \mathcal{C}_-$ is characterized by $p_t = H_S = 0$, and {is} depicted by the thick black line. We have assumed that $H_S$ is not degenerate (see Fig.~1 of \cite{Hoehn:2018whn} for a similar depiction when $H_S$ is doubly degenerate). }
\label{fig_cc}
\end{figure}

\subsubsection{Quantum theory}

The Dirac quantization of the kinematical phase space $\cp_{\rm kin}=\cp_C\oplus\cp_S$ leads to the kinematical Hilbert space $\ch_{\rm kin} \simeq\ch_C\otimes\ch_S$ describing the clock and system, where $\ch_C\simeq L^2(\mathbb{R})$ and $\ch_S$ is the  Hilbert space associated with $S$. We assume the system Hamiltonian to be promoted to a self-adjoint operator $\hat H_S$ on $\ch_S$. An element of $\ch_{\rm kin}$ may be expanded in the eigenstates of the clock and system Hamiltonians as
\begin{align}
|\psi_{\rm kin} \rangle= {\ \,{\intsum}_{E}} \,\int_\mathbb{R} dp_t\,\psi_{\rm kin}(p_t,E)\,\ket{p_t}_C\ket{E}_S ,\nn
\end{align}
where the integral-sum highlights that $\hat H_S$ may either have a continuous or discrete spectrum.\footnote{The way we have written physical states implicitly assumes the non-positive part of the spectrum of $s\,\hat H_S$ to be non-degenerate. Were this not the case, additional degeneracy labels would be necessary. However, this would {not} otherwise affect the subsequent analysis. See \cite{Hoehn:2018whn} for an explicit construction of the  flat FLRW model with a massless scalar field, whose Hamiltonian constraint can also be interpreted as a free relativistic particle, and thus features a twofold system energy degeneracy.}

Physical states of the theory satisfy Eq.~\eqref{Wheeler-DeWitt}, which for the Hamiltonian constraint in Eq.~\eqref{constraint2} becomes 
\begin{align}
\hat C_H\,\ket{\psi_{\rm phys}} =\left(s\, \f{\hat{p}_t^2}{2}\otimes I_S+I_C\otimes\hat H_S\right)\,\ket{\psi_{\rm phys}}=0.\label{WdW2}
\end{align}
We assume here that this constraint has zero-eigenvalues, i.e. that solutions to Eq.~\eqref{WdW2} exist. Note that this requires the spectrum of $s\,\hat H_S$ to contain non-positive eigenvalues, in analogy with the classical case.

Quantizing $C_\sigma$ in Eq.~\eqref{factorize} 
yields ${[\hat C_+,\hat C_-]=0}$, so that the group averaging projector in Eq.~\eqref{projector} can be expressed as 
\begin{align}
\delta(\hat C_H) = \delta(s\,\hat C_+\,\hat C_-) ={\f{1}{2(-\,s\, \hat{H}_S)^{\frac{1}{2}}} }\sum_\sigma\,\delta(\hat C_\sigma).
\label{physicalProjector}
\end{align}
The form of $\delta(\hat C_H)$  implies the decomposition of the physical Hilbert space into a direct sum of positive and negative frequency sectors $\mathcal{H}_{\rm phys}\simeq \mathcal{H}_+ \oplus \mathcal{H}_-$  (see also \cite{Hartle:1997dc,Hoehn:2018whn}). Acting with the projector $\delta(\hat{C}_H)$ on an arbitrary kinematical state yields a physical state 
\begin{align}
\ket{\psi_{\rm phys}} &= \delta(\hat{C}_H)\,|\psi_{\rm kin}\rangle \nn \\
&=\sum_\sigma\,{\ \,{\intsum}_{E\in\sigma_{SC}}} \,\f{\psi_\sigma(E)}{{(2|E|)^{1/4}}}\,\ket{p_{t,\sigma}(E)\,}_C
\ket{E}_S, \label{pt22}
\end{align}
where  $\psi_\sigma(E)$  are Newton-Wigner-type wave functions associated to the positive and negative frequency modes~\cite{haag2012local}:\footnote{The fourth root comes about because the Newton-Wigner wave function is usually defined for Klein-Gordon systems where what we call $E$  is in fact the square of the energy $\omega_p:=\sqrt{\vec{p}^2+m^2}$. Note also that for the Klein-Gordon case one has a doubly degenerate system energy, which  we are not considering here. In that case, it is more convenient to use a momentum, rather than an energy representation of physical states, and a distinct measure. Eq.~\eqref{NWwave} can then indeed be interpreted as the usual Newton-Wigner wave function, written in terms of kinematical states. This will be discussed in more detail in  Sec.~\ref{sec_kuchar} (see also \cite{Hoehn:2018whn}).}
\begin{align}
\psi_\sigma(E) := \f{\psi_{\rm kin}\left(p_{t,\sigma}(E),E\right)}{{(2|E|)^{1/4}}}, \label{NWwave}
\end{align}
and we have defined the function $p_{t,\sigma}(E)\ce-\sigma\sqrt{2|E|}$ and spectrum 
\begin{align}
\sigma_{SC}&\ce\spec(\hat H_S)\cap \spec(-\hat H_C)\nn\\
&=\big\{E\in\spec(\hat H_S)\,\big|\, s\,E\leq0\big\}.\label{spectrum}
\end{align}

Physical states are normalized   with respect to the physical inner product introduced in Eq.~\eqref{PIP22}
\begin{align}
\braket{\psi_{\rm phys}|\phi_{\rm phys}}_{\rm phys}&\ce\langle\psi_{\rm kin}|\delta(\hat C_H)|\phi_{\rm kin}\rangle_{\rm kin} \label{PIP}\\
&=\sum_\sigma\, {\ \,{\intsum}_{E\in\sigma_{SC}}} \,\psi_\sigma^*(E)\,\phi_\sigma(E), \nn
\end{align}
 which takes the usual form of nonrelativistic quantum mechanics ($\sigma$-sector-wise), in line with the properties of Newton-Wigner-type wave functions. This observation will be crucial when discussing relativistic localization in Sec.~\ref{sec_kuchar}.

\section{Covariant clocks}\label{sec_covtime}

\subsection{Relational dynamics with a classical covariant clock}\label{sec_classrel}

Exploiting the splitting of the degrees of freedom into clock $C$ (our temporal reference system) and evolving system $S$, we now choose a \textit{clock function} $T$ on $\cp_C$  relative to which we describe the evolution of $S$ in terms of relational observables, as discussed in Sec.~\ref{sec_clneutral}. 

We could simply choose the phase space coordinate $T=t$ as the clock function. It follows from Eq.~\eqref{gaugeFlow} that in the $s=-1$ case $t$ runs `forward' on the positive frequency sector $\cc_+$ and `backward' on the negative frequency sector $\cc_-$ along the flow generated by $C_H$; for $s=+1$ the converse holds. Note that every point in $\cc_+\cap\,\cc_-$ corresponds to a static orbit of $t$ (since $p_t=0$ there), and $t$ is therefore a maximally bad clock function on $\cc_+\cap\,\cc_-$. This leads to challenges in describing relational dynamics relative to $t$: inverse powers of $p_t$ appear in the construction of relational observables encoding the evolution of system degrees of freedom relative to $t$ when canonical pairs on $\cp_S$ are used \cite{hoehnHowSwitchRelational2018,Hoehn:2018whn, Gambini:2000ht, Dittrich:2006ee}.\footnote{These challenges are related to those facing the definition of time-of-arrival operators in quantum mechanics \cite{Grot:1996xu,aharonov1998measurement, muga2000arrival, Gambini:2000ht,Dittrich:2006ee,hoehnHowSwitchRelational2018}.} One can solve this problem and obtain a well-defined relational dynamics by using affine (rather than canonical) pairs of evolving phase space coordinates on $\cp_S$ in the construction of relational observables \cite{Hoehn:2018whn}, or in the quantum theory by carefully regularizing inverse powers of $p_t$~\cite{hoehnHowSwitchRelational2018}.

However, in this article we shall sidestep these challenges and provide an arguably more elegant solution. We choose a different clock function according to the classical covariance condition: that it be canonically conjugate to $H_C$. This has the consequence of incorporating the pathology at $p_t=0$ into the clock function (which will nevertheless be meaningfully quantized in Sec.~\ref{sec_covtimedeg}), and leads to relational observables which work independently of the choice of phase space coordinates on $\cp_S$.  Solving $\{T,H_C\}=1$, we find that a covariant clock function $T$ must be of the form
$
T = s\,t/{p_t} + g(p_t)\,,
$
where $g(p_t)$ is an arbitrary function. Henceforth, we choose $g(p_t) = 0$  for simplicity, so that we have
\ba
T = s\,\f{t}{p_t}\,.\label{T}
\ea
This clock function is well-defined everywhere, except on the line $p_t=0$, where $H_C$ is non-degenerate. It is clear that $T$ runs `forward' everywhere on $\cc$ for both $s=\pm1$, except on $\cc_+\cap\,\cc_-$.

The covariance condition, combined with our assumption that the clock does not interact with the system, implies that $\{T,C_H\}=1$, which simplifies  the form of the relational Dirac observables in Eq.~\eqref{RelationalDiracObservable1}. 
 For example, the relational observable corresponding to the question `what is the value of the system observable $f_S$ when the clock $T$ reads $\tau$?' now takes the simple form \cite{dittrichPartialCompleteObservables2007,Hoehn:2019owq}
\begin{align}\label{F2}
F_{f_{S},T}(\tau)\approx\sum_{n=0}^{\infty}\,\f{(\tau-T)^n}{n!}\,\{f_S,H_S\}_n.
\end{align}

\subsection{Covariant quantum time observable for quadratic Hamiltonians}\label{sec_covtimedeg}

One might try to construct a time operator in  the quantum theory by directly quantizing the  covariant clock function in Eq.~\eqref{T} on the clock Hilbert space $\ch_C \simeq L^2(\mathbb{R})$ \cite{aharonov1961time,peres1980measurement,holevoProbabilisticStatisticalAspects1982}. Choosing a symmetric ordering, this yields
\begin{align}
 \label{quantTsym}
\hat T &=s\, \f{1}{2}\left(\hat t\,\hat{p}_t^{-1} + \hat{p}_t^{-1}\,\hat t\right).  
\end{align}
Here, $\hat{p}_t^{-1}$ is defined in terms of a spectral decomposition such that $\hat T \ket{p_t=0}$ is undefined, analogous to the classical case. While the operator $\hat T $ is canonically conjugate to the clock Hamiltonian, $[\hat T,\hat H_C] = i$, it is a symmetric operator that does not admit a self-adjoint extension~\mbox{\cite{holevoProbabilisticStatisticalAspects1982,busch1994time}}. Since $\hat T$ is not self-adjoint, its status as an observable is unclear.\footnote{Using the  commutation relation $[\hat t, \hat{p}_t^{-1}]= -i \hat{p}_t^{-2}$, which follows from multiplying $[\hat t,\hat p_t]=i$ from both sides with $\hat{p}_t^{-1}$, we can also write this operator as
\ba
\hat T&=&s\, \hat{p}_t^{-1}\left(\hat t - \f{i}{2}\,\hat{p}_t^{-1}\right)\label{quantT} 
\ea
We note in passing that the operator $\hat t - \f{i}{2}\,\hat{p}_t^{-1}$  is precisely the ``complex time operator" derived in \cite{Bojowald:2010xp} (see also \cite{Bojowald:2010qw,Hohn:2011us,Gielen:2020abd}) when constructing a relational Schr\"odinger picture for Wheeler-DeWitt type equations for constraints of the form Eq.~\eqref{WdW2}.
}
 This is a manifestation of Pauli's objection against the construction of time observables in quantum mechanics: For $\hat{H}_C$ bounded below, there does not exist a self-adjoint operator satisfying $[\hat T,\hat H_C] = i$. Pauli's conclusion was that we are forced to treat time as a classical parameter, different to the way other observables (e.g.\ position and momentum) are treated~\cite{pauli1958allgemeinen}. 

However, it was later realized that by appealing to the more general notion of an observable offered by a POVM, a \emph{covariant} time observable\footnote{We emphasize that the covariant time observable is a \emph{kinematical}, not a Dirac observable, as by construction its moments will not commute with the constraint.} $E_T$ can be constructed whose first moment corresponds to the operator $\hat{T}$ \cite{holevoProbabilisticStatisticalAspects1982,buschOperationalQuantumPhysics,braunsteinGeneralizedUncertaintyRelations1996}. Such a time observable  is defined by a set of effect operator densities ${E_T(dt)\geq 0}$ normalized as $ \int_{\mathbb{R}}  \, E_T(dt) = I_C$, and the covariance condition is implemented  by demanding that the effect operators $E_T(X) \ce \int_X  \, E_T(dt)$ for $X \subset \mathbb{R}$ are connected to one another by 
\begin{align}
E_T(X+t) = U_C(t) E_T(X) U_C^\dagger(t),\label{covariance}
\end{align}
where $U_C(t) \ce e^{-i\hat{H}_C t}$ is the unitary action of the one-dimensional group generated by the clock Hamiltonian. This will give rise to a generalization of canonical conjugacy of the time observable and the clock Hamiltonian, and permit us to extend the approach to relational quantum dynamics based on covariant clock POVMs~\cite{Hoehn:2019owq}  to relativistic models. In particular, we obtain a valid quantum time observable despite the classical clock pathologies.

Such an observable  can be constructed purely from the self-adjoint quantization of the clock Hamiltonian $\hat H_C$ and its eigenstates. The effect densities can be defined as a sum of `projections'
\ba
E_T(dt)=\f{1}{2\pi}\,\sum_{\sigma}\,dt\,\ket{t,\sigma}\!\bra{t,\sigma}
\ea
onto the \emph{clock states} corresponding to  the clock reading $t {\in \mathbb{R}}$ in the negative and positive frequency (i.e.\ positive and negative clock momentum) sector\footnote{Compared to \cite{braunsteinGeneralizedUncertaintyRelations1996}, we use a different definition of the degeneracy label $\sigma$ (here adapted to positive and negative frequency modes), change the normalization slightly, fix the relative phase, keep the momentum eigenstates as energy eigenstates and introduce $s$. For notational simplicity, we also set an arbitrary function in \cite{braunsteinGeneralizedUncertaintyRelations1996} (accounting for  a freedom in choosing the clock states) to zero. This is the quantum analog of the classical choice we made above, where we also set $g(p_t)$ in $T=t/p_t+g(p_t)$ to zero (see also {Appendix~B of} \cite{Hoehn:2019owq}). It would, however, be straightforward to reinsert this $g(p_t)$ in each of the following expressions.}
\begin{align}
\ket{t,\sigma}\ce \int_\mathbb{R}\,dp_t\,\sqrt{|p_t|}\,\theta(-\sigma\,p_t)\,e^{-i\,t\,s\,p_t^2/2}\,\ket{p_t}\,.
\end{align}\label{degclock} 
The covariance condition in Eq.~\eqref{covariance} is ensured by the fact that the clock states transform as
\begin{align}
\ket{t+t',\sigma} = U_C(t)\,\ket{t',\sigma}.\label{niceevol}
\end{align}
Note that the clock states  are orthogonal to the pathological state $\ket{p_t=0}$, and that they are not mutually orthogonal:
\begin{align}
\braket{t',\sigma'|t,\sigma}
= \delta_{\sigma\sigma'}\left[ \pi \delta \left( t - t' \right) - i \mbox{P} \frac{1}{t-t'} \right], \label{nonorth}
\end{align}
where P denotes the Cauchy principal value. Hence $E_T(dt)$ is not a true projector.  Nevertheless, the following lemma demonstrates that the clock states {$\ket{t,\sigma}$}  form an over-complete basis for {the $\sigma$-frequency sector of} $\mathcal{H}_C$, and in turn a properly normalized {covariant} time observable $E_T$ on $\mathcal{H}_{\rm kin}$.
\begin{Lemma}\label{lem_degresid}
The clock states $\ket{t,\sigma}$ defined in Eq.~(\ref{degclock}) integrate to projectors {$\theta(-\sigma\,\hat{p}_t)$} onto the positive/negative frequency sector on $\ch_C$
\ba
\f{1}{2\pi}\,\int_\mathbb{R}\,dt\,\ket{t,\sigma}\!\bra{t,\sigma} = \theta(-\sigma\,\hat{p}_t)\,
\ea
and hence form a resolution of the identity as follows:
\begin{align}
\int_{\mathbb{R}} E(dt)=\frac{1}{2\pi} \sum_{\sigma }\int_\mathbb{R} dt \,\ket{t,\sigma}\!\bra{t,\sigma}  = I_C\,.\label{degresid}
\end{align}
\end{Lemma}

\begin{proof}
The proof is given in Appendix~\ref{app_deg}.
\end{proof}

The $n^\text{th}$-moment operator of the time observable $E_T$  is defined as
\begin{align}
\hat T^{(n)} \ce {\int_\mathbb{R} E_T(dt) \,t^n =} \f{1}{2\pi} \sum_{\sigma }\int_\mathbb{R} dt \,t^n \ket{t,\sigma}\!\bra{t,\sigma}. \label{degnthmom}
\end{align}
With this definition, we find that the first-moment operator  $\hat T^{(1)}$ of $E_T$ is in fact equal to the operator $\hat T$ in Eq.~\eqref{quantTsym}.
 {This was previously noticed in \cite{holevoProbabilisticStatisticalAspects1982,busch1994time} (for the $s=+1$ case). This provides a concrete interpretation of the time observable $E_T$ in terms of the classical theory --- the time operator $\hat{T}^{(1)}$, namely the first moment of the time observable $E_T$, is the quantization of the classical clock function $T$ in Eq.~\eqref{T}.} 
\begin{Lemma} \label{lem_sameT}
{The  operator $\hat{T}$ and the first moment operator $\hat T^{(1)}$ of the covariant time observable $E_T$ are equal, $\hat T \equiv \hat T^{(1)}$.}
\end{Lemma}
\begin{proof}
The proof is given in Appendix~\ref{app_deg}.
\end{proof}

{Eq.~\eqref{degnthmom} demonstrates that the time operator $\hat T$ automatically splits into a positive and negative frequency part, in contrast to $\hat t$, the quantization of the phase space coordinate $t$.}

{Next, we find that while the clock states are not orthogonal, they are `almost' eigenstates of the covariant time operator $\hat T$ on each $\sigma$-sector:} 
\begin{Lemma}\label{lem_almosteigen}
The clock states $\ket{t,\sigma}$ defined in Eq.~(\ref{degclock}) are \emph{not} eigenstates of $\hat T = \hat{T}^{(1)}$. However, for all ${\ket{\psi}\in\cd(\hat T)}$, where $\cd(\hat T)$ is the domain of $\hat T$, they satisfy:
\ba
\bra{\psi}\,\hat T\,\ket{t,\sigma} = t\,\braket{\psi|t,\sigma}\,,\q\forall\,t\in\mathbb{R}\,, \, \sigma =\pm1\,.\nn
\ea
\end{Lemma}

\begin{proof}
The proof is given in Appendix~\ref{app_deg}.
\end{proof}

{This leads to a another result, which underscores why the covariance condition Eq.~\eqref{covariance} can be regarded as yielding a generalization of canonical conjugacy:}

\begin{Lemma}\label{thm_sameT2}
The $n^\text{th}$-moment operator defined in Eq.~\eqref{degnthmom} satisfies $[\hat T^{(n)},\hat H_C] = i\,n\,\hat T^{(n-1)}$. Furthermore, $\forall\ket{\psi} \in\mathcal{D}(\hat{T}^{n})$ we have $\hat T^{(n)}\ket{\psi}= \hat T^n\ket{\psi}$.
\end{Lemma}

\begin{proof}
The proof is given in Appendix~\ref{app_deg}.
\end{proof}

{We emphasize that the second statement of Lemma~\ref{thm_sameT2} does not hold on all of $\ch_C$.}

The effect {density does}  not commute with the clock Hamiltonian, $[E_T(dt), \hat H_C]\neq0$, which implies the time indicated by the clock (i.e.\  a measurement outcome of  $E_T$) and the clock energy cannot be  determined simultaneously.  However, importantly, the following lemma shows that the clock reading and the frequency sector, {i.e.\ the value of $\sigma$} \emph{can} be simultaneously determined. 
\begin{Lemma}\label{lem_projTcommute}
The effect density $E_T(dt)$ of the covariant clock POVM and the projectors onto the $\sigma$-sectors commute: $[E_T(dt),\theta(-\sigma\,\hat p_t)]=0$.
\end{Lemma}
\begin{proof}
The proof is given in Appendix~\ref{app_deg}.
\end{proof}
\begin{corol}\label{cor_5}
Since the effect density integrates to the effect and moment operators, this entails that  $[E_T(X),\theta(-\sigma\,\hat p_t)]=[\hat T^{(n)},\theta(-\sigma\,\hat p_t)]=0$, for all $X\subset\mathbb{R}$ and $n\in\mathbb{N}$.
\end{corol}

The significance of this lemma and  corollary is that they permit us to condition on the time indicated by the clock and the frequency sector simultaneously. This will become crucial when defining the quantum reduction maps below that take us from the physical Hilbert space to the relational Schr\"odinger and Heisenberg pictures which exist for each $\sigma$-sector. This lemma is thus an important for extending the quantum reduction procedures of~\cite{Hoehn:2019owq} to the class of models considered here.

\section{The trinity of relational quantum dynamics: Quadratic clock Hamiltonians}\label{sec_trinity}

Having introduced  the clock-neutral structure of the classical and quantum theories in Sec.~\ref{sec_cRDOs2}, a natural partitioning of the kinematical {degrees of freedom} into a clock $C$ and system $S$ in Sec.~\ref{Quadratic clock Hamiltonians}, {and a covariant time observable $E_T$ in Sec.~\ref{sec_covtime},} we are now able to {construct a \emph{relational quantum dynamics}, describing} how $S$ evolves relative to $C$.

As noted in the introduction, we showed in~\cite{Hoehn:2019owq} that three formulations of relation quantum dynamics, namely (i)~quantum relational Dirac observables, (ii) the relational Schr\"odinger picture of the Page-Wootters formalism, and (iii) the relational Heisenberg picture obtained through quantum deparametrization, are equivalent for models described by the Hamiltonian constraint in Eq.~\eqref{firstConstraint} when the clock Hamiltonian has a continuous, non-degenerate spectrum; the three formulations form a \emph{trinity} of relational quantum dynamics. Here we demonstrate that this equivalence extends to constraints of the form in Eq.~\eqref{WdW2}, involving the doubly degenerate clock Hamiltonian.\footnote{We also refer the reader to the recent work \cite{chataignier2020relational}, which we became aware of while completing this manuscript. It extends some of the results of \cite{Hoehn:2019owq} as well, though using the different formalism developed in \cite{Chataignier:2019kof}. It also does not employ covariant clocks in the case of quadratic clock Hamiltonians.}

Thanks to the direct sum structure of the physical Hilbert space ${\ch_{\rm phys}=\ch_{+}\oplus \ch_{-}}$ and the separation of the clock moment operators, Eq.~\eqref{degnthmom}, into non-degenerate positive and negative frequency sectors, all the technical results needed for establishing the equivalence in \cite{Hoehn:2019owq} will hold per $\sigma$-sector for the present class of models. We will thus  state some of the following results without proofs, {referring the reader as approriate to} the proofs of the corresponding results in \cite{Hoehn:2019owq}, {which} apply here per $\sigma$-sector. In particular, Corollary~\ref{cor_5} implies that we are permitted to simultaneously condition on the clock reading and the frequency sector.

  \begin{table}
 \begin{tabular}{c}
 Summary of different Hilbert spaces
  \\ \ \vspace{-5pt} \\ 
  \hline\hline \vspace{-5pt}
  \\
  Clock $C$ and system $S$ Hilbert spaces \\
  $\mathcal{H}_C$ and $\mathcal{H}_S$
  \\   \ \vspace{-5pt} \\ 
  \hline \vspace{-5pt}
 \\
  Kinematical Hilbert space\\
  $\mathcal{H}_{\rm kin} \simeq \mathcal{H}_C \otimes \mathcal{H}_S$
  \\ \ \vspace{-5pt} \\ 
 \hline \vspace{-5pt} 
  \\
  Physical Hilbert space \\
  $\mathcal{H}_{\rm phys} \simeq \delta(\hat{C}_H) (\mathcal{H}_{\rm kin}) = \mathcal{H}_+ \oplus \mathcal{H}_-$
  \\ \ \vspace{-5pt} \\ 
  \hline \vspace{-5pt} 
  \\
  Physical system Hilbert space\\
  $
  \mathcal{H}^{\rm phys}_S =\Pi_{\sigma_{SC}}(\ch_S) \subseteq\mathcal{H}_S
  $
  \\ \ \vspace{-5pt} \\ 
  \hline \vspace{-5pt} 
\\
\end{tabular}
\caption{The various Hilbert spaces appearing in the construction of the trinity. The physical system Hilbert space is the subspace of $\ch_S$ spanned by the energy eigenstates permitted upon solving the constraint. The $\sigma$-sector $\ch_\sigma$ of $\ch_{\rm phys}$ is also defined through solutions to the constraint $\hat C_\sigma$. \label{tab_table}} 
\end{table}

Lastly, we also provide a discussion of the relational quantum dynamics obtained through reduced phase space quantization. In this case, one deparametrizes the model classically relative to the clock function $T$, which amounts to a classical symmetry reduction. While the relational quantum dynamics thus obtained yields a relational Heisenberg picture resembling dynamics (iii) of the trinity, it is not always equivalent and thus not necessarily part of the trinity. For this reason, we have moved the exposition of reduced phase space quantization to Appendix~\ref{app_redQT}. It is however useful for understanding why the quantum symmetry reduction explained below is the quantum analog of classical phase space reduction through deparametrization. We emphasize that symmetry reduction and quantization do not commute in general~\cite{Ashtekar:1982wv,Kuchar:1986jj,Schleich:1990gd,Romano:1989zb,Loll:1990rx,Kunstatter:1991ds,Hoehn:2019owq}.

To aid the reader, we summarize the various Hilbert spaces appearing in the construction of the trinity in Table~\ref{tab_table}.

\subsection{The three faces of the trinity}

\subsubsection{Dynamics (i): Quantum relational Dirac observables}

We now {quantize the relational Dirac observables in Eq.~\eqref{F2}}, substantiating {the discussion of relational quantum dynamics in the clock-neutral picture in Sec.~\ref{sec_cnqt} for Hamiltonian constraints of the form Eq.~\eqref{WdW2}. Quantization of relational Dirac observables has been studied when the quantization of the classical time function $T$ results in a self-adjoint time operator $\hat{T}$ (see \cite{rovelliQuantumGravity2004,Rovelli:1990jm, Rovelli:1989jn,Rovelli:1990ph,Rovelli:1990pi,thiemannModernCanonicalQuantum2008,Tambornino:2011vg,Ashtekar:1993wb,Dittrich:2015vfa,Dittrich:2016hvj,Gambini:2000ht, Marolf:1994nz,Marolf:1994wh,Giddings:2005id, hoehnHowSwitchRelational2018,Hoehn:2019owq,Hoehn:2018whn,Chataignier:2019kof,chataignier2020relational} and references therein); however, when $\hat{T}$ fails to be self-adjoint, such as in Eq.~\eqref{quantTsym}, a more general quantization procedure is needed.}

Such a procedure was introduced in \cite{Hoehn:2019owq} based upon the quantization of Eq.~\eqref{F2} using covariant time observable{s}. 
Applying this procedure  to the present class of models described by quadratic clock Hamiltonians, we quantize the relational Dirac observables in Eq.~\eqref{F2} using the $n^\text{th}$-moment operators defined in Eq.~\eqref{degnthmom}:
\begin{align}
\hat F_{f_S,T}(\tau)&\ce {\int_{\mathbb{R}} \, E_T(dt) \otimes\sum_{n=0}^{\infty}{ \f{i^n}{n!}\,(t-\tau)^n\,\big[\hat f_S,\hat H_S\big]_n}}
\nn \\
&=\sum_{\sigma}\int_\mathbb{R} \f{dt}{2\pi}  \,U_{CS}(t)\!\left(\ket{\tau,\sigma}\! \bra{\tau,\sigma}\otimes \hat f_S\right)\!U_{CS}^\dag(t)\nn\\
&=:\sum_{\sigma}\,\cg \left(\ket{\tau,\sigma}\!\bra{\tau,\sigma}\otimes \hat f_S\right),\label{RDOq}
\end{align}
where $[\hat f_S,\hat H_S]_n \ce [ [\hat f_S,\hat H_S]_{n-1},\hat H_S]$ is the $n^\text{th}$-order nested commutator with the convention ${[\hat f_S,\hat H_S]_0 \ce  \hat f_S}$, $U_{CS}(t) \ce \exp(-i\,t\,\hat C_H) $, and the second line follows upon a change of integration variable and invoking the covariance condition in Eq.~\eqref{niceevol}. The relational Dirac observable $\hat F_{f_S,T}(\tau)$ is thus revealed to be an incoherent average over the {one-parameter noncompact} gauge group {$G$} generated by the constraint operator $\hat{C}_H$ of the kinematical operator $\ket{\tau,\sigma}\!\bra{\tau,\sigma}\otimes \hat f_S$, which is the system observable of interest $\hat f_S$ paired with the projector onto the clock reading $\tau$ and the $\sigma$-frequency sector. Such a group averaging is known as the $G$-twirl operation {and we denote it $\cg$ as} in the last line of Eq.~\eqref{RDOq}.  {$G$-twirl operations have previously been mostly  studied in the context of spatial quantum reference frames, e.g.\ see \cite{Bartlett:2007zz,smithCommunicatingSharedReference2019,smithQuantumReferenceFrames2016}, but have  also appeared in some constructions of quantum Dirac observables, e.g.\ see~\cite{Giulini:1998kf,thiemannModernCanonicalQuantum2008,Hoehn:2019owq,Chataignier:2019kof,chataignier2020relational}.\footnote{{The recent \cite{Chataignier:2019kof,chataignier2020relational} also develop a systematic quantization procedure for relational Dirac observables, based on integral techniques rather than the sum techniques used here and in \cite{Hoehn:2019owq}, and which too yields an expression similar to the one in the second line of Eq.~\eqref{RDOq}. While the construction procedure in \cite{Chataignier:2019kof,chataignier2020relational} encompasses a more general class of models (but implicitly assumes globally monotonic clocks too), it uses a more restrictive choice of clock observables which, in contrast to the covariant clock POVMs here and in \cite{Hoehn:2019owq}, are required to be self-adjoint. However, the two quantization procedures of relational observables are compatible and it will be fruitful to combine them.}}}
 As discussed in  \cite{Hoehn:2019owq}, this $G$-twirl constitutes the quantum analog of a gauge-invariant extension of a gauge-fixed quantity.

The relational Dirac observables
$\hat F_{f_S,T}(\tau)$ in Eq.~\eqref{RDOq} constitute  a one-parameter family of strong Dirac observables on $\ch_{\rm phys}$ (Theorem 1 of \cite{Hoehn:2019owq} whose proof applies here in each $\sigma$-sector):
\ba
[\hat F_{f_S,T}(\tau),\hat C_H]=0,\quad \forall\,\tau\in\mathbb{R}\,.\label{RDOcommute}
\ea
{We thus obtain a gauge-invariant relational quantum dynamics by letting the evolution parameter $\tau$ in the physical expectation values $\braket{\psi_{\rm phys}|\hat F_{f_S,T}(\tau)\,|\psi_{\rm phys}}_{\rm phys}$ run.}

The decomposition of $\hat F_{f_S,T}(\tau)$ in Eq.~\eqref{RDOq} into positive and negative frequency sectors gives rise to a reducible representation of 
the Dirac observable algebra on the physical Hilbert space. More precisely, relational Dirac observables are superselected across the $\sigma$-frequency sectors, and the $\sigma$-sum in Eq.~\eqref{RDOq} should thus be understood as a direct sum. To see this, consider the operator $\hat Q\ce  \theta({-}\hat{p}_t) - \theta(\hat{p}_t)$, where we recall that $\theta(-\sigma\,\hat{p}_t)$ is a projector onto the corresponding $\sigma$-sector.  By construction $[\hat Q,\hat C_H]=0$, which means that $\hat Q$ is a {strong} Dirac observable. {Its} eigenspaces, {with eigenvalues $+1$ and $-1$}, correspond to {the positive and negative frequency sector subspaces} $ \mathcal{H}_+$ and $\mathcal{H}_-$. Furthermore, $\hat Q$ commutes with any relational Dirac observable $F_{f_S,T}(\tau)$ in Eq.~\eqref{RDOq} on account of Lemma~\ref{lem_projTcommute}, which implies that  $Q$ and {any self-adjoint} $F_{f_S,T}(\tau)$ can be diagonalized in the same eigenbasis. This in turn implies the following superselection rule 
\begin{align}
\hat F_{f_S,T}(\tau)&= \hat F_{f_S,T}^+(\tau) \oplus \hat F_{f_S,T}^-(\tau),
\label{directSum}
\end{align}
where $\hat F_{f_S,T}^\sigma(\tau) \ce \mathcal{G} \big(\ket{\tau,\sigma}\!\bra{\tau,\sigma}\otimes \hat f_S) \in \mathcal{L}(\mathcal{H}_{\sigma} \big)$.\footnote{{In particular, when the spectrum of $\hat H_S$ does not contain zero, the $G$-twirl $\mathcal{G}$ can on each $\sigma$-sector be weakly rewritten as a reduced $G$-twirl $\mathcal{G}_\sigma$, i.e.\ one generated by $\hat C_\sigma$, rather than $\hat C_H$. Indeed, it is easy to see that the observables in Eq.~\eqref{RDOq} satisfy
\ba
\hat F^\sigma_{f_S,T}(\tau) = \mathcal{G} \big(\ket{\tau,\sigma}\!\bra{\tau,\sigma}\otimes \hat f_S)\approx \delta(\hat C_H)(\ket{\tau,\sigma}\bra{\tau,\sigma}\otimes\hat f_S),\nn
\ea
where $\approx$ is the quantum weak equality introduced in Eq.~\eqref{qweak}. Now use Eq.~\eqref{physicalProjector} and notice that $\delta(\hat C_{-\sigma})\,\ket{\tau,\sigma}\otimes\ket{E}_S = 0$ when zero does not lie in the spectrum of $\hat H_S$. This observation yields
\ba
\hat F^\sigma_{f_S,T}(\tau) &\approx& \f{1}{2(-s\hat H_S)^{\f{1}{2}}}\delta(\hat C_\sigma)(\ket{\tau,\sigma}\bra{\tau,\sigma}\otimes\hat f_S),\nn\\
&\approx& \f{1}{2(-s\hat H_S)^{\f{1}{2}}}\,\mathcal{G}_\sigma(\ket{\tau,\sigma}\bra{\tau,\sigma}\otimes\hat f_S)\,,\nn
\ea
where the last weak equality is restricted to $\ch_\sigma$.
When zero does lie in the spectrum of $\hat H_S$, the decomposition Eq.~\eqref{physicalProjector} is not well-defined and one needs to regularize.}} 

 While there do exist states in the physical Hilbert space that exhibit coherence across the $\sigma$-frequency sectors, for example $\ket{\psi_{\rm phys}} \sim  \ket{\psi_+} + \ket{\psi_-}  $, where $\ket{\psi_\sigma}\in\ch_\sigma$, such coherence is not physically accessible  because it  does not affect the expectation value of any relational Dirac observable on account of the decomposition in Eq.~\eqref{directSum}. In other words, superpositions and classical mixtures across the $\sigma$-frequency sectors are indistinguishable. {Hence, superpositions of physical states across $\sigma$-sectors are mixed states and the pure physical states are those of either $\ch_+$ or $\ch_-$ (see also \cite{marolfRefinedAlgebraicQuantization1995,Giulini:1998rk} for a discussion on superselection in group averaging). }

For example, this superselection rule manifests as a superselection across positive and negative frequency modes {in the case of the relativistic particle and across expanding and contracting solutions in the case of the FLRW model with a massless scalar field in Table~\ref{Table:Examples} \cite{Hoehn:2018whn}. On account of the reducibility of the representation, one usually restricts to one frequency sector (e.g.\ see \cite{Ashtekar:2011ni,Ashtekar:2007em,Banerjee:2011qu,Ashtekar:1993wb}).}  One might conjecture that the analogous superselection rule in a quantum field theory would manifest as a superselection rule between matter and anti-matter sectors.

Superselection rules induced by the $G$-twirl are often interpreted as arising from the lack of knowledge about a reference frame, and that if an appropriate reference frame is used the superselection rule can be lifted~\cite{Bartlett:2007zz}. {This interpretation seems unsuitable here. Firstly, lifting the superselection rule would entail undoing the group averaging, in violation of gauge invariance. Secondly, such an interpretation is usually tied to an average over a given group action which parametrizes one's ignorance about relative reference frame orientations. By contrast, the origin of the superselection of Dirac observables here is not the group generated by the constraint, but is a consequence of a property of the constraint, i.e.\ the group generator. Indeed, the superselection rule above originates in the factorizability of the constraint and the ensuing decomposition of the projector onto the constraint, Eq.~\eqref{physicalProjector}. Both these properties rely on the absence of a $\hat t$-dependent term in the constraint Eq.~\eqref{WdW2}; if such a term is introduced, one generally finds $[\hat C_+,\hat C_-]\neq0$, { where the $\hat C_\sigma$ are the quantization of the classical factors, but we emphasize that $\hat{C}_H \neq s\,\hat{C}_+ \hat{C}_-$ in that case}. While such a modified constraint may generate the same group,\footnote{{E.g.\ when $C_H=p_t^2-H^2(q_i,p_i,t)$ and $H^2>0$ $\forall$ $t\in\mathbb{R}$.}} no superselection rule across the $\sigma$-sectors would arise. 
 The above superselection rule can thus not be associated with the lack of a shared physical reference frame. This resonates with the interpretation of the physical Hilbert space as a clock-neutral, i.e.\ temporal-reference-frame-neutral structure (see Sec.~\ref{sec_cnqt}).}

Consider now the projector  $\Pi_{\sigma_{SC}}=\theta(-s\,\hat H_S)$ from $\ch_S$ to its subspace spanned by all system energy eigenstates $\ket{E}_S$ with $E\in\sigma_{SC}$; that is, those permitted upon solving the constraint Eq.~\eqref{WdW2}. {We shall henceforth denote this system Hilbert subspace $\ch_S^{\rm phys}:=\Pi_{\sigma_{SC}}(\ch_S)$ and call it the \emph{physical system Hilbert space}. We will obtain two copies of the physical system Hilbert space, one for each frequency sector.}  In analogy to the classical case, we introduce the \emph{quantum weak equality} between operators, signifying that they are equal on the `quantum constraint surface' $\ch_{\rm phys}$:
\ba
\hat O_1&\approx&\hat O_2\label{qweak}\\
&&\Leftrightarrow\, \hat O_1\,\ket{\psi_{\rm phys}} = \hat O_2\,\ket{\psi_{\rm phys}}\,,\q\forall\,\ket{\psi_{\rm phys}}\in\ch_{\rm phys}\,.\nn
\ea
It follows from Lemma 1 of \cite{Hoehn:2019owq}, whose proof applies here per $\sigma$-sector, that
\ba
\hat F_{f_S,T}(\tau)\approx\hat F_{\Pi_{\sigma_{SC}} {f}_S\, \Pi_{\sigma_{SC}},T}(\tau)\,,\label{weakequiv}
\ea
are weakly equal relational Dirac observables. Hence, the relational Dirac observables in Eq.~\eqref{RDOq} form weak equivalence classes on $\ch_{\rm phys}$, where $\hat F_{f_S,T}(\tau)\sim\hat F_{g_S,T}(\tau)$  if   $\Pi_{\sigma_{SC}} \hat{f}_S\, \Pi_{\sigma_{SC}}=\Pi_{\sigma_{SC}} \hat{g}_S\, \Pi_{\sigma_{SC}}$. These weak equivalence classes are labeled by what we shall denote
\begin{align}
\hat f_S^{\rm phys} \ce \Pi_{\sigma_{SC}}\,\hat f_S\,\Pi_{\sigma_{SC}}\in\cl\left(\ch_S^{\rm phys}\right) , \label{PWobs}
\end{align}
for arbitrary $\hat f_S\in\cl\left(\ch_S\right)$, where $\cl$ denotes the set of linear operators. {For later use, we note that the algebras of the physical system observables $\hat f_S^{\rm phys}$ on $\ch_S^{\rm phys}$ and the $\hat F_{f_S^{\rm phys},T}(\tau)$ on $\ch_{\rm phys}$ are weakly homomorphic with respect to addition, multiplication and commutator relations. More precisely,
\begin{align}
\hat F_{f^{\rm phys}_S+g^{\rm phys}_S\cdot h^{\rm phys}_S,T}(\tau) &\approx   \hat F_{f^{\rm phys}_S,T}(\tau) \nn\\
&\quad +\hat F_{g^{\rm phys}_S,T}(\tau)\cdot\hat F_{h^{\rm phys}_S,T}(\tau) \nn
\end{align}
{$\forall f_S , g_S , h_S\in\cl\left(\ch_S\right)$}. This is a consequence of Theorem~2 of \cite{Hoehn:2019owq} (whose proof again applies here per $\sigma$-sector). Together with \cite{Hoehn:2019owq}, this translates the weak classical algebra homomorphism defined through relational observables  in \cite{dittrichPartialCompleteObservables2007} into the quantum theory.} 

\subsubsection{Dynamics (ii): The Page-Wootters formalism}
 
{Suppose we are given a quantum Hamiltonian} constraint Eq.~\eqref{Wheeler-DeWitt} which splits into a clock and system contribution as in Eq.~\eqref{firstConstraint}, but for the moment not necessarily  assuming it to be of the quadratic form in Eq.~\eqref{WdW2}. Suppose further we are given some (kinematical)  time observable on the clock Hilbert space, which need not necessarily be a clock POVM which is covariant with respect to the group generated by the clock Hamiltonian, but is taken to define the clock reading. 
Page and Wootters~\cite{pageEvolutionEvolutionDynamics1983,woottersTimeReplacedQuantum1984,pageTimeInaccessibleObservable1989,pageClockTimeEntropy1994} proposed to extract a relational quantum dynamics between the clock and system from physical states in terms of conditional probabilities: what is the probability of an observable $\hat f_S$  associated with the system $S$ giving a particular outcome $f$, if the measurement of the {clock's} time observable yields the time $\tau$? 
If $e_C(\tau)$ and $e_{f_S}(f)$ are the projectors onto the clock reading~$\tau$ and the system observable $\hat f_S$  taking the value $f$, this conditional probability is postulated in the form
\begin{align}
\prob\left(f \ \mbox{when} \ \tau \right) &= \frac{\prob \left(f \ \mbox{and} \ \tau \right) }{\prob \left(\tau \right) } \label{fwhenT} \\
&=
\frac{\bra{\psi_{\rm phys} }e_C(\tau) \otimes e_{f_S}(f) \ket{\psi_{\rm phys} }_{\rm kin}}{\bra{\psi_{\rm phys} }e_C(\tau)\otimes I_S  \ket{\psi_{\rm phys} }_{\rm kin}}. \nn
\end{align}
This expression appears at first glimpse to be in violation of the constraints, as it acts with operators on physical states that are not Dirac observables; this is the basis of Kucha\v{r}'s criticism (b) that the conditional probabilities of the Page-Wootters formalism are incompatible with the constraints \cite{kucharTimeInterpretationsQuantum2011a}. However, for a class of models we have shown in \cite{Hoehn:2019owq} that the expression Eq.~\eqref{fwhenT} is a quantum analog of a gauge-fixed expression of a manifestly gauge-invariant quantity and thus consistent with the constraint. In this section we extend this result to relativistic settings.

Here we shall expand the Page-Wootters formalism to the more general class of Hamiltonian constraints of the form Eq.~\eqref{WdW2}  exploiting the covariant clock POVM $E_T$ of Sec.~\ref{sec_covtimedeg}.\footnote{See also the recent \cite{chataignier2020relational} for a complementary approach. Note that it does not employ covariant POVMs for quadratic Hamiltonians, and is thus subject to Kucha\v{r}'s criticism (a) described in Sec.~\ref{sec_intro}.} On the one hand, this will permit us to prove full equivalence of the so-obtained relational quantum dynamics with the manifestly gauge-invariant formulation in terms of relational Dirac observables on $\ch_{\rm phys}$ of Dynamics (i). As an aside, this will also resolve the normalization issue of physical states appearing in \cite{Diaz:2019xie}, where the kinematical rather than physical inner product was used to normalize physical states, thus yielding a divergence (when used for equal mass states). On the other hand, the covariant clock POVM will allow us, in Sec.~\ref{sec_kuchar}, to address the observation by Kucha\v{r} \cite{kucharTimeInterpretationsQuantum2011a} that using the Minkowski time observable leads to incorrect localization probabilities for relativistic particles in the Page-Wootters formalism.

The Page-Wootters formalism produces the system state at clock time $\tau$ by conditioning physical states on the clock reading $\tau$  \cite{pageEvolutionEvolutionDynamics1983,woottersTimeReplacedQuantum1984,pageTimeInaccessibleObservable1989,pageClockTimeEntropy1994}. {Henceforth focusing on the class of models defined by the constraint in Eq.~\eqref{WdW2} and the covariant clock POVM of Sec.~\ref{sec_covtimedeg}}, and given the reducible representation of $\ch_{\rm phys}$, we may additionally condition on the frequency sector thanks to Lemma~\ref{lem_projTcommute}. In extension of  \cite{Hoehn:2019owq}, we may use this conditioning to define two reduction maps ${\calr_{\rm PW}^\sigma(\tau):\ch_{\rm phys}\rightarrow\ch_{S,\sigma}^{\rm phys}}$, one per $\sigma$-frequency sector,
\ba
\calr_{\rm PW}^\sigma(\tau):=\bra{\tau,\sigma}\otimes I_S\, ,\label{PWred}
\ea  
where $\ch_{S,\sigma}^{\rm phys}$ is a copy of $\ch_S^{\rm phys}=\Pi_{\sigma_{SC}}(\ch_S)$, i.e.\ the subspace of the system Hilbert space permitted upon solving the constraint, corresponding to the $\sigma$-frequency sector. Due to the decomposition $\ch_{\rm phys}=\ch_+\oplus\ch_-$, we equip the two copies $\ch_{S,\sigma}^{\rm phys}$ with the frequency label $\sigma$ in order to remind ourselves which reduced theory corresponds to which positive or negative frequency mode.

The reduced states (whose normalisation factor $1/\sqrt{2}$ will be explained later),
\ba
\f{1}{\sqrt{2}}\,\ket{\psi^\sigma_S(\tau)} &:=& \calr_{\rm PW}^\sigma(\tau)\,\ket{\psi_{\rm phys}}\label{degprojt}\\
&=& {\ \,{\intsum}_{E\in\sigma_{SC}}}\,\psi_\sigma(E)\,e^{-i\,\tau\,E}\,\ket{E}_S,\nn
\ea
where $\psi_\sigma(E)$ is the Newton-Wigner type wave function defined in Eq.~\eqref{NWwave}, satisfy the Schr\"odinger equation with respect to $\hat H_S$:
\ba
i\f{d}{d\tau}\,\ket{\psi_S^\sigma(\tau)} = \hat H_S\,\ket{\psi_S^\sigma(\tau)}.\label{degSchrod}
\ea
We interpret this as the dynamics of $S$ relative to the temporal reference frame $C$. In particular,  this Schr\"odinger equation looks the same for both the positive and negative frequency sectors because the time defined by the covariant clock POVM $E_T$ runs forward in both sectors. This is clear from Eq.~\eqref{niceevol} and is the quantum analog of the earlier classical observation that the clock function $T$ runs `forward' on both frequency sectors $\cc_+$ and $\cc_-$ (in contrast to $t$).\footnote{Note that we could also define linear combinations of clock states 
$
\ket{\tau}:= \sum_\sigma\,c_\sigma\,\ket{\tau,\sigma}\,.\nn
$
Clearly, then we would also find that 
\ba
\ket{\psi_S(\tau)}&:=& \braket{\tau|\psi_{\rm phys}}\nn\\
&=&\left( \sum_\sigma\,c^*_\sigma\,{\ \,{\intsum}_{E\in\sigma_{SC}}}\,\psi_\sigma(E)\,e^{-i\,\tau E}\right)\ket{E}_S\nn
\ea
satisfies the same Schr\"odinger equation (\ref{degSchrod}). However, it is straightforward to check, using Lemma \ref{lem_degresid}, that these new clock states do not give rise to a resolution of the identity, $\int\,d\tau\,\ket{\tau}\bra{\tau}\neq I_C$ and so $\ket{\psi_{\rm phys}}\neq \int\,d\tau\,\ket{\tau}\ket{\psi_S(\tau)}$. In fact, if $c_\sigma\neq0$ for $\sigma=+,-$, then $\ket{\psi_S(\tau)}$ will mix contributions from the positive and negative frequency sectors such that it will become impossible to reconstruct (either of the positive or negative frequency part of) the physical state from it. That is, a reduction map $\bra{\tau}\otimes I_S$, which only conditions on the clock time, would not be invertible. This is another consequence of the superselection rule discussed above which entails that superpositions and mixtures across $\sigma$-sectors are indistinguishable through Dirac observables. It is also another reason why we condition also on the frequency sector when defining the reduction map in Eq.~\eqref{PWred}.}

Thanks to Eq.~(\ref{degresid}), the decomposition of the physical states into positive and negative frequency modes, Eq.~(\ref{pt22}), can also be written as follows:
\ba
\ket{\psi_{\rm phys}} =  \f{1}{2\sqrt{2}\,\pi}\,\sum_\sigma\,\int_\mathbb{R}\,dt\,\ket{t,\sigma}\ket{\psi^\sigma_S(t)}\,.
\ea
Together with Lemma~\ref{lem_degresid}, this implies that the $\sigma$-sector left inverse $\ch_{S,\sigma}^{\rm phys}\rightarrow \ch_\sigma$ of the reduction map defined in Eq.~\eqref{PWred} is given by 
\ba
(\calr_{\rm PW}^\sigma(\tau))^{-1} &=&\f{1}{2\pi} \int_\mathbb{R}\,dt\,\ket{t,\sigma}\otimes U_S(t-\tau)\nn\\
&=&\delta(\hat C_H)(\ket{\tau,\sigma}\otimes I_S)\,,\label{PWredinv}
\ea
where $U_S(t)=\exp(-i\,\hat H_S\,t)$,
so that
\ba
(\calr_{\rm PW}^\sigma(\tau))^{-1}\,\calr_{\rm PW}^\sigma(\tau) &=&\delta(\hat C_H)(\ket{\tau,\sigma}\bra{\tau,\sigma}\otimes I_S)\nn\\
& \approx&\theta(-\sigma\,\hat p_t)\otimes I_S\,,\label{PWsigmainv}
\ea
where $\approx$ is the quantum weak equality, and thus
\ba
\sum_\sigma\,(\calr_{\rm PW}^\sigma(\tau))^{-1}\,\calr_{\rm PW}^\sigma(\tau) \approx I_{\rm phys}\,.\label{nochmal}
\ea
Conversely, we can write the identity on $\ch_{S,\sigma}^{\rm phys}$ for solutions of the Schr\"odinger equation at time $\tau$ in the form
\ba
\calr_{\rm PW}^\sigma(\tau)\,(\calr_{\rm PW}^\sigma(\tau))^{-1}\,\ket{\psi_S^\sigma(\tau)}&=&\braket{\tau,\sigma|\,\delta(\hat C_H)\,|\tau,\sigma}\,\ket{\psi_S^\sigma(\tau)}\nn\\
&=&\ket{\psi_S^\sigma(\tau)}\,.\nn
\ea
{A summary of these maps can be found in Fig.~\ref{fig_PWmaps}.}
\begin{figure}[t]
\includegraphics[width= .45\textwidth]{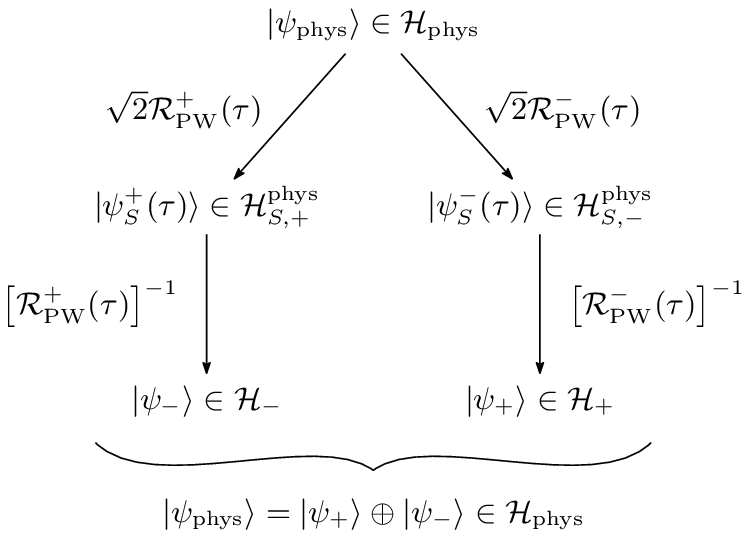}
\caption{
A summary of the Page-Wootters reduction maps and their inverses. {The analogous state of affairs holds for the quantum symmetry reduction maps and their inverses.}}
\label{fig_PWmaps}
\end{figure}

Using these reduction maps and their inverses, we can define an encoding operation ${\mathcal{E}_{\rm PW}^{\tau,\sigma}:\cl\left(\ch_{S,\sigma}^{\rm phys}\right)\rightarrow \cl\left(\ch_\sigma\right)}$, mapping the observables in Eq.~\eqref{PWobs}, acting on the physical system Hilbert space $\ch_{S,\sigma}^{\rm phys}$, into Dirac observables on the $\sigma$-sector of $\ch_{\rm phys}$: 
\ba
\mathcal{E}_{\rm PW}^{\tau,\sigma}\left(\hat f^{\rm phys}_S\right)&\ce& (\calr_{\rm \tiny PW}^\sigma(\tau))^{-1}\,\hat f_S^{\rm phys}\,\calr_{\rm PW}^\sigma(\tau)\nn\\
&=&\delta(\hat C_H)(\ket{\tau,\sigma}\!\bra{\tau,\sigma}\otimes \hat f^{\rm phys}_S)\,.\label{encode}
\ea

These encoded observables turn out to be the $\sigma$-sector part of the relational Dirac observables in Eq.~\eqref{RDOq}, as articulated in the following theorem.

\begin{theorem}\label{thm_PWobseq}
Let $\hat f_S\in\cl\left(\ch_S\right)$. The quantum relational Dirac observable $\hat F_{f_S,T}(\tau)$ acting on $\ch_{\rm phys}$, Eq.~\eqref{RDOq},  reduces under $\calr^\sigma_{\rm PW}(\tau)$ to the corresponding  projected observable  on $\ch_{S,\sigma}^{\rm phys}$, 
\begin{align}
\calr^\sigma_{\rm PW} \left(\tau\right)\,\hat F_{f_S,T}(\tau)\,(\calr^\sigma_{\rm PW}(\tau))^{-1} = \Pi_{\sigma_{SC}} \hat{f}_S\, \Pi_{\sigma_{SC}}\equiv\hat f_S^{\rm phys}. \nn 
\end{align}
Conversely, let $\hat f_S^{\rm phys}\in\cl\left(\ch_{S,\sigma}^{\rm phys}\right)$. The encoding operation in Eq.~\eqref{encode} of system observables coincides \emph{on the physical Hilbert space} $\ch_{\rm phys}$ with the quantum relational Dirac observables in Eq.~\eqref{RDOq} projected into the $\sigma$-sector:
\begin{align}
\mathcal{E}_{\rm PW}^{\tau,\sigma}\left(\hat{f}^{\rm phys}_S\right)\approx\hat F^\sigma_{f^{\rm phys}_S,T}(\tau)\,,\label{encode2}
\end{align}
where $\hat F^\sigma_{f^{\rm phys}_S,T}(\tau)=\hat F_{f^{\rm phys}_S,T}(\tau)\,(\theta(-\sigma\,\hat p_t)\otimes I_S)$ \,---\, c.f.\ Eq.~\eqref{directSum}.
\end{theorem}

\begin{proof}
The proof of Theorem 3 in~\cite{Hoehn:2019owq} applies here per $\sigma$-sector.
\end{proof}

Note that the relational Dirac observables $\hat F_{f_S,T}(\tau)$ commute with the projectors $\theta(-\sigma\,\hat p_t)$ due to the reducible representation in Eq.~\eqref{RDOq}.

Apart from providing the $\sigma$-sector-wise dictionary between the observables on the physical Hilbert space and the physical system Hilbert space, Theorem~\ref{thm_PWobseq}, in conjunction with the weak equivalence in Eq.~\eqref{weakequiv}, also implies an equivalence between the full sets of relational Dirac observables $\hat F^\sigma_{f^{\rm phys}_S,T}(\tau)$ on $\ch_\sigma$ and system observables on $\ch_{S,\sigma}^{\rm phys}$.

Crucially, the expectation values of the relational Dirac observables Eq.~\eqref{RDOq} in the physical inner product Eq.~\eqref{PIP} coincide, $\sigma$-sector-wise, with the expectation values of the physically permitted system observables $\hat f_S^{\rm phys}$ in the  states solving the Schr\"odinger equation Eq.~\eqref{degSchrod}  on $\ch_{S,\sigma}^{\rm phys}$. 

\begin{theorem}\label{thm_PWexpec}
Let $\hat f_S\in\cl\left(\ch_S\right)$, and denote its associated operator on $\ch_S^{\rm phys}$ by ${\hat f_S^{\rm phys} = \Pi_{\sigma_{SC}}\,\hat f_S\,\Pi_{\sigma_{SC}}}$. Then
\ba
\braket{\psi_{\rm phys}|\hat F^\sigma_{f_S,T}(\tau)\,|\phi_{\rm phys}}_{\rm phys} = \f{1}{2}\,\braket{\psi_S^\sigma(\tau)|\,\hat f^{\rm phys}_S\,|\phi_S^\sigma(\tau)}\,,\nn
\ea
where $\ket{\psi_S^\sigma(\tau)}=\sqrt{2}\,\calr^\sigma_{\rm PW}(\tau)\,\ket{\psi_{\rm phys}}$ as in Eq.~\eqref{degprojt}. Hence,
\ba
\braket{\psi_{\rm phys}|\hat F_{f_S,T}(\tau)\,|\phi_{\rm phys}}_{\rm phys} = \f{1}{2}\,\sum_\sigma\,\braket{\psi_S^\sigma(\tau)|\,\hat f^{\rm phys}_S\,|\phi_S^\sigma(\tau)}.\nn
\ea
\end{theorem}

\begin{proof}
The proof of Theorem 4 in~\cite{Hoehn:2019owq} applies here per $\sigma$-sector.
\end{proof}
Hence, the expectation values in the relational Schr\"odinger picture (i.e.\ the Page-Wootters formalism) are equivalent to the the gauge-invariant ones of the corresponding relational Dirac observables on $\ch_{\rm phys}$. Accordingly, equations such as Eq.~\eqref{fwhenT} (adapted to $\sigma$-sectors) are not in violation of the constraint as claimed by Kucha\v{r}~\cite{kucharTimeInterpretationsQuantum2011a}.

This immediately implies that the reduction maps $\calr^\sigma_{\rm PW}(\tau)$ preserve inner products per $\sigma$-sector as follows.

\begin{corol}\label{corol_2}
Setting $\hat f_S=\Pi_{\sigma_{SC}}$ in Theorem~\ref{thm_PWexpec} yields
\ba
\braket{\psi_{\rm phys}|\,\theta(-\sigma\,\hat p_t)\otimes I_S\,|\phi_{\rm phys}}_{\rm phys}= \f{1}{2}\,\braket{\psi_S^\sigma(\tau)|\phi_S^\sigma(\tau)}\,,\nn
\ea
where $\ket{\psi_S^\sigma(\tau)}=\sqrt{2}\,\calr^\sigma_{\rm PW}(\tau)\,\ket{\psi_{\rm phys}}$. Hence,
\ba
\braket{\psi_{\rm phys}|\phi_{\rm phys}}_{\rm phys} &=& \f{1}{2}\,\sum_\sigma\,\braket{\psi_S^\sigma(\tau)|\,\phi_S^\sigma(\tau)}\label{PWip}\\
&=&\sum_\sigma\, \braket{\psi_{\rm phys}|(\ket{\tau,\sigma}\bra{\tau,\sigma}\otimes I_S)|\phi_{\rm phys}}_{\rm kin}\,.\nn
\ea
\end{corol}

\noindent The reason for introducing the normalization factor $1/\sqrt{2}$ in Eq.~\eqref{degprojt} is now clear: it permits us to work with normalized states $\braket{\psi_S^\sigma(\tau)|\psi_S^\sigma(\tau)}=1=\braket{\psi_{\rm phys}|\psi_{\rm phys}}_{\rm phys}$ in each reduced $\sigma$-sector and in the physical Hilbert space simultaneously. 

The above results show:
\begin{itemize}
\item[(1)] Applying the Page-Wootters reduction map $\calr^\sigma_{\rm PW}(\tau)$ to the physical Hilbert space $\ch_{\rm phys}$ yields a \emph{relational Schr\"odinger picture} with respect to the clock $C$ on the  physical system Hilbert space $\ch_{S,\sigma}^{\rm phys}$ corresponding to the $\sigma$-frequency sector.

\item[(2)] $\sigma$-sector wise, the relational quantum dynamics encoded in the relational Dirac observables on the physical Hilbert space is equivalent to the dynamics in the relational Schr\"odinger picture on the physical system Hilbert space of the Page-Wootters formalism.

\item[(3)] Given the invertibility of the reduction map, Theorem~\ref{thm_PWexpec} formally shows that if $\hat f^{\rm phys}_S$ is self-adjoint on $\ch_{S,\sigma}^{\rm phys}$, then so is $\hat F^\sigma_{f_S,T}(\tau)$ on $\ch_\sigma$.
\end{itemize}

We note that the expression in the second line of Eq.~\eqref{PWip} also defines an inner product on the space of solutions to the Wheeler-DeWitt-type constraint Eq.~\eqref{WdW2}, which is equivalent to the physical inner product in Eq.~\eqref{PIP} obtained through group averaging. These two inner products thus define the same physical Hilbert space $\ch_{\rm phys}$. The expression in the second line of Eq.~\eqref{PWip} is the adaptation of the Page-Wootters inner product  introduced in~\cite{Smith:2017pwx} to the reducible representation of the physical Hilbert space associated to Hamiltonian constraints with quadratic clock Hamiltonians.

\subsubsection{Dynamics (iii): Quantum deparametrization}\label{sssec_qdepar}

Classically, one can perform a symmetry reduction of the clock-neutral constraint surface by gauge-fixing the flow of the constraint. {In this case,} this yields  two copies of a gauge-fixed reduced phase space, one for each frequency sector, each equipped with a standard Hamiltonian dynamics, hence yielding a deparametrized theory (see Appendix~\ref{app_redQT}). In contrast to the classical constraint surface, the `quantum constraint surface' $\ch_{\rm phys}$ is automatically gauge-invariant since the exponentiation of the symmetry generator $\hat C_H$ acts trivially on all physical states and Dirac observables. Hence, there is no more gauge-fixing in the quantum theory after solving the constraint. Nevertheless, following \cite{hoehnHowSwitchRelational2018,Hoehn:2018whn,Hoehn:2019owq}, we now demonstrate the quantum analog of the classical symmetry reduction procedure for the class of models considered in this article. As such it is the quantum analog of deparametrization, which we henceforth refer to as \emph{quantum deparametrization}.
This quantum symmetry reduction maps the clock-neutral Dirac quantization to {a relational Heisenberg picture} relative to clock {observable $E_T$}, and involves the following two steps.
\begin{enumerate}
\item \emph{Constraint trivialization}: A transformation of the constraint such that it only acts on the chosen reference system (here clock $C$), fixing its degrees of freedom. {The classical analog is a canonical transformation which turns the constraint into a momentum variable and separates gauge-variant from gauge-invariant degrees of freedom \cite{Vanrietvelde:2018pgb,Hoehn:2019owq}.}
\item \emph{Projection onto classical gauge fixing conditions}: A `projection' which removes the now redundant reference frame degrees of freedom.\footnote{{We put projection in quotation marks because it is not a true projection when applied to the physical Hilbert space as it only removes redundancy in the description, i.e.\ degrees of freedom which are fixed through the constraint. No physical information is lost. It would however be a projection on $\ch_{\rm kin}$.}}
\end{enumerate}

We begin by defining the \emph{constraint trivialization map} $\ct_{T,\epsilon}:\ch_{\rm phys}\rightarrow \ct_{T,\epsilon}(\ch_{\rm phys})$ relative to the covariant time observable $E_T$. This map will transform the physical Hilbert space into a new Hilbert space, while preserving inner products and algebraic properties of observables
\begin{align}
\ct_{T,\epsilon}&\ce \sum_{n=0}^\infty\,\f{i^n}{n!}\,\hat T^{(n)}\otimes\left(\hat H_S+s\,\f{\epsilon^2}{2}\right)^n \nn \\
&= \frac{1}{2\pi}  \sum_\sigma\,\int_\mathbb{R}\,dt\,\ket{t,\sigma}\bra{t,\sigma}\otimes e^{i\,t\,(\hat H_S+s\,\epsilon^2/2)}\,.\label{TT} 
\end{align}
In analogy to \cite{hoehnHowSwitchRelational2018,Hoehn:2018whn,Hoehn:2019owq}, we introduce an arbitrary positive parameter $\epsilon>0$ so that the map becomes invertible.
Note that $s\,\epsilon^2/{2}\in\spec(\hat H_C)$.

\begin{Lemma}\label{lem_triv1}
The left inverse of the trivialization $\ct_{T,\epsilon}$ is given by
\begin{align}
\ct_{T,\epsilon}^{-1}= \frac{1}{2\pi}  \sum_\sigma\,\int_\mathbb{R}\,dt\,\ket{t,\sigma}\bra{t,\sigma}\otimes e^{-i\,t\,(\hat H_S+s\,\epsilon^2/2)}\,,\nn
\end{align}
so that, for any $\epsilon>0$, $\ct_{T,\epsilon}^{-1}:\ct_{T,\epsilon}(\ch_{\rm phys})\rightarrow\ch_{\rm phys}$ and
\ba
\ct_{T,\epsilon}^{-1}\circ\ct_{T,\epsilon}\approx I_{\rm phys}.\nn
\ea
\end{Lemma}

\begin{proof}
The proof of Lemma 2  in~\cite{Hoehn:2019owq} applies here $\sigma$-sector wise.
\end{proof}

The main property of the trivialization map is summarized in the following lemma.

\begin{Lemma}\label{lem_triv2}
The map $\ct_{T,\epsilon}$ trivializes the constraint in Eq.~\eqref{WdW2} to the clock degrees of freedom
\begin{align}
\ct_{T,\epsilon}\,\hat C_H\,\ct_{T,\epsilon}^{-1} \,\,\,\overset{*}{\approx}\,\,\, \f{s}{2}(\hat p_t^2 - \epsilon^2)\otimes I_S,\label{trivial3}
\end{align}
where $\overset{*}{\approx}$ is the quantum weak equality on the trivialized physical Hilbert space $\ct_{T,\epsilon}(\ch_{\rm phys})$,
and transforms physical states into a sum of two product states with a  fixed and redundant clock factor
\begin{align}
\ct_{T,\epsilon}\,\ket{\psi_{\rm phys}} &= \f{1}{\sqrt{\epsilon}}\,\sum_\sigma\,\ket{p_t = -\sigma\epsilon}_C\label{trivializedstate2}\\
&\q\q\q\otimes\,\,\,{\ \,{\intsum}_{E\in\sigma_{SC}}}\,\psi_\sigma(E)\,\ket{E}_S\,.\nn
\end{align}
\end{Lemma}
\begin{proof}
The proof of Lemma 2  in~\cite{Hoehn:2019owq} applies here $\sigma$-sector wise.
\end{proof}

Hence, per $\sigma$-frequency sector, the trivialized physical states are product states with respect to the tensor product decomposition of the kinematical Hilbert space. {Recalling the discussion of the superselection rule across $\sigma$-sectors, the physical state in Eq.~\eqref{trivializedstate2} is indistinguishable from a separable mixed state when it contains both positive and negative frequency modes.}  One can therefore also view the trivialization as a $\sigma$-sector-wise disentangling operation, given that physical states in Eq.~\eqref{pt22} appear to be entangled. However, we emphasize that this notion of entanglement is kinematical and not gauge-invariant (see \cite{Hoehn:2019owq} for a detailed  discussion of this and how the trivialization can also be used to clarify the role of entanglement in the Page-Wootters approach).

The clock factors in Eq.~\eqref{trivializedstate2} have become redundant, apart from distinguishing between the positive and negative frequency sector. Indeed, if we had $\epsilon=0$, then (disregarding the diverging prefactor) both the negative and positive frequency terms in Eq.~\eqref{trivializedstate2} would have a common redundant factor $\ket{p_t=0}_C$, so that one could no longer distinguish between them at the level of the eigenbases of $\hat p_t$ and $\hat H_S$ in which the states have been expanded. That illustrates why $\ct_{T,\epsilon=0}$ is {\it not} invertible when acting on physical states. Indeed, $\ct^{-1}_{T,\epsilon=0}$ is undefined on states of the form $\ket{p_t=0}_C\ket{\psi}_S$, since $\hat T^{(n)}$ is not defined on $\ket{p_t=0}_C$. This is similar to the construction of the trivialization maps in \cite{hoehnHowSwitchRelational2018,Hoehn:2018whn}, except that here the decomposition  into positive and negative frequency sectors proceeds somewhat differently. This concludes step 1.\ above.

In order to complete the reduction to the states of the relational Heisenberg picture, and thus also complete step 2.\ above, we `project' out the redundant clock factor of the trivialized states by projecting onto the classical gauge-fixing condition $T=\tau$ (see Appendix~\ref{sec_redQT2} for a discussion of the classical gauge-fixing).  That is, we now proceed as in the Page-Wootters reduction and condition states in the trivialized physical Hilbert space $\ct_{T,\epsilon}(\ch_{\rm phys})$ on the clock reading $\tau$, separating positive and negative frequency modes.
Altogether, the quantum symmetry reduction map takes the form
\ba
\calr^\sigma_{\rm QR}(\tau)&:=&{e^{-i\,\tau\,s\,\epsilon^2/2}}\,(\bra{\tau,\sigma}\otimes I_S) \,\ct_{T,\epsilon}.\nn
\ea
Using that
\ba
\braket{\tau,\sigma|p_t=-\sigma'\epsilon} =\delta_{\sigma\sigma'} \sqrt{\epsilon}\,e^{i\,\tau\,s\,\epsilon^2/2}\,,\label{tptover2}
\ea
which is another reason why $\epsilon>0$ is chosen, in Eq.~\eqref{trivializedstate2} 
one finds $\tau$-independent system states
\ba
\calr^\sigma_{\rm QR}(\tau)\,\ket{\psi_{\rm phys}} &=& {\ \,{\intsum}_{E\in\sigma_{SC}}}\,\psi_\sigma(E)\,\ket{E}_S\nn\\
&=:&\f{1}{\sqrt{2}}\, \ket{\psi_S^\sigma}\,,\label{redfromdirac}
\ea
as appropriate for a Heisenberg picture (compare with Eq.~\eqref{degprojt}). The factor $1/\sqrt{2}$ has again been introduced for normalization purposes.
Since the wave function
\ba
\psi_S^\sigma(E) \equiv\sqrt{2}\,\psi_\sigma(E)\,,\label{degident}
\ea
is square-integrable/summable, it is clear that $\ket{\psi_S^\sigma}$ is an element of the physical system Hilbert space $\ch^{\rm phys}_{S,\sigma}$, corresponding to the $\sigma$-sector. We therefore also have ${\calr^\sigma_{\rm QR}(\tau):\ch_{\rm phys}\rightarrow\ch_{S,\sigma}^{\rm phys}}$, just as in Page-Wootters reduction. 
Using Lemmas~\ref{lem_triv1} and~\ref{lem_triv2}, it is now also clear how to invert the quantum symmetry reduction\,---\,at least per $\sigma$-sector:
\ba
(\calr^\sigma_{\rm QR})^{-1}&:=&\ct_{T,\epsilon}^{-1}\,\left(\f{1}{\sqrt{\epsilon}}\,\ket{p_t=-\sigma\epsilon}_C\otimes I_S\right)\nn
\ea
defines a map $\ch_{S,\sigma}^{\rm phys}\rightarrow\ch_\sigma$, so that\footnote{This is understood as appending the new clock tensor factor $\ket{p_t=-\sigma\epsilon}_C$ to the reduced system state $\ket{\psi_S^\sigma}$ and then applying the inverse of the trivialization {(recall that the conditioning of physical states on clock readings is not a true projection and thus invertible, cf.\ previous footnote)}. Note that embedding the reduced system states back into the physical Hilbert space is \emph{a priori} highly ambiguous since the system state alone no longer carries any information about the clock state which had been projected out. However, here it is the physical interpretation of the reduced state as being the description of the system $S$ relative to the temporal reference system $C$ that singles out the embedding into the clock-neutral (i.e.\ temporal-reference-system-neutral) $\ch_{\rm phys}$. This physical interpretation is, of course, added information, but it is crucial. For a more detailed discussion of this topic, see \cite{Vanrietvelde:2018pgb}.} 
\ba
(\calr_{\rm QR}^\sigma)^{-1}\,\ket{\psi_S^{\sigma}} =\sqrt{2}\,\theta(-\sigma\,\hat p_t)\ket{\psi_{\rm phys}}\,.\label{idphys2}
\ea
Hence, from the physical system Hilbert space of the positive/negative frequency modes one can only recover the positive/negative frequency sector of the physical Hilbert space.  Note that the inverse map is independent of $\tau$ in contrast to the Page-Wottters case.

More precisely, the following holds.
{\begin{Lemma}\label{lem_3}
The quantum symmetry reduction map is weakly equal to the Page-Wootters reduction map and an (inverse) system time evolution
\ba
\calr^\sigma_{\rm QR}(\tau)\,&\approx& \bra{\tau,\sigma}\otimes U_S^\dag(\tau)\,,\nn\\
&=&(I_C\otimes U_S^\dag(\tau))\,\calr^\sigma_{\rm PW}(\tau).\nn
\ea
Similarly, the inverse of the quantum symmetry reduction is equal to a system time evolution and the inverse of the Page-Wootters reduction:
\ba
(\calr^\sigma_{\rm QR})^{-1}&=&\delta (\hat{C}_H ) \left( \ket{\tau,\sigma}\otimes U_S(\tau) \right)\nn\\
&=&(\calr^\sigma_{\rm PW}(\tau))^{-1}\,(I_C\otimes U_S(\tau)).\nn
\ea
Hence
\ba
(\calr_{\rm QR}^\sigma(\tau))^{-1}\,\calr_{\rm QR}^\sigma(\tau)  \approx\theta(-\sigma\,\hat p_t)\otimes I_S\,\nn
\ea
and
\ba
\calr_{\rm QR}^\sigma(\tau)\,(\calr_{\rm QR}^\sigma)^{-1}\,\ket{\psi_S^\sigma}=\ket{\psi_S^\sigma}\,.\nn
\ea
\end{Lemma}}
\begin{proof}
The proof of Lemma 3  in~\cite{Hoehn:2019owq} applies here $\sigma$-sector wise.
\end{proof}

Given the Heisenberg-type states in Eq.~\eqref{redfromdirac}, we may consider evolving Heisenberg observables on $\ch_{S,\sigma}^{\rm phys}$
\ba
\hat f_S^{\rm phys}(\tau)=U_S^\dag(\tau)\,\hat f_S^{\rm phys}\,U_S(\tau).\label{Heisenbergobs}
\ea
Indeed, the following theorem shows that these Heisenberg observables are equivalent to the relational Dirac observables on the $\sigma$-sector of the physical Hilbert space $\ch_\sigma$, thereby demonstrating that the quantum symmetry reduction map yields a relational Heisenberg picture.
To this end, we employ these reduction maps and their inverses to define another encoding operation ${\mathcal{E}_{\rm QR}^{\sigma,\tau'}:\cl\left(\ch_{S,\sigma}^{\rm phys}\right)\rightarrow \cl\left(\ch_\sigma\right)}$,
\ba
\mathcal{E}_{\rm QR}^{\sigma,\tau'}\left(\hat f^{\rm phys}_S(\tau)\right)&\ce& (\calr_{\rm QR}^\sigma)^{-1}\,\hat f_S^{\rm phys}(\tau)\,\calr_{\rm QR}^\sigma(\tau').\label{QRencode}
\ea
The choice of $\tau'$ turns out to be irrelevant.

\begin{theorem}\label{thm_obstrafo2}
Let $\hat f_S\in\cl\left(\ch_S\right)$. The quantum relational Dirac observable $\hat F_{f_S,T}(\tau)$ acting on $\ch_{\rm phys}$, Eq.~\eqref{RDOq},  reduces under $\calr^\sigma_{\rm QR}(\tau')$ to the corresponding  projected evolving observable in the Heisenberg picture on $\ch_{S,\sigma}^{\rm phys}$, Eq.~\eqref{Heisenbergobs}, for all $\tau'\in\mathbb{R}$, i.e.
\begin{align}
\calr^\sigma_{\rm QR} \left(\tau'\right)\,\hat F_{f_S,T}(\tau)\,(\calr^\sigma_{\rm QR})^{-1} &= \Pi_{\sigma_{SC}} \hat{f}_S(\tau)\, \Pi_{\sigma_{SC}}\nn\\
&\equiv\hat f_S^{\rm phys}(\tau). \nn 
\end{align}
Conversely, let $\hat f_S^{\rm phys}(\tau)\in\cl\left(\ch_{S,\sigma}^{\rm phys}\right)$ be any evolving Heisenberg observable. The encoding operation in Eq.~\eqref{QRencode} of system observables coincides \emph{on the physical Hilbert space} $\ch_{\rm phys}$ with the quantum relational Dirac observables in Eq.~\eqref{RDOq} projected into the $\sigma$-sector:
\begin{align}
\mathcal{E}_{\rm QR}^{\sigma,\tau'}\left(\hat{f}^{\rm phys}_S(\tau)\right)\approx\hat F^\sigma_{f^{\rm phys}_S,T}(\tau).\label{encode2}
\end{align}
\end{theorem}

\begin{proof}
The proof of Theorem 5 in~\cite{Hoehn:2019owq} applies here per $\sigma$-sector.
\end{proof}

Once more, the theorem establishes an equivalence between the full sets of relational Dirac observables relative to clock $E_T$ on $\ch_\sigma$ and evolving Heisenberg observables on the physical system Hilbert space of the $\sigma$-modes $\ch^{\rm phys}_{S,\sigma}$. Hence, one can recover the action of the relational Dirac observables only $\sigma$-sector wise from the Heisenberg observables.

Lemma~\ref{lem_3} and Theorem~\ref{thm_PWexpec} directly imply that we again have preservation of expectation values per $\sigma$-sector, as the following theorem shows.

\begin{theorem}\label{thm_QRexpec}
Let $\hat f_S\in\cl\left(\ch_S\right)$ and ${\hat f_S^{\rm phys}(\tau) =U_S^\dag(\tau)\, \Pi_{\sigma_{SC}}\,\hat f_S\,\Pi_{\sigma_{SC}}}U_S(\tau)$ be its associated evolving Heisenberg operator on $\ch_S^{\rm phys}$. Then
\ba
\braket{\psi_{\rm phys}|\hat F^\sigma_{f_S,T}(\tau)\,|\phi_{\rm phys}}_{\rm phys} = \f{1}{2}\,\braket{\psi_S^\sigma|\,\hat f^{\rm phys}_S(\tau)\,|\phi_S^\sigma}\,,\nn
\ea
where $\ket{\psi_S^\sigma}=\sqrt{2}\,\calr^\sigma_{\rm QR}(\tau')\,\ket{\psi_{\rm phys}}$ for all $\tau'\in\mathbb{R}$. 
\end{theorem}
\begin{proof}
The proof of Theorem 6 in~\cite{Hoehn:2019owq} applies here per $\sigma$-sector.
\end{proof}

Therefore, the quantum symmetry reduction map $\calr_{\rm QR}(\tau')$ is  an isometry, as we state in the following corollary.

\begin{corol}
Setting $\hat f_S=\Pi_{\sigma_{SC}}$ in Theorem~\ref{thm_QRexpec} yields
\ba
\braket{\psi_{\rm phys}|\,\theta(-\sigma\,\hat p_t)\otimes I_S\,|\phi_{\rm phys}}_{\rm phys}= \f{1}{2}\,\braket{\psi_S^\sigma|\phi_S^\sigma}\,,\nn
\ea
where $\ket{\psi_S^\sigma}=\sqrt{2}\,\calr^\sigma_{\rm QR}(\tau')\,\ket{\psi_{\rm phys}}$, $\forall\,\tau'\in\mathbb{R}$. Hence,
\ba
\braket{\psi_{\rm phys}|\phi_{\rm phys}}_{\rm phys} &=& \f{1}{2}\,\sum_\sigma\,\braket{\psi_S^\sigma|\,\phi_S^\sigma}\,.\nn
\ea
\end{corol}

\noindent Accordingly, we can work with normalized states in each reduced $\sigma$-sector and in the physical Hilbert space simultaneously. 

In conclusion:
\begin{itemize}
\item[(1)] Applying the quantum symmetry reduction map $\calr^\sigma_{\rm QR}(\tau)$ to  the clock-neutral picture on the physical Hilbert space $\ch_{\rm phys}$ yields a \emph{relational Heisenberg picture} with respect to the clock $C$ on the  physical system Hilbert space of the $\sigma$-modes, $\ch_{S,\sigma}^{\rm phys}$.

\item[(2)] $\sigma$-sector wise, the relational quantum dynamics encoded in the relational Dirac observables on the physical Hilbert space is equivalent to the dynamics in the relational Heisenberg picture on the physical system Hilbert space.

\item[(3)] Given the invertibility of the reduction map, Theorem~\ref{thm_QRexpec} formally shows that if $\hat f^{\rm phys}_S(\tau)$ is self-adjoint on $\ch_{S,\sigma}^{\rm phys}$, then so is $\hat F^\sigma_{f_S,T}(\tau)$ on $\ch_\sigma$.
\end{itemize}

\subsubsection{Equivalence of Dynamics (ii) and (iii)}

The previous subsections establish a $\sigma$-sector wise equivalence between the relational dynamics, on the one hand, in the clock-neutral picture of Dirac quantization and, on the other, the relational Schr\"odinger and Heisenberg pictures, obtained through Page-Wootters reduction and quantum deparametrization, respectively. It is thus evident that also the relational Schr\"odinger and Heisenberg pictures are indeed equivalent up to the unitary $U_S(\tau)$ as they should. This is directly implied by Lemma~\ref{lem_3} which shows that the Page-Wotters and quantum symmetry reduction maps are (weakly) related by $U_S(\tau)$.

\subsection{Quantum analogs of gauge-fixing and gauge-invariant extensions}

In contrast to the classical constraint surface, the `quantum constraint surface' $\ch_{\rm phys}$ is automatically gauge-invariant since the exponentiation of the symmetry generator $\hat C_H$ acts trivially on all physical states and Dirac observables. Nevertheless, extending the interpretation established in \cite{Hoehn:2019owq}, we can understand the quantum symmetry reduction map $\calr_{\rm QR}(\tau)$ (and given their unitary relation, also $\calr_{\rm PW}^\sigma(\tau)$) as the quantum analog of a classical phase space reduction through gauge-fixing. For completeness, the latter procedure is explained in Appendix~\ref{app_redQT} for the class of models of this article. In particular, we may think of the physical system Hilbert space $\ch_{S,\sigma}^{\rm phys}$ for the $\sigma$-sector as the quantum analog of the gauge-fixed reduced phase space obtained by imposing for example the gauge $T=0$ on the classical $\sigma$-frequency sector $\cc_\sigma$.\footnote{However, note that the quantization of this classical reduced phase space will in some cases, but not in general, coincide with the quantum theory on $\ch_{S,\sigma}^{\rm phys}$ due to the generic inequivalence between Dirac and reduced quantization (see Appendix~\ref{app_redQT}).} Also classically, one obtains two identically looking gauge-fixed reduced phase spaces, one for each frequency sector.
Consequently, the relational Schr\"odinger and Heisenberg pictures can both be understood as the quantum analog of a gauge-fixed formulation of a manifestly gauge-invariant theory. 

In this light, Theorems~\ref{thm_PWobseq} and~\ref{thm_obstrafo2} imply that the encoding operations of system observables in Eqs.~\eqref{encode} and~\eqref{QRencode} can be understood as the quantum analog of gauge-invariantly extending a gauge-fixed quantity {(see also \cite{Hoehn:2019owq})}. Similarly, the alternative physical inner product in the second line of Eq.~\eqref{PWip} is the quantum analog of a gauge-fixed version of the manifestly gauge-invariant physical inner product obtained through group averaging in Eq.~\eqref{PIP}. Indeed, $\sum_\sigma\, \braket{\psi_{\rm phys}|(\ket{\tau,\sigma}\!\bra{\tau,\sigma}\otimes I_S)|\phi_{\rm phys}}_{\rm kin}$ is the (kinematical) expectation value of the `projector' onto clock time $\tau$ in physical states. However, it is clear from the unitarity of the Schr\"odinger dynamics on $\ch_{S,\sigma}^{\rm phys}$ and Eq.~\eqref{PWip} that this inner product does not depend on $\tau$ (the `gauge'), in line with the  interpretation of it being the quantum analog of a gauge-fixed version of a manifestly gauge-invariant quantity.

\subsection{Interlude: alternative route}
\label{alternative route}

As an aside, we mention that there exists an alternative route to establishing a trinity for clock Hamiltonians quadratic in momenta. This again exploits the reducible representation on the physical Hilbert space. The $\sigma$-sector of $\ch_{\rm phys}$ is defined by the constraint ${\hat C_\sigma=\hat H'_C+\hat H_S'}$, where $\hat H'_C:=\f{\hat p_t}{\sqrt{2}}$ and $\hat H'_S:=\sigma\sqrt{-s\,\hat H_S}$. Clearly, $\hat H'_C$ is now a non-degenerate clock Hamiltonian. In \cite{Hoehn:2019owq} we established the trinity for non-degenerate clock Hamiltonians and the $\sigma$-sector defined by ${\hat C_\sigma=\hat H'_C+\hat H_S'}$ yields a special case of that. This immediately implies a trinity per $\sigma$-sector, however, now relative to a clock POVM which is covariant with respect to $\hat H'_C$. It is evident that the covariant clock POVM is in this case defined through the eigenstates of $\hat t$ which (up to a factor of $\sqrt{2}$) is also the first moment of the POVM. Indeed, the equivalence between the clock-neutral Dirac quantization and the relational Heisenberg picture  has previously been established for models with quadratic clock Hamiltonians precisely in this manner in \cite{hoehnHowSwitchRelational2018,Hoehn:2018whn} {(see also the recent \cite{chataignier2020relational})}. However, as mentioned in Sec.~\ref{sec_classrel}, one either has to regularize the relational observables or write them as functions of affine, rather than canonical pairs of evolving degrees of freedom. {This is a consequence of the square root nature of $\hat H'_S$.} None of these extra steps were needed in the trinity construction of this article, which is based on a clock POVM which is covariant with respect to $s\,\f{\hat p_t^2}{2}$, rather than $\hat p_t/\sqrt{2}$.


\section{Relativistic localization: addressing Kucha\v{r}'s criticism}\label{sec_kuchar}

In his seminal review on the problem of time, Kucha\v{r} raised three criticisms against the Page-Wootters formalism \cite{kucharTimeInterpretationsQuantum2011a}: the Page-Wootters conditional probability in Eq.~\eqref{fwhenT} (a) yields the wrong localization probabilities for a relativistic particle, (b) violates the Hamiltonian constraint, and (c) produces incorrect transition probabilities. As mentioned in the introduction, criticisms (b) and (c) have been resolved in~\cite{Hoehn:2019owq}\,---\,see Theorem~\ref{thm_PWexpec} which extends the resolution of (b) to the present class of models. 

Here, we shall now also address Kucha\v{r}'s first criticism (a) on relativistic localization, which is more subtle to resolve. The main reason, as is well-known from the theorems of Perez-Wilde \cite{FernandoPerez:1976ib} and Malament \cite{Malament1996} (see also the discussion in \cite{Yngvason:2014oia,Papageorgiou:2019ezr}), is that there is no relativistically covariant position-operator-based notion of localization which is compatible with relativistic causality and positivity of energy. This is a key motivation for quantum field theory \cite{haag2012local,Yngvason:2014oia} -- and here a challenge for specifying what the `right' localization probability for a relativistic particle should be. Instead, one may resort to an approximate and relativistically non-covariant notion of localization proposed by Newton and Wigner \cite{haag2012local,Newton:1949cq}. We will address criticism (a) by demonstrating that our formulation of the Page-Wootters formalism, based on covariant clocks for relativistic models, yields a localization in such an approximate sense.

For the sake of an explicit argument, we shall, just like Kucha\v{r} \cite{kucharTimeInterpretationsQuantum2011a}, focus solely on the free relativistic particle, whose Hamiltonian constraint reads (cf.\ Table~\ref{Table:Examples})
\ba
\hat C_H = -\hat p_t^2+\Hat{\bm{p}}{}^2+m^2\,,\nn
\ea
{where $\hat{\bm p}$ denotes the spatial momentum vector}.
However, the argument could be extended to the entire class of models considered in this manuscript. It is straightforward to check that the physical inner product Eq.~\eqref{PIP} reads in this case \cite{Hartle:1997dc,Hoehn:2018whn}\footnote{Note that here we have a doubly degenerate system energy $\hat H_S=\Hat{\bm{p}}{}^2+m^2$ in contrast to the expression in Eq.~\eqref{PIP} where we ignored degeneracy.}
\begin{align}
\braket{\phi_{\rm phys}|\psi_{\rm phys}}_{\rm phys}&=\int_{\mathbb{R}^3}\f{d^3\bm{p}}{2\,\varepsilon_p}\Big[\phi^*_{\rm kin}(\varepsilon_p,\bm{p})\,\psi_{\rm kin}(\varepsilon_p,\bm{p})\nn\\
&\quad {+} \phi^*_{\rm kin}(-\varepsilon_p,\bm{p})\,\psi_{\rm kin}(-\varepsilon_p,\bm{p})\Big],
\end{align}
where $\varepsilon_p=\sqrt{\bm{p}{}^2+m^2}$ is the relativistic energy of the particle and the first and second term in the integrand correspond to negative and positive frequency modes, respectively. Fourier transforming to solutions to the Klein-Gordon equation in Minkowski space
\ba
\psi_{\rm phys}^\sigma(t,\bm{x}) =\f{1}{(2\pi)^{3/2}}\,\int_{\mathbb{R}^3}\f{d^3\bm{p}}{2\,\varepsilon_p}\,e^{i(\bm{x}\cdot\bm{p}-\sigma\,t\,\varepsilon_p)}\,\psi_{\rm kin}(-\sigma\,\varepsilon_p,\bm{p})\,,\nn
\ea
one may further check that \cite{Hartle:1997dc,Hoehn:2018whn}
\ba
\braket{\phi_{\rm phys}|\psi_{\rm phys}}_{\rm phys}=\left(\phi^+_{\rm phys},\psi^+_{\rm phys}\right)_{\rm KG}-\left(\phi^-_{\rm phys},\psi^-_{\rm phys}\right)_{\rm KG}\,,\nn\\\label{KGa}
\ea
where
\begin{align}
\left(\phi^\sigma_{\rm phys},\psi^\sigma_{\rm phys}\right)_{\rm KG} &= i\int_{\mathbb{R}^3}d^3\bm{x}\,\Big[\left(\phi_{\rm phys}^\sigma(t,\bm{x})\right)^*\partial_t\,\psi_{\rm phys}^\sigma(t,\bm{x}) \nn \\
&\quad -\left(\partial_t\,\phi_{\rm phys}^\sigma(t,\bm{x})\right)^*\psi_{\rm phys}^\sigma(t,\bm{x})\Big]\,,\label{KG}
\end{align}
is the Klein-Gordon inner product in which positive  frequency modes are positive semi-definite, negative frequency modes are negative semi-definite and positive and negative frequency modes are mutually orthogonal. The physical inner product is thus equivalent to the Klein-Gordon inner product (with correctly inverted sign for the negative frequency modes), which provides the correct and conserved normalization for the free relativistic particle.\footnote{This also resolves the normalization issue appearing in \cite{Diaz:2019xie} where physical states were normalized using the kinematical, rather than physical inner product, thus yielding a divergent normalization (for equal mass states) in contrast to here.} This raises hopes that the conditional probabilities of the Page-Wootters formalism may yield the correct localization probability for the relativistic particle. Note that so far we have not yet made a choice of time operator.

Suppose now that the Minkowski time operator $\hat t$, quantized as a self-adjoint operator on $\ch_{\rm kin}$, is used to define the projector onto clock time $t$ as $e_C(t)=\ket{t}\!\bra{t}$ and $e_{\bm{x}}=\ket{\bm{x}}\!\bra{\bm{x}}$ is the projector onto position $\bm{x}$. This time operator is not covariant with respect to the quadratic clock Hamiltonian. The conditional probability Eq.~\eqref{fwhenT} then becomes
\begin{align}
\prob\left(\bm{x} \ \mbox{when} \ t \right) =
\frac{|\psi_{\rm phys}(t,\bm{x})|^2}{\int_{\mathbb{R}^3}\,d^3\bm{x'}\,|\psi_{\rm phys} (t,\bm{x'})|^2}, \label{xwhent}
\end{align}
where $\psi_{\rm phys}(t,\bm{x})=(\bra{t}\otimes\bra{\bm{x}})\ket{\psi_{\rm phys}}$ is a general solution to the Klein-Gordon equation. As Kucha\v{r} pointed out~\cite{kucharTimeInterpretationsQuantum2011a}, while this would be the correct localization probability for a non-relativistic particle, it is the wrong result for a relativistic particle.  Indeed, apart from not separating positive and negative frequency modes, which is necessary for a probabilistic interpretation {(e.g., if $\psi_{\rm phys}$ contains both positive and negative frequency modes then the denominator in Eq.~\eqref{xwhent} is not conserved)}, Eq.~\eqref{xwhent} neither coincides with the charge density of the Klein-Gordon current in Eq.~\eqref{KG}, nor with the Newton-Wigner approximate localization probability \cite{haag2012local,Newton:1949cq}. 
In particular, one can \emph{not} interpret a solution $\psi_{\rm phys}(t,\bm{x})$ to the Klein-Gordon equation as a probability amplitude to find the relativistic particle at position $\bm{x}$ at time $t$. The reason, as explained in  \cite{haag2012local}, is that the conserved density is the one in Eq.~\eqref{KG} and $\psi_{\rm phys}$ and $\partial_t\psi_{\rm phys}$ inside it are not only dependent, but related by a non-local convolution
\ba
\partial_t\psi_{\rm phys}(t,\bm{x}) = \int_{\mathbb{R}^3}d^3\bm{x'}\,\varepsilon(\bm{x}-\bm{x'})\,\psi_{\rm phys}(t,\bm{x'})\,,\nn
\ea
where $\varepsilon(\bm{x}-\bm{x'})$ is the Fourier transform of $-i\,\varepsilon_p$.

By contrast, let us now exhibit what form of conditional probabilities the covariant clock POVM $E_T$ of Sec.~\ref{sec_covtimedeg} gives rise to. We now insert $e_C(\tau)=\sum_\sigma\ket{\tau,\sigma}\!\bra{\tau,\sigma}$ and, as before, $e_{\bm{x}}$ into the conditional probability Eq.~\eqref{fwhenT}. {The crucial difference between the covariant clock POVM $E_T$ and the clock operator $\hat t$ (which is covariant with respect to $\hat C_\sigma$, but not $\hat C_H$) is that the denominator of Eq.~\eqref{fwhenT} is equal to the physical inner product in the former case (see Corollary~\ref{corol_2}) but not in the latter.\footnote{{It is instructive to see how this is linked to the {(non-)}covariance of the clock observable with respect to $\hat C_H$. Let $e_C$ be either the covariant $e_C(\tau)$ or non-covariant $e_C(t)$.  The denominator of Eq.~\eqref{fwhenT} reads
\ba
\,\,\,\bra{\psi_{\rm phys}}e_C\otimes I_S  \ket{\psi_{\rm phys}}_{\rm kin}&=& \bra{\psi_{\rm kin}}\delta(\hat C_H)\,(e_C\otimes I_S)\ket{\psi_{\rm phys} }_{\rm kin}.\nn
\ea
Eq.~\eqref{nochmal} implies that $\delta(\hat C_H)(e_C(\tau)\otimes I_S) \approx I_{\rm phys}$. This exploits the covariance and immediately shows that the denominator coincides with the physical inner product Eq.~\eqref{PIP}. By contrast,  the non-covariance entails $\delta(\hat C_H)(e_C(t)\otimes I_S) \neq I_{\rm phys}$, so that in this case the denominator differs from the physical inner product.}}}
Supposing that we work with normalized physical system states $\braket{\psi_S^\sigma(\tau)|\psi_S^\sigma(\tau)}=1$, Theorem~\ref{thm_PWexpec} implies
\begin{align}
\prob\left(\bm{x} \ \mbox{when} \ \tau \right) &=
\f{1}{2}\sum_\sigma |\psi_S^\sigma(\tau,\bm{x})|^2 \label{xwhent2}\\
&=\braket{\psi_{\rm phys}|\hat F_{e_{\bm{x}},T}(\tau)|\psi_{\rm phys}}_{\rm phys}\,,\nn
\end{align}
where $\psi_S^\sigma(\tau,\bm{x}):=\sqrt{2}(\bra{\tau,\sigma}\otimes\bra{\bm{x}})\ket{\psi_{\rm phys}}$ and $\tau$ is now \emph{not} Minkowski time.  For concreteness, let us now focus on positive frequency modes. Using Eqs.~\eqref{pt22} and~\eqref{degclock}, one obtains
\ba
\psi_S^+(\tau,\bm{x}) = \f{1}{(2\pi)^{3/2}}\int_{\mathbb{R}^3}\f{d^3\bm{p}}{\sqrt{2\,\varepsilon_p}}e^{i(\bm{x}\cdot\bm{p}-\tau\,\varepsilon_p^2)}\,\psi_{\rm kin}(-\varepsilon_p,\bm{p})\,.\nn\\
\label{NWtau}
\ea
This is \emph{almost} the Newton-Wigner position space wave function for positive frequency modes, which relative to Minkowski time reads \cite{haag2012local}
\ba
\psi_{\rm NW}^+(t,\bm{x}) &=&\int_{\mathbb{R}^3}d^3\bm{x}\, K(\bm{x}-\bm{x'})\,\psi^+_{\rm phys}(t,\bm{x'})\label{NW}\\
&=& \f{1}{(2\pi)^{3/2}}\int_{\mathbb{R}^3}\f{d^3\bm{p}}{\sqrt{2\,\varepsilon_p}}e^{i(\bm{x}\cdot\bm{p}-t\,\varepsilon_p)}\,\psi_{\rm kin}(-\varepsilon_p,\bm{p})\,,\nn
\ea
where $K(\bm{x})$ is the Fourier transform of $\sqrt{2\varepsilon_p}$. The key property of $|\psi_{\rm NW}^+(t,\bm{x})|^2$ is that, while  not relativistically covariant, it \emph{does} admit the interpretation of an approximate localization probability, with accuracy of the order of the Compton wave length, for finding the particle at position $\bm{x}$ at Minkowski time $t$ \cite{haag2012local,Newton:1949cq}. In particular, 
\ba
\left(\phi^+_{\rm phys},\psi^+_{\rm phys}\right)_{\rm KG} = \int_{\mathbb{R}^3}d^3\bm{x}\left(\phi_{\rm NW}^+(t,\bm{x})\right)^*\psi_{\rm NW}^+(t,\bm{x})\,,\nn
\ea
i.e.\ the Klein-Gordon inner product assumes the usual Schr\"odinger form for the Newton-Wigner wave function.

Noting that {due to Eq.~\eqref{T}} we can \emph{heuristically} view $\tau$ as $ t/\varepsilon_p$, and comparing with Eq.~\eqref{NW} we can interpret Eq.~\eqref{NWtau} as a Newton-Wigner wave function as well, but expressed relative to a different time coordinate $\tau$. Indeed, in line with this interpretation, we find that in this case too the physical inner product, Eq.~\eqref{KGa}, for the positive frequency modes assumes the form of the standard Schr\"odinger theory inner product
\ba
\left(\phi^+_{\rm phys},\psi^+_{\rm phys}\right)_{\rm KG} = \int_{\mathbb{R}^3}d^3\bm{x}\left(\phi_S^+(\tau,\bm{x})\right)^*\psi_{S}^+(\tau,\bm{x})\,.\nn
\ea
The analogous statement is true for the negative frequency modes. In that sense, Eq.~\eqref{xwhent2},  in contrast to Eq.~\eqref{xwhent}, does admit the interpretation as a valid, yet approximate localization probability for the relativistic particle per frequency sector, just like in the standard Newton-Wigner case.\footnote{The physical inner product for the positive frequency solutions $\psi^+_{\rm phys}(t,\bm{x})$ to the Klein-Gordon equation takes, of course, the standard Klein-Gordon (and not the Schr\"odinger) form Eq.~\eqref{KG}.
Nevertheless, the kernel $K(\bm{x}-\bm{x'})$ in the nonlocal convolution in the first line in Eq.~\eqref{NW}  decreases quickly as a function of $m|\bm{x}-\bm{x'}|$ \cite{haag2012local}. Hence, for massive particles, we can interpret even Eq.~\eqref{xwhent} as providing an approximate localization.} 

Accordingly, computing the conditional probabilities of the Page-Wootters formalism relative to the covariant clock POVM, rather than the non-covariant Minkowski time operator $\hat t$, leads to an acceptable localization probability for a relativistic particle, thereby addressing also Kucha\v{r}'s first criticism (a). 
  Given the equivalence of the Page-Wootters formalism with the clock-neutral and the relational Heisenberg pictures, established through the trinity in Sec.~\ref{sec_trinity}, this result also equips the quantum relational Dirac observables $\hat F_{{\bm{x}},T}(\tau)$ and the evolving Heisenberg observables ${\bm{x}}(\tau)$ with the interpretation of providing an approximate, Newton-Wigner type localization in Minkowski space.

 \section{Changing quantum clocks}\label{sec_cqc}
 
So far we have worked with a single choice of clock. Let us now showcase how to change from the evolution relative to one choice of clock to that relative to another.   Our discussion will apply to both the relational Schr\"odinger picture of the Page-Wootters formalism and the relational Heisenberg picture obtained through quantum deparametrization. 

For concreteness, suppose we are given a Hamiltonian constraint of the form
\ba
\hat C_H = s_1\f{\hat p_1^2}{2}+s_2\f{\hat p_2^2}{2}+  \hat H_S\,,\label{CH2}
\ea
where $s_i=\pm1$ and $\hat p_i$ denotes the momentum of clock subsystem $C_i$, $i=1,2$ and we have suppressed tensor products with identity operators. In particular, suppose $\hat H_S$ does not depend on either of the clock degrees of freedom. We will work with the covariant clock POVM of Sec.~\ref{sec_covtimedeg} for both clock choices. For example, the constraints of the relativistic particle, the flat ($k=0$) FLRW model with a massless scalar field and the Bianchi I and II models from Table~\ref{Table:Examples} are of the above form.\footnote{Indeed, the Hamiltonian constraint of the vacuum Bianchi I and II models can be written in the form \cite{Ashtekar:1993wb}
\ba
\hat C_H = -\f{\hat{\bar p}_0^2}{2}+\f{\hat{\bar p}_-^2}{2}+\f{\hat{\bar p}_+^2}{2}+k_+\,e^{-4\sqrt{3}\hat{\bar\beta}^+}\,.\nn
\ea}
Our subsequent discussion will thus directly apply to these models.
 
Since we will exploit the Page-Wootters and symmetry reduction maps as `quantum coordinate maps' from the clock-neutral picture to the given `clock perspective', we will be able to change from the description of the dynamics relative to one clock to that relative to another in close analogy to coordinate changes on a manifold. Due to the shape of the constraint in Eq.~\eqref{CH2} we now have superselection of Dirac observables and the physical Hilbert space across both the $\sigma_1$-frequency sectors of clock $C_1$ and the $\sigma_2$-frequency sectors of clock $C_2$. The physical Hilbert space takes the form
\ba
\ch_{\rm phys} = \bigoplus_{\sigma_1,\sigma_2}\ch_{\sigma_1,\sigma_2}\,,\label{overlap}
\ea
where $\ch_{\sigma_1,\sigma_2}\ce\ch_{\sigma_1}\cap\ch_{\sigma_2}$ is the overlap of the $\sigma_1$-sector of clock $C_1$ and the $\sigma_2$-sector of clock $C_2$.
As we have seen the reduction maps are only invertible per frequency sector. Hence, we will only be able to change from a given $\sigma_1$-sector to the part of the $\sigma_2$-frequency sector which is contained in it. In other words, the ``quantum coordinate changes'' are restricted to each overlap $\ch_{\sigma_1,\sigma_2}$.

The method of changing temporal reference frames exhibited below is a direct extension of several previous works: \cite{hoehnHowSwitchRelational2018,Hoehn:2018whn} developed the method $\sigma$-sector-wise for states and observables in the relational Heisenberg picture for Hamiltonians of the type in Eq.~\eqref{CH2} for two example models, but used clock operators canonically conjugate to the clock momenta $\hat p_i$ (and thus not a clock POVM covariant with respect to the full clock Hamiltonian). The method of transforming relational observables from one clock description to another was demonstrated in \cite{hoehnHowSwitchRelational2018,Hoehn:2018whn} for a subset of relational Dirac observables, paying, however, detailed attention to regularization necessities arising from time-of-arrival observables \cite{Grot:1996xu,aharonov1998measurement, muga2000arrival, Gambini:2000ht,Dittrich:2006ee}. Our previous article \cite{Hoehn:2019owq} developed the method comprehensively for both states and observables for clock Hamiltonians with non-degenerate and continuous spectrum in both the relational Schr\"odinger and Heisenberg pictures; specifically, the transformation of arbitrary relational observables corresponding to relations between $S$ and the clocks was developed for the corresponding class of models. In \cite{castro-ruizTimeReferenceFrames2019} the clock change method was exhibited for state transformations in the relational Schr\"odinger picture for ideal clocks whose Hamiltonian coincides with  the clock momentum itself. Our discussion can also be viewed as a full quantum extension of the semiclassical method in \cite{Bojowald:2010xp,Bojowald:2010qw,Hohn:2011us} which is equivalent at semiclassical order, however, also applies to clock functions which are non-monotonic, i.e.\ have turning points, in contrast to the other works mentioned. (See \cite{giacominiQuantumMechanicsCovariance2019,Vanrietvelde:2018dit,Vanrietvelde:2018pgb,hamette2020quantum} for related spatial quantum frame changes.)

In particular, owing to our focus on covariant clock POVMs, all the results and proofs  \cite{Hoehn:2019owq} apply $\sigma$-sector-wise to the present case.  However, we will also study novel effects such as the temporal frame dependence of comparing clock readings.

\subsection{State transformations}

Denote by $\calr^{\sigma_i}_I(\tau_i)$, where $I\in\{\rm{PW},\rm{QR}\}$, the Page-Wootters or quantum symmetry reduction map to the $\sigma_i$-sector of clock $C_i$. The temporal frame change (TFC) map $\Lambda_{I\to J}^{\sigma_i\to \sigma_j}: \ch^{\rm phys}_{C_j,\sigma_i}\otimes\ch^{\rm phys}_{S,\sigma_i}\rightarrow   \ch^{\rm phys}_{C_i,\sigma_j}\otimes\ch^{\rm phys}_{S,\sigma_j}$ from clock $C_i$'s $\sigma_i$-sector to clock $C_j$'s $\sigma_j$-sector then reads\footnote{For notational simplicity we write all inverse maps as functions of $\tau$. Recall, however, that $\calr_{\rm QR}^\sigma$ does not depend on $\tau$.}
\ba
\Lambda_{I\to J}^{\sigma_i\to \sigma_j} \ce \calr^{\sigma_j}_J(\tau_j) \circ \left(\calr^{\sigma_i}_I(\tau_i)\right)^{-1}\,.\label{TFC}
\ea
Here $\ch^{\rm phys}_{C_j,\sigma_i}$ denotes the physical clock $C_j$ Hilbert space corresponding to the $\sigma_i$-sector of clock $C_i$, i.e.\ the subspace of $\ch_{C_j}$ compatible with solutions to the constraint Eq.~\eqref{CH2} and similarly for the other Hilbert spaces. When $I\neq J$ in Eq.~\eqref{TFC}, then the TFC map changes not only the temporal reference frame, but also between the corresponding relational Heisenberg and Schr\"odinger pictures. Let us write $\Lambda_{I}^{\sigma_i\to \sigma_j} \ce\Lambda_{I\to I}^{\sigma_i\to \sigma_j}$ when no relational picture change takes place.

More explicitly, the TFC map from the $\sigma_1$-frequency sector of clock $C_1$ in the relational Schr\"odinger picture to the $\sigma_2$-frequency sector of clock $C_2$ in the relational Schr\"odinger picture takes the form
\ba
\Lambda_{\rm PW}^{\sigma_1\to \sigma_2} = \left(\bra{\tau_2,\sigma_2}\otimes I_{C_1S}\right)\,\delta(\hat C_H)\,\left(\ket{\tau_1,\sigma_1}\otimes I_{C_2S}\right)\,.\nn
\ea
Here we have made use of Eqs.~\eqref{PWred} and~\eqref{PWredinv} and the covariant clock states Eq.~\eqref{degclock} for both clocks. The reduced states transform under this map as follows:
\ba
\left(\theta(-\sigma_1\,\hat p_1)\otimes I_S\right)\,\ket{\psi^{\sigma_2}_{C_1S|C_2}(\tau_2)} = \Lambda_{\rm PW}^{\sigma_1\to \sigma_2}\,\ket{\psi^{\sigma_1}_{C_2S|C_1}(\tau_1)} \,,\nn\\
\label{PWtrafo}
\ea
where we made use of Eq.~\eqref{PWsigmainv} and $\ket{\psi^{\sigma_i}_{C_jS|C_i}(\tau_i)}$ is the relational Schr\"odinger picture state of clock $C_j$ and system $S$ relative to clock $C_i$, which is chosen as the temporal reference frame, in its $\sigma_i$-sector. In other words, the Heaviside-function on the l.h.s.\ highlights that we can only map from the $\sigma_1$-sector of clock $C_1$ to that part of the $\sigma_2$-sector of clock $C_2$ which is contained in the $\sigma_1$-sector of clock $C_1$. This is clear, because any reduction map is only invertible on its associated $\sigma$-sector: from the $\sigma_1$ relational Schr\"odinger picture one can only recover the $\sigma_1$-sector of the physical Hilbert space. Hence, the subsequent Page-Wootters reduction map to the $\sigma_2$-sector of clock $C_2$ in Eq.~\eqref{TFC} can then only yield information in the overlap of the $\sigma_1$- and $\sigma_2$-sectors in the physical Hilbert space (see also \cite{hoehnHowSwitchRelational2018,Hoehn:2018whn} for explicit examples of this situation in the relational Heisenberg picture). This is a manifestation of the superselection across both the $\sigma_1$- and the $\sigma_2$-sectors.

\begin{figure}[t]
\includegraphics[width= .45\textwidth]{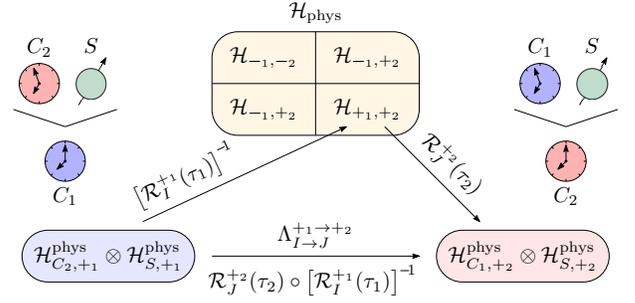}
\caption{Schematic representation of a temporal frame change, as defined through Eq.~\eqref{TFC}. The figure encompasses both the relational Schr\"odinger picture of the Page-Wotters formalism and the relational Heisenberg picture of the quantum deparametrization, as well as their mixtures, since $I,J\in\{\rm{PW},\rm{QR}\}$. Viewing the reduction maps $\calr_I^{+_i}(\tau_i)$ as quantum coordinate maps, any such temporal frame change takes the form of a quantum coordinate transformation from the description relative to clock $C_1$ to the one relative to clock $C_2$. Just as coordinate transformation  pass through the reference-frame-neutral manifold, the quantum coordinate transformations pass through the clock-neutral physical Hilbert space in line with the general discussion of the clock-neutral structure and quantum general covariance in Sec.~\ref{sec_cRDOs2}. Due to the double superselection rule, the quantum coordinate transformations have to preserve the overlaps of the frequency sectors of $C_1$ and $C_2$. Here we illustrate the example of the overlap of the positive frequency sectors of both clocks, so that the corresponding frame transformation passes through $\ch_{+_1,+_2}$ (cf.\ Eq.~\eqref{overlap}).  }
\label{fig_TFC}
\end{figure}

Similarly,  the TFC map from the $\sigma_1$-frequency sector of clock $C_1$ in the relational Heisenberg picture to the $\sigma_2$-frequency sector of clock $C_2$ in the relational Heisenberg picture reads
\ba
\Lambda_{\rm QR}^{\sigma_1\to \sigma_2} &=& \left(\bra{\tau_2,\sigma_2}\otimes U^\dag_{C_1S}(\tau_2)\right)\,\\
&&\q\q\q\times\delta(\hat C_H)\,\left(\ket{\tau_1,\sigma_1}\otimes U_{C_2S}(\tau_1)\right)\,,\nn
\ea
where we have made use of Lemma~\ref{lem_3}. Using the same lemma, the reduced states transform under this map in complete analogy to Eq.~\eqref{PWtrafo}
\ba
\Big[\theta(-\sigma_1\,\hat p_1)\otimes I_S\Big]\,\ket{\psi^{\sigma_2}_{C_1S|C_2}} = \Lambda_{\rm PW}^{\sigma_1\to \sigma_2}\,\ket{\psi^{\sigma_1}_{C_2S|C_1}} \,,\nn
\ea
with $U_{C_jS}(\tau_i):=\exp\left[-i\,\tau_i\,(s_j\f{\hat p_j^2}{2}+\hat H_S)\right]$ the evolution operator of the composite system of clock $C_j$ and system $S$ relative to clock $C_i$.

Note that, interpreting the  reduction maps as the ``quantum coordinate maps'' taking one from the clock-neutral physical Hilbert space to a specific `clock perspective', any such TFC map in Eq.~\eqref{TFC} takes the same compositional form as coordinate changes on a manifold. In particular, any such temporal frame change proceeds by mapping via the clock-neutral physical Hilbert space in analogy to how coordinate changes always proceed via the manifold, see Fig.~\ref{fig_TFC}. This observation lies at the heart of the perspective-neutral approach to quantum reference frame changes  \cite{Vanrietvelde:2018pgb,Vanrietvelde:2018dit,hoehnHowSwitchRelational2018,Hoehn:2018whn,Hoehn:2019owq}. It is also the reason why we may interpret the physical Hilbert space as a clock-neutral structure, providing a description of the dynamics prior to having chosen a temporal reference frame relative to which the other degrees of freedom evolve. In line with this, in terms of different one-parameter families of relational Dirac observables, the physical Hilbert space contains the complete information about the dynamics relative to all the different possible clock choices at once.

\subsection{Observable transformations}

Just as we transformed reduced states from the perspective of one clock to the perspective of another by passing through the gauge-invariant physical Hilbert space (see Fig.~\ref{fig_TFC}), we now  transform the description of observables relative to one clock to that relative to the other by passing through the gauge-invariant algebra of Dirac observables on the physical Hilbert space. The observable transformations will thus be dual to the state transformations. The idea is always that we describe the \emph{same} physics, encoded in the gauge-invariant states and observables of the clock-neutral physical Hilbert space, but relative to different temporal frames. Again, we have to pay attention to the two superselection rules on the clock-neutral physical Hilbert space and we will demonstrate the observable transformations separately for the relational Schr\"odinger and Heisenberg pictures.

\subsubsection{Observable transformations in the relational Schr\"odinger picture}

Suppose we are given an observable $\hat O^{\rm phys}_{C_2S|C_1}$ describing certain properties of the composite system $C_2S$ in the relational Schr\"odinger picture of clock $C_1$ in either frequency sector of the latter.\footnote{Recall that the label `$\rm{phys}$' highlights that the observable acts on the physical $C_2S$ Hilbert space, i.e.\ on $\Pi_{\sigma_{C_2SC_1}}\left(\ch_{C_2}\otimes\ch_S\right)$, where $\Pi_{\sigma_{C_2SC_1}} = \theta\left[-s_1\left(s_2\f{\hat p_2^2}{2}+\hat H_S\right)\right]$ is the projector onto the subspace corresponding to the spectrum $\sigma_{C_2SC_1} :=\spec\left(s_2\f{\hat p_2^2}{2}+\hat H_S\right)\cap\spec\left(-s_1\f{\hat p_1^2}{2}\right)$ permitted by the constraint Eq.~\eqref{CH2} (cf.\ Eq.~\eqref{spectrum}). We can thus also understand this observable as a projection $\hat O^{\rm phys}_{C_2S|C_1}\ce\Pi_{\sigma_{C_2SC_1}}\,\hat O_{C_2S|C_1}\,\Pi_{\sigma_{C_2SC_1}}$ of some observable $\hat O_{C_2S|C_1}\in\cl(\ch_{C_2}\otimes\ch_S)$; cf.\ Eq.~\eqref{PWobs}. } Owing to Theorem~\ref{thm_PWobseq}, we can write this as a reduction of a corresponding relational Dirac observable on $\ch_{\rm phys}$:
\begin{align}
\calr^{\sigma_1}_{\rm PW} \left(\tau_1\right)\,\hat F_{O_{C_2S|C_1},T_1}(\tau_1)\,(\calr^{\sigma_1}_{\rm PW}(\tau_1))^{-1} =\hat O^{\rm phys}_{C_2S|C_1}. \nn 
\end{align}
We can now also map the same relational Dirac observable into the $\sigma_2$-sector of the relational Schr\"odinger picture of clock $C_2$:
\ba
\hat O^{\rm phys}_{C_1S|C_2}(\tau_1,\tau_2)\ce\calr^{\sigma_2}_{\rm PW} \left(\tau_2\right)\,\hat F_{O_{C_2S|C_1},T_1}(\tau_1)\,(\calr^{\sigma_2}_{\rm PW}(\tau_2))^{-1}\,.\nn\\
\label{PWobsC2}
\ea
The result will be the image of the original observable $\hat O^{\rm phys}_{C_2S|C_1}$, describing properties of $C_2S$ relative to $C_1$, in the `perspective' of clock $C_2$. Hence, if $\hat O^{\rm phys}_{C_2S|C_1}$ depends non-trivially on $C_2$, an indirect self-reference effect occurs in the last equation \cite{Hoehn:2019owq}.
Notice that, while the original observable in the Schr\"odinger picture of $C_1$ is independent of the evolution parameter $\tau_1$, the description of that same observable in the Schr\"odinger picture relative to clock $C_1$ will generally depend on both evolution parameters $\tau_1,\tau_2$. The dependence on $\tau_1$ is a consequence of it being the reduction of a relational Dirac observable with evolution parameter $\tau_1$, but into the `perspective' of $C_2$. The possible $\tau_2$ dependence may arise as a consequence of said indirect self-reference. For example, suppose $\hat O^{\rm phys}_{C_2S|C_1} = \hat T_2\otimes I_S$ so that the relational Dirac observable is $\hat F_{O_{C_2S|C_1},T_1}(\tau_1) = \hat F_{T_2,T_1}(\tau_1)$. The observable on the l.h.s.\ in Eq.~\eqref{PWobsC2} then describes how the first moment operator $\hat T_2$ associated with  $C_2$ evolves relative to $C_1$ from the `perspective' of $C_2$; this certainly should yield a $\tau_1$ dependence. We will explain this in more detail shortly. 

Taking into account the two superselection rules across the $\sigma_1$- and $\sigma_2$-sectors, these observations imply that observable transformations from the relational Schr\"odinger picture of the $\sigma_1$-sector of clock $C_1$ into the relational Schr\"odinger picture of the $\sigma_2$-sector of clock $C_2$ read
\begin{widetext}
\ba
\Lambda^{\sigma_1\to \sigma_2}_{\rm PW} \, \hat O^{\rm phys}_{C_2S|C_1}\,\left(\Lambda^{\sigma_1\to\sigma_2}_{\rm PW}\right)^{-1}&=& \left(\Lambda^{\sigma_2\to \sigma_1}_{\rm PW}\right)^{-1} \, \hat O^{\rm phys}_{C_2S|C_1}\,\Lambda^{\sigma_2\to\sigma_1}_{\rm PW}\nn\\
&=&\calr^{\sigma_2}_{\rm PW}(\tau_2)\circ\left(\calr^{\sigma_1}_{\rm PW}(\tau_1)\right)^{-1}\, \hat O^{\rm phys}_{C_2S|C_1}\,\calr^{\sigma_1}_{\rm PW}(\tau_1)\circ\left(\calr^{\sigma_2}_{\rm PW}(\tau_2)\right)^{-1}\nn\\
&=&\calr^{\sigma_2}_{\rm PW}(\tau_2)\,\,\, \mathcal{E}_{\rm PW}^{\tau_1,\sigma_1}\left(\hat O^{\rm phys}_{C_2S|C_1}\right)\,\left(\calr^{\sigma_2}_{\rm PW}(\tau_2)\right)^{-1}\label{PWobstrafo}\\
&\approx& \calr^{\sigma_2}_{\rm PW}(\tau_2)\, \hat F_{O_{C_2S|C_1},T_1}(\tau_1)\,\left(\theta(-\sigma_1\hat p_1)\otimes I_{C_2S}\right)\,\left(\calr^{\sigma_2}_{\rm PW}(\tau_2)\right)^{-1}\nn\\
&=& \left(\theta(-\sigma_1\hat p_1)\otimes I_{S}\right)\,\hat O^{\rm phys}_{C_1S|C_2}(\tau_1,\tau_2)\,.\nn
\ea
In the second line we made use of Eq.~\eqref{encode}, in the third of Theorem~\ref{thm_PWobseq}, and in the fourth of Eq.~\eqref{PWobsC2} and the fact that $\theta(-\sigma_1\hat p_1)$ commutes with the reduction map of the $C_2$ clock and with $\hat F_{O_{C_2S|C_1},T_1}(\tau_1)$ (see Lemma~\ref{lem_projTcommute}). 

Observe that the structure of this transformation shows that reduced observables relative to one clock will transform always via the gauge-invariant Dirac observable algebra to reduced observables relative to another clock. 

Using Eqs.~\eqref{encode}, \eqref{PWred} and~\eqref{PWredinv}, we can write this transformation  as

\ba
 \left(\theta(-\sigma_1\hat p_1)\otimes I_{S}\right)\,\hat O^{\rm phys}_{C_1S|C_2}(\tau_1,\tau_2) = \bra{\tau_2,\sigma_2}\,\delta(\hat C_H)\,\left(\ket{\tau_1,\sigma_1}\!\bra{\tau_1,\sigma_1}\otimes \hat O^{\rm phys}_{C_2S|C_1}\right)\,\delta(\hat C_H)\,\ket{\tau_2,\sigma_2}\,.\label{PWobstrafo2}
 \ea

 This transformation reveals that expectation values are preserved in the following manner:
 \begin{align}
 &\bra{\psi^{\sigma_2}_{C_1S|C_2}(\tau_2)}\,\left(\theta(-\sigma_1\hat p_1)\otimes I_S\right)\,\hat O^{\rm phys}_{C_1S|C_2}(\tau_1,\tau_2)\,\ket{\phi_{C_1S|C_2}^{\sigma_2}(\tau_2)}   \nn \\ 
 & \hspace{2in }= \bra{\psi^{\sigma_1}_{C_2S|C_1}(\tau_1)}\,\left(\theta(-\sigma_2\hat p_2)\otimes I_S\right)\,\hat O^{\rm phys}_{C_2S|C_1}\,\left(\theta(-\sigma_2\hat p_2)\otimes I_S\right)\,\ket{\phi_{C_2S|C_1}^{\sigma_1}(\tau_1)}\,.\label{expecpres}
\end{align}
The projectors onto the $\sigma_2$-sector on the r.h.s.\ appears because the $C_2$ reduction map in Eq.~\eqref{PWobstrafo} induces such a projection (compare this with the state transformations Eq.~\eqref{PWtrafo} which are dual).  In other words, \emph{only the physical information in the overlap of the $\sigma_1$- and $\sigma_2$-sector is preserved when changing from the description relative to clock $C_1$ to one relative to clock $C_2$, or vice versa}. Once more, this is a direct consequence of the double superselection rule induced by the shape of the constraint Eq.~\eqref{CH2}.

 \subsubsection{Observable transformations in the relational Heisenberg picture}
 
 The argumentation for the relational Heisenberg picture proceeds in complete analogy. We thus just quote the results, which immediately follow from those of the previous subsection through use of Lemma~\ref{lem_3}. Of course, in this case, the reduced observables have an explicit dependence on the evolution parameter, Eq.~\eqref{Heisenbergobs}.
 
 The observable transformations from the relational Heisenberg picture of the $\sigma_1$-sector of clock $C_1$ into the relational Heisenberg picture of the $\sigma_2$-sector of clock $C_2$ are given by
\ba
\Lambda^{\sigma_1\to \sigma_2}_{\rm QR} \, \hat O^{\rm phys}_{C_2S|C_1}(\tau_1)\,\left(\Lambda^{\sigma_1\to\sigma_2}_{\rm QR}\right)^{-1}
&\approx& \calr^{\sigma_2}_{\rm QR}(\tau_2)\, \hat F_{O_{C_2S|C_1},T_1}(\tau_1)\,\left(\theta(-\sigma_1\hat p_1)\otimes I_{C_2S}\right)\,\left(\calr^{\sigma_2}_{\rm QR}\right)^{-1}\nn\\
&=& \left(\theta(-\sigma_1\hat p_1)\otimes I_{S}\right)\,U^\dag_{C_1S}(\tau_2)\,\hat O^{\rm phys}_{C_1S|C_2}(\tau_1,\tau_2)\,U_{C_1S}(\tau_2)\,\label{QRobstrafo}\\
&=:& \left(\theta(-\sigma_1\hat p_1)\otimes I_{S}\right)\,\hat O^{\rm Heis}_{C_1S|C_2}(\tau_1,\tau_2)\,,\nn
\ea
where $\hat O^{\rm phys}_{C_1S|C_2}(\tau_1,\tau_2)$ is given by Eqs.~\eqref{PWobsC2} and~\eqref{PWobstrafo2}.
 Thanks to the double superselection rule, this transformation preserves expectation values again per overlap of a $\sigma_1$- with a $\sigma_2$-sector, in obvious analogy to Eq.~\eqref{expecpres}.

\subsection{Occurrence of indirect clock self-reference}\label{sec_selfref}

Finally, let us now come back to the indirect self-reference effect of clock $C_2$ alluded to above. The following theorem, which is adapted from  \cite{Hoehn:2019owq} and whose proof applies here per pair of $\sigma_1$- and $\sigma_2$-sector, reveals the necessary and sufficient conditions for this indirect self-reference to occur:
\begin{theorem} \label{tauindependent} 
Consider an operator $\hat O_{C_2S|C_1}^{\rm phys} \in \mathcal{L}(\mathcal{H}_{C_2}^{\rm phys} \otimes \mathcal{H}_S^{\rm phys} )$ of the composite system  $C_2S$ described from the perspective of clock $C_1$. From the perspective of clock $C_2$,  this  operator is independent of $\tau_2$, so that $\hat{O}_{C_1S|C_2}^{\rm phys}(\tau_1, \tau_2) = \hat{O}_{C_1S|C_2}^{\rm phys}(\tau_1) \in \mathcal{L}(\mathcal{H}_{C_1}^{\rm phys} \otimes \mathcal{H}_S^{\rm phys})$ if and only if  
\begin{align}
\hat O_{C_2S|C_1}^{\rm phys}  = \sum_i \left( \hat O_{C_2|C_1}^{\rm phys} 
\right)_i \otimes \left( \hat f_{S|C_1}^{\rm phys} \right)_i, \nn
\end{align}
where $(\hat f_{S|C_1}^{\rm phys} )_i$ is an operator on $\ch^{\rm phys}_S$ and $( \hat O_{C_2|C_1}^{\rm phys} )_i$ is a constant of motion, $\left[( \hat O_{C_2|C_1}^{\rm phys} )_i,s_2\f{\hat p_2^2}{2}\right]=0$. Furthermore, in this case the transformed observable reads
\begin{align}
 \left[\theta(-\sigma_1\hat p_1)\otimes I_{S}\right]\hat O_{C_1S|C_2}^{\rm phys}(\tau_1) = \Pi_{\sigma_{C_1SC_2}} \Bigg[ \! \sum_i  \mathcal{G}_{C_1S} \!\left( \ket{\tau_1,\sigma_1}\! \bra{\tau_1,\sigma_1} \otimes 
\left(\hat f_{S|C_1}^{\rm phys} \right)_i  \right) \bra{t_2,\sigma_2} \left( \hat O_{C_2|C_1}^{\rm phys} \right)_i \!\delta(\hat C_H)\!\ket{t_2,\sigma_2}\! \Bigg] \Pi_{\sigma_{C_1SC_2}}, \nn
\end{align}
where $\Pi_{\sigma_{C_1SC_2}}= \theta\left[-s_2\left(s_1\f{\hat p_1^2}{2}+\hat H_S\right)\right]$ is the projector onto the physical subspace of $\ch_{C_1}\otimes\ch_S$,  $\ket{t_2,\sigma_2}$ is an arbitrary $\sigma_2$-sector clock state of $C_2$, and $\mathcal{G}_{C_1S}$ is the $G$-twirl over the group generated by the evolution generator $s_1\f{\hat p_1^2}{2} + \hat{H}_S$ of the composite system $C_1S$.
\end{theorem}
\end{widetext}
That is to say, the indirect self-reference effect and thus $\tau_2$-dependence of Eq.~\eqref{PWobstrafo} is absent if and only if the relational Dirac observable encoding how $C_2S$ properties evolve relative to $C_1$ does not contain any degrees of freedom of clock $C_2$ that evolve.

When $\hat O_{C_2S|C_1}^{\rm phys} =I_{C_2}\otimes \hat f^{\rm phys}_{S|C_1}$, i.e.\ only the evolution of system degrees of freedom relative to $C_1$ is described, Theorem~\ref{tauindependent} entails that the transformation to the description relative to $C_2$ simplifies as follows:
\ba
&& \left(\theta(-\sigma_1\hat p_1)\otimes I_{S}\right)\hat O_{C_1S|C_2}^{\rm phys}(\tau_1) = \nn\\
 &&\q\q\q\q\Pi_{\sigma_{C_1SC_2}} \mathcal{G}_{C_1S} \left( \ket{\tau_1,\sigma_1}\! \bra{\tau_1,\sigma_1} \otimes 
\hat f_{S|C_1}^{\rm phys} \right)\,\Pi_{\sigma_{C_1SC_2}}\,.\nn
\ea
 In particular, the transformed system observable is perspective independent, i.e.\ its description relative to $C_1$ and $C_2$ coincide if and only if it is a constant of motion (see  \cite{Hoehn:2019owq} for the proof of this statement, which again applies here per pair of $\sigma_1$- and $\sigma_2$-sector):
  \begin{corol}\label{corol_4}
An operator of $C_2S$ relative to $C_1$
\begin{align}
\hat O_{C_2S|C_1}^{\rm phys} =  I_{C_2} \otimes \hat{f}_{S|C_1}^{\rm phys}. \nn
\end{align}
transforms under a temporal frame change map to the perspective of $C_2$ as follows
\begin{align}
\hat O_{C_1S|C_2}^{\rm phys} =  I_{C_1} \otimes \hat{f}_{S|C_2}^{\rm phys}, \nn
\end{align}
where $\hat{f}^{\rm phys}_{S|C_1} =  \hat{f}^{\rm phys}_{S|C_2}$ if and only if $\hat{f}^{\rm phys}_{S|C_1}$ is a constant of motion, $ [ \hat{f}_{S|C_1}^{\rm phys}, \hat{H}_{S} ] = 0$. 
\end{corol}

 Theorem~\ref{tauindependent} translates as follows into the relational Heisenberg picture (see \cite{Hoehn:2019owq} for the proof which applies here per pair of $\sigma_1$- and $\sigma_2$-sector):
 \begin{corol}
Let $\hat O_{C_2S|C_1}^{\rm phys}(\tau_1) \in \mathcal{L}(\mathcal{H}_{C_2}^{\rm phys} \otimes \mathcal{H}_S^{\rm phys})$ be an operator describing the dynamics of properties of the composite system $C_2S$ relative to $C_1$ in the Heisenberg picture. Under a temporal frame change Eq.~\eqref{QRobstrafo} to the perspective of $C_2$, this operator transforms to an operator $\hat O^{\rm Heis}_{C_1S|C_2}(\tau_1,\tau_2)$ that satisfies the Heisenberg equation of motion in clock $C_2$ time $\tau_2$ without an explicitly $\tau_2$ dependent  term, 
\begin{align}
\frac{d}{d \tau_2} \hat O_{C_1S|C_2}^{\rm Heis}(\tau_1, \tau_2) = i \left[s_2\f{\hat p_2^2}{2}+ \hat{H}_S ,\hat O_{C_1S|C_2}^{\rm Heis}(\tau_1, \tau_2)\right],\nn 
\end{align}
if and only if  
\begin{align}
\hat O_{C_2S|C_1}^{\rm phys}(\tau_1)  = \sum_i \left( \hat O_{C_2|C_1}^{\rm phys} 
\right)_i \otimes \left( \hat f_{S|C_1}^{\rm phys} (\tau_1)  \right)_i, \nn
\end{align}
and $\left(\hat O_{C_2|C_1}^{\rm phys}\right)_i$ is a constant of motion, $[ s_1\f{\hat p_1^2}{2} , \hat O_{C_2|C_1}^{\rm phys}] =0$. 
\end{corol} 
 
The interpretation of the transformations is of course completely analogous to the relational Schr\"odinger picture.\\~

\subsection{Application: comparing clock readings}

One application of the temporal frame change method developed above is comparing readings of different clocks. This is also a prerequisite for developing a notion of clock synchronization. 


For example, we may wish to compare the evolution of some system property $\hat f_S$ relative to clock $C_1$ with $\hat f_S$ relative to clock $C_2$. These two relational evolutions will be encoded in two one-parameter families of Dirac observable of the form $\hat F_{I_{C_2}\otimes f_S,T_1}(\tau_1)$ and $\hat F_{I_{C_1}\otimes f_S,T_2}(\tau_2)$.  In order to relate these two dynamics, we need a consistent method for relating the different clock readings $\tau_1,\tau_2$. While classically, there is an unambiguous way to answer the question ``what is the value of the reading $\tau_2$ of clock $C_2$, when clock $C_1$ reads $\tau_1$?", namely by setting $\tau_2(\tau_1):=F_{T_2,T_1}(\tau_1)$, this is not so in the quantum theory because both clocks are now described in terms of quantum operators and their relation depends on the quantum state. In fact, we shall argue shortly that comparing clock readings is generally dependent on the choice of temporal frame (here either $C_1$ or $C_2$) in the quantum theory. 

\subsubsection{Three ways of comparing clock readings}

To address this conundrum in the quantum theory, let us recall the conditional probabilities in Eq.~\eqref{fwhenT} and ask for the probability that $C_2$ reads $\tau_2$ when $C_1$ reads $\tau_1$ (ignoring frequency sectors for simplicity for the moment):
\begin{align}
&P(T_2=\tau_2|T_1=\tau_1) \nn \\
& \quad ={\bra{\psi_{\rm phys} }e_{C_1}(\tau_1) \otimes e_{C_2}(\tau_2) \otimes I_S \ket{\psi_{\rm phys} }_{\rm kin}} \nn\\
& \quad = P(T_1=\tau_1|T_2=\tau_2) . \label{T2whenT1}
\end{align}
Here we have assumed that the physical state is normalized such that by Corollary~\ref{corol_2} also the reduced states in the Schr\"odinger picture of either clock are normalized. \\~

\noindent {\bf Comparing clock readings.} \emph{Given the conditional probabilities Eq.~\eqref{T2whenT1}, we may consider the following  three generally distinct options for comparing clock readings.  }
\begin{itemize}
\item[(A)] \emph{The clock reading of $C_2$ when $C_1$ reads $\tau_1$ is defined to be the value of $\tau_2$ that maximizes the conditional probability $P(T_2=\tau_2|T_1=\tau_1)$. This assumes the distribution to have a unique maximal peak.}

\item[(B)] \emph{The clock reading of $C_2$ when $C_1$ reads $\tau_1$ is defined to be the expectation value}
\ba
\tau_2(\tau_1)\ce \int_\mathbb{R}\,d\tau'\,\tau'\,P(T_2=\tau'|T_1=\tau_1)\,.\label{clockB}
\ea
\item[(C)] \emph{The clock reading of $C_2$ when $C_1$ reads $\tau_1$ is defined to be }
$
\left(\tau_2^{(n)}(\tau_1)\right)^{1/n}
$
\emph{for $n>1$, where}
\ba
\tau^{(n)}_2(\tau_1)\ce \int_\mathbb{R}\,d\tau'\,(\tau')^n\,P(T_2=\tau'|T_1=\tau_1)\,\label{clockBn}
\ea
\emph{is the $n^{\rm th}$-moment of the conditional probability distribution in Eq.~\eqref{T2whenT1}.}

\end{itemize}
Relating different clock readings in terms of expectation values, as in (B), is arguably the most natural choice and has originally been discussed in \cite{hoehnHowSwitchRelational2018,Hoehn:2018whn,Bojowald:2010xp,Bojowald:2010qw,Hohn:2011us,Smith:2019imm}; we expand on this here.

Clearly, the two definitions (A) and (B) only agree when the conditional probability distribution is peaked on the expectation value. Furthermore, all three definitions (A)--(C) agree in the special case that ${P(T_2=\tau'|T_1=\tau_1)=\delta(\tau'-\tau_1)}$, i.e.\ when there are no fluctuations in the conditional probability distribution.

\subsubsection{Comparing clock readings for quadratic clock Hamiltonians}

Let us now explore these definitions in our present class of models defined by Eq.~\eqref{CH2}, taking into account the different frequency sectors again. Minding the double superselection rule, we replace Eq.~\eqref{T2whenT1} by 
\begin{align}
    &P_{\sigma_1,\sigma_2}(T_2=\tau_2|T_1=\tau_1) \nn \\
    &\quad ={\bra{\psi_{\rm phys} }e^{\sigma_1}_{C_1}(\tau_1) \otimes e^{\sigma_2}_{C_2}(\tau_2) \otimes I_S \ket{\psi_{\rm phys} }_{\rm kin}} \nn \\
    &\quad ={\bra{\psi_{\sigma_1,\sigma_2} }e_{C_1}(\tau_1) \otimes e_{C_2}(\tau_2) \otimes I_S \ket{\psi_{\sigma_1,\sigma_2} }_{\rm kin}} \,,\label{t2t1sigma}
\end{align}
where $\ket{\psi_{\sigma_1,\sigma_2}}\in\ch_{\sigma_1,\sigma_2}$ lies in the overlap of the $\sigma_1$- and $\sigma_2$-sectors (see Eq.~\eqref{overlap}) and $e^{\sigma_i}_{C_i}(\tau_i)\ce\f{1}{2\pi}\,\ket{\tau_i,\sigma_i}\!\bra{\tau_i,\sigma_i}$, $i=1,2$. We can then write the $n^{\rm th}$-moment of the conditional probability distribution in Eqs.~\eqref{clockB} and~\eqref{clockBn} for $n\in\mathbb{N}$, thus considering both definitions (B) and (C), as follows:
\ba
\tau^{(n)}_2(\tau_1)&=& \int_\mathbb{R}\,d\tau'\,(\tau')^n\,P_{\sigma_1,\sigma_2}(T_2=\tau'|T_1=\tau_1)\nn\\
&=& \braket{\psi_{C_2S|C_1}^{\sigma_1}(\tau_1)\,|\,\hat T^{(n)}_{2,\sigma_2}\otimes I_S\,|\psi_{C_2S|C_1}^{\sigma_1}(\tau_1)}\,\,\q\q\label{t2whent1sigma}\\
&=& \braket{\psi_{\sigma_1,\sigma_2}\,|\,\hat F_{T^{(n)}_2\otimes I_S,T_1}(\tau_1)\,)|\,\psi_{\sigma_1,\sigma_2}}_{\rm phys}\,,\nn
\ea
where by Eq.~\eqref{degnthmom}
\ba
\hat T^{(n)}_{2,\sigma_2} &=& \f{1}{2\pi}\,\int_\mathbb{R}\,dt\,t^n\,\ket{t,\sigma_2}\!\bra{t,\sigma_2}\nn\\
&=&\theta(-\sigma_2\hat p_2)\,\hat T^{(n)}_2\,\theta(-\sigma_2\hat p_2)\,\nn
\ea
is the $\sigma_2$-sector $n^{\rm th}$-moment of the covariant clock POVM corresponding to $C_2$.
In the second line of Eq.~\eqref{t2whent1sigma} we have made use of Eqs.~\eqref{PWred} and~\eqref{t2t1sigma}, while in the third line we invoked Theorem~\ref{thm_PWexpec}. Note that by Eq.~\eqref{expecpres}, the expression in Eq.~\eqref{t2whent1sigma} defines an expectation value which is preserved during a temporal frame change between $C_1$ and $C_2$.

Thanks to Lemmas~\ref{lem_almosteigen} and~\ref{thm_sameT2} we can write the $n^{\rm th}$-moment in~Eq.~\eqref{t2whent1sigma} also in the form
\ba
\tau_2^{(n)}(\tau_1) &=&\braket{\psi_{C_2S|C_1}^{\sigma_1}(\tau_1)\,|\,\hat T^n_{2,\sigma_2}\otimes I_S\,|\psi_{C_2S|C_1}^{\sigma_1}(\tau_1)}\,\,\nn\\
&=& \braket{\psi_{\sigma_1,\sigma_2}\,|\,\hat F_{T^n_2\otimes I_S,T_1}(\tau_1)\,)|\,\psi_{\sigma_1,\sigma_2}}_{\rm phys}\,,\nn
\ea
as long as $\ket{\psi_{C_2S|C_1}^{\sigma_1}(\tau_1)}\in\left(\mathcal{D}(\hat T_{2}^{n})\cap\ch_{C_2}^{\rm phys}\right)\otimes\ch_S^{\rm phys}$. 
Since $
\left(\tau_2^{(n)}(\tau_1)\right)^{1/n}\neq \tau_2(\tau_1)$ for $n>1$ for general states, definitions (B) and (C) will generically not be equivalent. In the sequel, we shall mostly consider definition (B) in extension of  \cite{hoehnHowSwitchRelational2018,Hoehn:2018whn,Bojowald:2010xp,Bojowald:2010qw,Hohn:2011us,Smith:2019imm}. This seems to be the physically most appealing one, especially if an ensemble interpretation could be developed for the models under consideration. Definition (A) is only unambiguous when the conditional probability distribution has a single maximal peak and definition (C) is operationally unnatural and convoluted. That is, we set for the value of the reading of clock $C_2$ when $C_1$ reads $\tau_1$:
\ba
\tau_2(\tau_1):=\tau_2^{(1)}(\tau_1)\,.\label{synchro}
\ea
The following discussion, however, qualitatively also applies to definition (C).

\subsubsection{Comparing clock readings is temporal frame dependent}

Notice that definitions (A)--(C) treat $C_2$ as the fluctuating subsystem. We can thus interpret them as providing a definition of the clock reading of $C_2$ relative to the temporal reference frame $C_1$. Conversely, we can of course switch the roles of $C_1$ and $C_2$ above and ask for the clock reading of $C_1$ relative to $C_2$. Resorting to definition (B), this would yield
\ba
\tau_1(\tau_2)= \int_\mathbb{R}\,d\tau'\,\tau'\,P_{\sigma_1,\sigma_2}(T_2=\tau_2|T_1=\tau')\,.\label{t1whent2sigma}
\ea
Dropping the labels of the arguments in Eqs.~\eqref{t2whent1sigma} and~\eqref{t1whent2sigma}, both of which run over all of $\mathbb{R}$, it is important to note that $\tau_1(\tau)$ and $\tau_2(\tau)$ will generally \emph{not} be the same functions of $\tau$. This is because generally $P_{\sigma_1,\sigma_2}(T_2=\tau'|T_1=\tau)\neq P_{\sigma_1,\sigma_2}(T_2=\tau|T_1=\tau')$ in Eq.~\eqref{t2t1sigma}. Said another way, the evolution of $C_2$ from the perspective of $C_1$ according to definition (B) may differ from the evolution of $C_1$ relative to $C_2$ (for the same physical state).

One might wonder whether the function $\tau_1(\tau_2)$ in Eq.~\eqref{t1whent2sigma} is the inversion of $\tau_2(\tau_1)$ in Eq.~\eqref{t2whent1sigma}, i.e.\ obtained by solving $\tau_2(\tau_1)$ for $\tau_1$. Classically, this is certainly the case and it would entail that for a fixed clock reading $\tau_1^*$ of $C_1$ one finds $\tau_1(\tau_2(\tau_1^*)) = \tau_1^*$.
 Physically this would mean that both temporal reference frames $C_1$ and $C_2$ agree that when $C_1$ reads $\tau_1^*$, $C_2$ reads the value $\tau_2(\tau_1^*)$. This does occur in a special case when definitions (A)--(C) all coincide, namely when $P_{\sigma_1,\sigma_2}(T_2=\tau_2(\tau_1^*)|T_1=\tau') = \delta(\tau'-\tau_1^*)$ in Eq.~\eqref{t1whent2sigma} in which case expectation value, most probable value and the value defined through the $n^{\rm th}$-moment all agree. While this does happen in simple models with a high degree of symmetry between $C_1$ and $C_2$ \cite{Hoehn:2018whn}, this will in more interesting cases not be the case because the physical state will generically have a different spread along the $\tau_1$ and $\tau_2$ axes \cite{hoehnHowSwitchRelational2018,Bojowald:2010xp,Bojowald:2010qw,Hohn:2011us}. In our case this means that the wave function 
 \ba
 \psi^{\sigma_1,\sigma_2}_{C_2S|C_1}(\tau_1,\tau_2)\ce\left(\bra{\tau_1,\sigma_1}\otimes\bra{\tau_2,\sigma_2}\otimes\bra{\phi_S}\right)\,\ket{\psi_{\rm phys}}\,,\nn
 \ea
 for some physical system state $\ket{\phi_S}\in\ch_S^{\rm phys}$, which can be viewed as either a wave function in the $C_1$ or $C_2$ relational Schr\"odinger picture, may have a different spread in $\tau_1$ than in $\tau_2$. In such a case we will generally find $\tau_1(\tau_2(\tau_1^*)) \neq \tau_1^*$. This effect will occur in the class of models considered here because physical states need not have the same momentum distribution in $p_1$ and $p_2$ (and thus neither in $\tau_1$ or $\tau_2$) due to the presence of the system $S$. 
This effect has also been demonstrated in a semiclassical approach in various models in \cite{Bojowald:2010xp,Bojowald:2010qw,Hohn:2011us} where one finds discrepancies of the order of $\hbar$ between $\tau_1^*$ and $\tau_1(\tau_2^* = \tau_2(\tau_1^*))$.

In conclusion, this effect can be interpreted as a temporal frame dependence of comparing clock readings according to definition (B) (or (C)): if  from the perspective of the temporal reference frame defined by $C_1$ the clock $C_2$ reads $\tau_2(\tau_1^*)$ (computed according to Eq.~\eqref{t2whent1sigma}) when $C_1$ reads $\tau_1^*$, then conversely from the perspective of the temporal reference frame defined by $C_2$ the clock  $C_1$ will not in general read $\tau_1^*$ when $C_2$ reads the value $\tau_2(\tau_1^*)$. That is, $C_1$ and $C_2$ will generally disagree about the pairings of their clock readings.

Let us now also briefly comment on the notion of quantum clock synchronization. Using the state dependent relation Eq.~\eqref{synchro}, we could ask for which state would yield $\tau_2(\tau_1^*)=\tau_1^*$ so that $C_1$ and $C_2$ read the \emph{same} value when $C_1$ reads the value $\tau_1^*$. Even stronger, we could ask whether there are states for which $\tau_2(\tau_1)=\tau_1+{const}$, for all $\tau_1\in\mathbb{R}$, so that, up to a constant offset, $C_1$ and $C_2$ are always synchronized. Eq.~\eqref{t2whent1sigma} tells us that this is the case if $P_{\sigma_1,\sigma_2}(T_2=\tau'|T_1=\tau_1)=\delta(\tau'-\tau_1-const)$. Again, while this happens in simple models \cite{Hoehn:2018whn}, this will generically not happen for the models in the class which we are studying on account of the above observations concerning the frame dependence of comparing clock readings. Such a notion of synchronization is therefore too strong and can generally not be implemented. It will furthermore generally be frame dependent too. 

\subsubsection{Comparing a system's evolution relative to two clocks}

Returning to our original ambition, it is thus more useful to employ the more general (frame dependent) clock comparison, according to definition (B), in order to compare the evolutions of $S$ with respect to $C_1$ and $C_2$. Working in the relational Schr\"odinger picture, if $\ket{\psi_{C_2S|C_1}^{\sigma_1}(\tau_1^*)}$ is the initial state of $C_2S$ from the perspective of $C_1$, then according to Eq.~\eqref{PWtrafo} the corresponding initial state of $C_1S$ from the perspective of $C_2$ is 
\begin{align}
    &\left(\theta(-\sigma_1\hat p_1)\otimes I_S\right)\ket{\psi_{C_1S|C_2}^{\sigma_2}(\tau_2(\tau_1^*))} \nn \\
    &\quad= \calr_{\rm PW}^{\sigma_2}(\tau_2(\tau_1^*))\circ \left(\calr^{\sigma_1}_{\rm PW}(\tau_1^*)\right)^{-1} \ket{\psi_{C_2S|C_1}^{\sigma_1}(\tau_1^*)}\,.\label{synchrostate}
\end{align}
We can then evaluate the `same' reduced system observable $I_{C_i}\otimes\hat f_S^{\rm phys}$ in the two states, where $i=1$ when evaluated relative to $C_2$ and vice versa (cf.\ Corollary \ref{corol_4}), in order to compare the evolution of property $\hat f_S^{\rm phys}$ relative to the two clocks in different quantum states (which amount also to quantum states of the clocks). To avoid confusion, we emphasize, that $I_{C_i}\otimes\hat f_S^{\rm phys}$, $i=1,2$, correspond to two different relational Dirac observables $\hat F_{I_{C_2}\otimes f_S,T_1}(\tau_1)$ and $\hat F_{I_{C_1}\otimes f_S,T_2}(\tau_2)$ on the clock-neutral physical Hilbert space $\ch_{\rm phys}$; in particular, the two are \emph{not} related by the TFC map $\Lambda_{\rm PW}^{\sigma_1\to\sigma_2}$. Hence, by evaluating these two reduced observables in the relational Schr\"odinger states related via the TFC map $\Lambda_{\rm PW}^{\sigma_1\to\sigma_2}$ by Eq.~\eqref{synchrostate}, we can compare two genuinely distinct relational dynamics. The construction in the relational Heisenberg picture is of course completely analogous.

In \cite{Hoehn:2019owq,castro-ruizTimeReferenceFrames2019} a frame dependent temporal non-locality effect was exhibited for idealized clocks whose Hamiltonian is the unbounded momentum operator. For example, when clock $C_2$ is seen to be in a superposition of two peaked states and in a product relation with $S$ from the perspective of $C_1$, then $C_1S$ will generally be entangled as seen from the perspective of $C_2$ and undergo a superposition of time evolutions.   This effect  applies here per overlap of the different $\sigma_1$- and $\sigma_2$-sectors. It will be interesting to study how such a frame dependent temporal locality affects the (potentially frame dependent) comparison and synchronization of the clocks and the comparison of the evolutions of $S$ relative to $C_1$ and $C_2$ in different quantum states, corresponding to different choices of the clock-neutral physical states. Such an exploration will appear elsewhere. 

Finally, these temporal frame changes and clock synchronizations will be relevant in quantum cosmology. For example, recently it was pointed out that singularity resolution in quantum cosmology depends on the choice of clock which one uses to define a relational dynamics \cite{Gielen:2020abd}. The different relational dynamics employed in \cite{Gielen:2020abd} can be interpreted as different choices of reduced dynamics in the sense of our relational Schr\"odinger/Heisenberg picture. Temporal frame changes as developed here can in principle be used to study the temporal frame dependence of the fate of cosmological singularities more systematically.


\section{Conclusions}

In this work we demonstrated the equivalence of three distinct approaches to relational quantum dynamics\,---\, relational Dirac observables,  the Page-Wootters formalism, and quantum deparametrizations\,---\,for  models described by a Hamiltonian constraint in which the momentum of the system being employed as a clock appears quadratically. {Since} this class of {models} {encompasses} many relativistic settings, 
{we have thereby extended our previous results of \cite{Hoehn:2019owq} into a relativistic context. A crucial ingredient in this extension has been  a clock POVM which is covariant with respect to the group generated by the Hamiltonian constraint and is used to describe the temporal reference frame defined by the clock. This choice differs from the usual resort to self-adjoint clock operators in relativistic settings. 

Owing to a superselection rule induced by the shape of the Hamiltonian constraint across positive and negative frequency modes, this equivalence, which we refer to as the trinity of relational quantum dynamics, holds frequency sector wise. Moreover, we {further} develop the {method} of temporal quantum frame changes \cite{hoehnHowSwitchRelational2018,Hoehn:2018whn,Hoehn:2019owq,castro-ruizTimeReferenceFrames2019,Bojowald:2010xp,Bojowald:2010qw,Hohn:2011us,Bojowald:2016fac} in this setting to address the multiple choice problem. This method is then used to explore an indirect self-reference phenomen{on} that arises when transforming between clock perspectives and to reveal the temporal frame and state dependence of comparing or even synchronizing the readings of different quantum clocks. This result adds to the growing list of quantum reference frame dependent physical properties, such as entanglement \cite{giacominiQuantumMechanicsCovariance2019,Vanrietvelde:2018dit,hamette2020quantum}, spin \cite{giacominiRelativisticQuantumReference2019}, classicality \cite{Vanrietvelde:2018dit} or objectivity \cite{le2020blurred,tuziemski2020decoherence} of a subsystem, superpositions \cite{giacominiQuantumMechanicsCovariance2019,Vanrietvelde:2018dit,Zych:2018nao}, certain quantum resources \cite{savi2020quantum}, measurements \cite{giacominiQuantumMechanicsCovariance2019,yang2020switching}, causal relations \cite{castro-ruizTimeReferenceFrames2019,Guerin:2018fja}, temporal locality \cite{Hoehn:2019owq,castro-ruizTimeReferenceFrames2019}, and even spacetime singularity resolution \cite{Gielen:2020abd}.  The temporal frame changes may also be employed to extend recent proposals for studying time dilation effects of quantum clocks} \cite{Smith:2019imm,Smith:2020gdj} (see also \cite{zych2011quantum,Khandelwal:2019mem,Paige:2018lrd,Grochowski:2020mru}).

Importantly, the covariant clock POVM permitted us to resolve Kucha\v{r}'s criticism that the Page-Wootters formalism does not produce the correct localization probability for a relativistic particle in Minkowski space \cite{kucharTimeInterpretationsQuantum2011a}. Indeed, such incorrect localization probabilities arise when conditioning on times defined by the quantization of an inertial Minkwoski time coordinate.
We showed that conditioning instead on the covariant clock POVM surprisingly produces a Newton-Wigner type localization probability, which, while approximate and not fully covariant, is usually regarded as the best possible notion of localization in relativistic quantum mechanics \cite{haag2012local,Yngvason:2014oia}.
This result underscores the benefits of covariant clock POVMs in defining a consistent relational quantum dynamics \cite{Brunetti:2009eq,Smith:2017pwx,Smith:2019imm, Hoehn:2019owq,Loveridge:2019phw}.

In conjunction with our previous article \cite{Hoehn:2019owq}, we have thus resolved all three criticisms (a)--(c) (see Introduction) that Kucha\v{r} raised against the Page-Wootters formalism in \cite{kucharTimeInterpretationsQuantum2011a}. The Page-Wootters formalism is therefore a viable approach to relational quantum dynamics. Through the equivalence established by the trinity, it also equips the relational observable formulation and deparametrizations with a consistent conditional probability interpretation. In particular, relational observables describing the evolution of a position operator relative to a covariant clock POVM yield a Newton-Wigner type localization in relativistic settings.

\begin{acknowledgments}
PAH is grateful for support from the Simons Foundation through an `It-from-Qubit' Fellowship and the Foundational Questions Institute through Grant number FQXi-RFP-1801A. ARHS acknowledges support from the Natural Sciences and Engineering Research Council of Canada and the Dartmouth Society of Fellows. MPEL acknowledges support from the ESQ Discovery Grant (ESQ-Projekts0003X2) of the Austrian Academy of Sciences (\"{O}AW), as well as from the Austrian Science Fund (FWF) through the START project Y879-N27. This work was supported in part by funding from Okinawa Institute of Science and Technology Graduate University. The initial stages of this project were made possible through the support of a grant from the John Templeton Foundation. The opinions expressed in this publication are those of the authors and do not necessarily reflect the views of the John Templeton Foundation.
\end{acknowledgments}

\bibliography{relativisticTrinity}

\onecolumngrid

\pagebreak

\appendix

\section{Proofs of lemmas of Sec.~\ref{sec_covtime}}
\label{app_deg}

We first state the conditions defining the domain of $\hat T$ \cite{holevoProbabilisticStatisticalAspects1982,buschOperationalQuantumPhysics}:
\ba \label{domain}
\cd(\hat T)&=&\Bigg\{\psi(p_t)\,\bigg\vert\,\lim_{p_t\to0}\left[\f{\psi(p_t)}{\sqrt{|p_t|}}\right]=0\,,\,\int_\mathbb{R}\,\f{dp_t}{|p_t|}\,\Bigg|\f{d}{d p_t}\,\f{\psi(p_t)}{\sqrt{|p_t|}}\Bigg|^2<\infty\Bigg\}\, .
\ea
These conditions will feature in some of the following proofs.

\medskip
\noindent{\bf Lemma \ref{lem_degresid}.} 
\emph{The clock states $\ket{t,\sigma}$ defined in Eq.~(\ref{degclock}) integrate to projectors onto the positive/negative frequency sectors on $\ch_C$
\ba
\f{1}{2\pi}\,\int_{-\infty}^{\infty}\,dt\,\ket{t,\sigma}\bra{t,\sigma} = \theta(-\sigma\,\hat{p}_t)\,\nn
\ea
and hence form a resolution of the identity as follows:
\ba
\f{1}{2\pi}\,\sum_{\sigma=+,-}\,\int_{-\infty}^{\infty}\,dt\,\ket{t,\sigma}\bra{t,\sigma} = I_C\,.\nn
\ea
}

\begin{proof}
For the negative frequency sector, direct computation yields 
\ba
\int_{-\infty}^{\infty}\,dt\,\ket{t,-}\bra{t,-}
&=&2\pi\,\int_0^\infty\,dp_t\,dp_t'\sqrt{p_tp_t'}\,\delta\left(\f{p_t^2-p_t'^2}{2}\right)\,\ket{p_t}\bra{p_t'}\nn\\
&=&2\pi \,\int_0^\infty\,d\varepsilon\,dp_t'\f{\sqrt{p_t'}}{(2\varepsilon)^{1/4}}\,\delta\left(\varepsilon - \f{p_t'^2}{2}\right)\,\ket{\sqrt{2\varepsilon}}\bra{p_t'}\nn\\
&=&2\pi\,\int_{0}^\infty\,dp_t'\,\ket{p_t'}\bra{p_t'} = 2\pi\,\theta(\hat{p}_t)\,.\nn
\ea
In the third line we have made use of the variable transformation $\varepsilon = p_t^2/2$. Note that since the clock states do not  have support on  $\ket{p_t=0}$, the case $p_t'=0$ does not occur in the delta function.
The computation for the positive frequency sector is analogous.
\end{proof}

\noindent{\bf Lemma~\ref{lem_sameT}.}
\emph{$\cd(\hat T^{(1)}) = \cd(\hat T)$ and $\hat T = \hat T^{(1)}$.}

\begin{proof}
We first prove that $\cd(\hat T^{(1)}) = \cd(\hat T)$, and then that $\hat T\ket{\psi}=\hat T^{(1)}\ket{\psi}$, $\forall\,\ket{\psi}\in\cd(\hat T)$, which together imply the second statement in Lemma~\ref{lem_sameT}. We begin by finding $\cd(\hat T^{(1)}) \subset \ch_C$, i.e.\ the elements of $\ch_C$ whose norm remains finite after the action of $\hat T^{(1)}$. Using the variable transformation $\varep = p_t^2/2$ again, we can write $\hat T^{(1)}$ as
\ba
\hat T^{(1)} = \f{1}{2\pi}\int_\mathbb{R}dt\,t \int_0^\infty d \varep \,d \varep' ( 4 \varep \varep')^{-1/4} \,e^{-i\,t\,s\,(\varep - \varep')}\, \left( \ket{p_t=\sqrt{2\varep}}\bra{p_t=\sqrt{2\varep'}} + \ket{p_t=-\sqrt{2\varep}}\bra{p_t=-\sqrt{2\varep'}} \right) \label{sameTT1}
\ea 
and therefore $\hat T^{(1)} $ acts on an arbitrary state $\ket{\psi}=\int_{-\infty}^{\infty}dp_t\,\psi(p_t)\ket{p_t} $ as
\ba
\hat T^{(1)}\ket{\psi} = \int_0^\infty d \varep \sum_{\sigma} \varphi_\sigma (\varep) \ket{p_t=\sigma \sqrt{2\varep}} \label{T1psienergy}
\ea
with the definition
\ba
\varphi_\sigma (\varep) &:=& \f{1}{2\pi}\int_\mathbb{R}dt\,t \int_0^\infty \,d \varep' ( 4 \varep \varep')^{-1/4} \,e^{-i\,t\,s\,(\varep - \varep')} \, \psi\left(\sigma \sqrt{2 \varep'}\right) \nn \\
&=& -i\,s \int_0^\infty \,d \varep' ( 4 \varep \varep')^{-1/4} \, \delta ' (s(\varep' - \varep)) \, \psi\left(\sigma \sqrt{2 \varep'}\right) \nn \\
&=& -i\,s \left\lbrace \left[ ( 2 \varep)^{-\frac{1}{4}} \, \delta\left(\frac{p_t'^2}{2} - \varep \right) \, \frac{\psi(\sigma p_t')}{\sqrt{p_t'}} \right]_{p_t'=0}^{p_t'=\infty} + \left[ \frac{\psi\left(\sigma \sqrt{2 \varep}\right)}{2(2\varep)^{3/2}} - \sigma \frac{\psi '\left(\sigma \sqrt{2 \varep}\right)}{2\varep} \right] \right\rbrace, \label{sameTactionT1}
\ea
where a $'$ before a function's argument denotes the derivative of that function with respect to the argument, the third line is obtained by integrating by parts, and we have used the fact that $s=1/s$ as well as the scaling identity of the delta function $\delta(\alpha \varep)=\delta(\varep)/|\alpha|$ for real $\alpha\neq 0$. Recalling that $\psi(p_t)$ is square-integrable, we have
\ba
\lim_{p_t' \to \infty} \left[ ( 2 \varep)^{-\frac{1}{4}} \, \delta\left(\frac{p_t'^2}{2} - \varep \right) \, \frac{\psi(\sigma p_t')}{\sqrt{p_t'}} \right] = 0 .
\ea
Considering the $p_t'=0$ term in Eq.~(\ref{sameTactionT1}), we see that Eq.~(\ref{T1psienergy}) diverges unless
\ba
\lim_{p_t \to 0} \left[ \frac{\psi( p_t)}{\sqrt{ |p_t|}} \right] = 0 ,
\ea
which is the first of the two conditions defining $ \cd(\hat T)$ (see Eq.~(\ref{domain})). Assuming this to be satisfied, Eq~(\ref{T1psienergy}) gives
\ba
\hat T^{(1)}\ket{\psi} &=& i\,s \int_0^\infty d \varep \sum_{\sigma}  \left[ \sigma \frac{\psi '\left(\sigma \sqrt{2 \varep}\right)}{2\varep} -\frac{\psi\left(\sigma \sqrt{2 \varep}\right)}{2(2\varep)^{3/2}} \right] \ket{p_t=\sigma \sqrt{2\varep}} \nn \\
&=&   i\,s \int_\mathbb{R} d p_t \left[ \sgn(p_t) \frac{\psi '(p_t)}{|p_t|}-\frac{\psi(p)}{2|p_t|^2} \right] \ket{p_t} \label{T1momrep}  \\
&=&   i\,s \int_\mathbb{R} d p_t \left[ \frac{\psi '(p_t)}{p_t}-\frac{\psi(p)}{2p_t^2} \right] \ket{p_t} . \nn
\ea
Now, noting that $\frac{1}{|p_t|} \left\vert\frac{d}{dp_t}\frac{\psi(p_t)}{\sqrt{|p_t|}}\right\vert^2=\left\vert \frac{\psi '(p_t)}{|p_t|}-\sgn(p_t)\frac{\psi(p)}{2|p_t|^2} \right\vert^2=\left\vert \sgn(p_t)\frac{\psi '(p_t)}{|p_t|}-\frac{\psi(p)}{2|p_t|^2} \right\vert^2$, we compare this with the second line of Eq.~\eqref{T1momrep} to obtain
\ba
\left\vert\left\vert \hat T^{(1)}\ket{\psi} \right\vert\right\vert^2 = \int_\mathbb{R} \frac{d p_t}{|p_t|} \left\vert\frac{d}{dp_t}\frac{\psi(p_t)}{\sqrt{|p_t|}}\right\vert^2 .
\ea
The requirement that this quantity be finite is the second of the two conditions defining $ \cd(\hat T)$ (see Eq.~(\ref{domain})). Thus $\cd(\hat T^{(1)}) = \cd(\hat T)$. It remains to show that $\hat T = \hat T^{(1)}$. To this end, consider the direct canonical quantization $\hat T$, given in Eq.~\eqref{quantT}, of the classical time variable $T$. In particular consider its action on a state $\ket{\psi} \in \cd(\hat T)$:
\ba
\hat T\,\ket{\psi}&=&s\,\left(\hat t + \f{i}{2}\,\widehat{p_t^{-1}}\right)\widehat{p_t^{-1}} \,\ket{\psi}\nn\\
&=&i\,s\,\int_\mathbb{R}dp_t\,\left(\f{d}{d p_t}+\f{1}{2p_t}\right)\f{\psi(p_t)}{p_t}\,\ket{p_t}\label{useful}\\
&=&i\,s\,\int_\mathbb{R}dp_t\,\left[\f{\psi'(p_t)}{p_t}-\f{\psi(p_t)}{2p_t^2}\right]\,\ket{p_t}\, . \nn\\
&=&\hat T^{(1)}\ket{\psi} 
\ea
where we have used the fact that $\hat t$ acts as $i \f{d}{d p_t}$ in the momentum representation, and the last line follows by comparison with the last line of Eq.~\eqref{T1momrep}. Thus the action of $\hat T$ and $\hat T^{(1)}$ in momentum space is the same, as are their domains, and therefore $\hat T = \hat T^{(1)}$. 
\end{proof}

\noindent{\bf Lemma~\ref{lem_almosteigen}.}
\emph{
The clock states $\ket{t,\sigma}$ defined in Eq.~(\ref{degclock}) are not eigenstates of $\hat T$. However, for all ${\ket{\psi}\in\cd(\hat T)}$, they satisfy:
\ba
\bra{\psi}\,\hat T\,\ket{t,\sigma} = t\,\braket{\psi|t,\sigma}\,,\q\forall\,t\in\mathbb{R}\,, \, \sigma =\pm1\,.\nn
\ea
}

\begin{proof}
 We begin with the negative frequency clock states in Eq.~\eqref{degclock}, which can be equivalently written as
\ba
\ket{t,-}=\int_\mathbb{R}dp_t\,\sqrt{|p_t|}\,e^{-i\,t\,s\,p_t^2/2}\,\theta(p_t)\,\ket{p_t}\,.\nn
\ea
Next, {we again use Eq.~\eqref{quantT} and write} $\hat t = i\,d/dp_t$ under the integral in momentum representation, so that
\ba
\hat T\,\ket{t,-} &=& s\, \left(\hat t + \f{i}{2}\,\widehat{p_t^{-1}}\right)\widehat{p_t^{-1}} \,\ket{t,-}\nn\\
&=&i\,s\, \int_\mathbb{R}\,dp_t\,\left(\f{d}{d p_t}+\f{1}{2p_t}\right)\f{\theta(p_t)}{\sqrt{|p_t|}}\,e^{-i\,t\,s\,p_t^2/2}\, \ket{p_t}\,.\label{bla}
\ea
Noting that  
\ba
i\,\f{d}{dp_t}\,\f{1}{\sqrt{|p_t|}} = -\f{i\,\sgn(p_t)}{2|p_t|^{3/2}}\,,\label{useful2}
\ea
 we find
\ba
\hat T\,\ket{t,-} &=& \int_\mathbb{R}\,dp_t\,\f{1}{\sqrt{|p_t|}}\,\left(t\,p_t\,\theta(p_t)+i\,s\,\delta(p_t)\right)\,e^{-i\,t\,s\,p_t^2/2}\,\ket{p_t}\nn\\
&\underset{(\ref{degclock})}{=}& t\,\ket{t,-}+\int_\mathbb{R}\,dp_t\,\f{i\,s}{\sqrt{|p_t|}}\,\delta(p_t)\,e^{-i\,t\,s\,p_t^2/2}\,\ket{p_t}\,.\label{useful3}
\ea
Consider now any $\ket{\psi}\in\cd(\hat T)$, which by Eq.~\eqref{domain} satisfies $\lim_{p_t\to 0}\,\psi(p_t)/\sqrt{|p_t|}=0$. This immediately implies
\ba
\bra{\psi}\,\hat T\,\ket{t,-} = t\,\braket{\psi|t,-}\,,\q\forall\,\ket{\psi}\in\cd(\hat T)\,,\q t\in\mathbb{R}\,.\nn
\ea
By contrast, now choose $\ket{\psi'}\in\ch_C\setminus\cd(\hat T)$ defined by the wave function $\psi'(p_t)=N\,\sqrt{|p_t|}\,\exp(-p_t^2)$, where $N$ is a normalization constant. In this case, Eq.~\eqref{useful3} yields $\bra{\psi'}\,\hat T\,\ket{t,-} = t\,\braket{\psi'|t,-}+i\,s/N\neq t\,\braket{\psi'|t,-}$. Hence, $\ket{t,-}$ is not algebraically an eigenstate of $\hat T$.
\end{proof}

\noindent{\bf Lemma~\ref{thm_sameT2}.}
\emph{The $n^\text{th}$-moment operator defined in Eq.~\eqref{degnthmom} satisfies $[\hat T^{(n)},\hat H_C] = i\,n\,\hat T^{(n-1)}$. Furthermore, $\forall\ket{\psi} \in\mathcal{D}(\hat{T}^{n})$ we have $\hat T^{(n)}\ket{\psi}= \hat T^n\ket{\psi}$.}

\begin{proof}
To prove the first statement, consider
\ba
    U_C(s) \hat{T}^{(n)} U^\dag_C(s) = \f{1}{2\pi} \sum_{\sigma }\int_\mathbb{R} dt \,(t-s)^n \ket{t,\sigma}\!\bra{t,\sigma},
\ea
which follows from the covariance property of the clock states and a shift of integration variables. Differentiating both sides with respect to $s$, and then setting $s=0$, one finds $[\hat T^{(n)},\hat H_C] = i\,n\,\hat T^{(n-1)}$. Now, to prove the second statement, assume $\ket{\psi}\in\cd(\hat{T}^n)$. Then by direct calculation using Eq.~\eqref{degnthmom}, we have
\ba
    \hat T^{(n)}\ket{\psi} = \f{1}{2\pi} \sum_{\sigma }\int_\mathbb{R} dt \,t^n \ket{t,\sigma}\!\braket{t,\sigma |\psi} =  \f{1}{2\pi} \sum_{\sigma }\int_\mathbb{R} dt \,t^{n-1} \ket{t,\sigma}\!\bra{t,\sigma}\hat{T}\ket{\psi} = \hat T^{(n-1)}\hat{T}\ket{\psi}
\ea
where the second equality follows from Lemma~\ref{lem_almosteigen}. Repeating the same procedure $n-1$ more times, one obtains $\hat T^{(n)}\ket{\psi}=\hat T^{n}\ket{\psi}$.
\end{proof}

\noindent{\bf Lemma~\ref{lem_projTcommute}.}
\emph{The effect density $E_T(dt)$ of the covariant clock POVM and the projectors onto the $\sigma$-sectors commute: $[E_T(dt),\theta(-\sigma\,\hat p_t)]=0$.}

\begin{proof}
We begin by noting that
\ba
\Big\langle t,\sigma\Big|-\sigma'|p_t|\Big\rangle &=& \delta_{\sigma\sigma'} \sqrt{|p_t|}\,e^{i\,t\,s\,{p_t^2/2}}\,,\nn
\ea
which implies for the negative frequency sector the following
\ba
E_T(dt) \,\theta(\,\hat p_t)&=&\f{1}{2\pi}\,\sum_\sigma\, dt \,\ket{t,\sigma}\bra{t,\sigma}\,\int_{0}^\infty\,dp_t\,\ket{p_t}\bra{p_t}\nn\\
&=&\f{1}{2\pi}\, dt\, \int_0^\infty\,dp_t\,\sqrt{p_t}\,e^{i\,t\,s\,p_t^2/2}\,\ket{t,-}\bra{p_t}\nn\\
&\underset{(\ref{degclock})}{=}&\f{1}{2\pi}\, dt\,\ket{t,-}\bra{t,-}\,.\nn
\ea
By symmetry, this result is identical to $\theta(\,\hat p_t)\, E_T(dt)$. The corresponding result for the positive frequency sector is proven in exactly the same manner.
\end{proof}

\section{Reduced phase space quantization in the degenerate case}\label{app_redQT}

The quantum symmetry reduction procedure of Sec.~\ref{sec_trinity}, which constitutes a quantum deparametrization is the quantum analog of a classical phase space reduction by gauge-fixing. In some cases, the quantum symmetry reduction of the Dirac quantized theory is equal to the quantization of the classically reduced theory. To clarify this, we explain the classical phase space reduction (i.e.\  classical deparametrization) and subsequent reduced quantization of the class of models defined by the Hamiltonian constraint in Eq.~\eqref{constraint2}. We will also discuss the relational dynamics in these reduced classical and quantum theories. This extends the discussion in \cite{hoehnHowSwitchRelational2018,Hoehn:2018whn, Hoehn:2019owq} to the case of clock variables which are conjugate to the degenerate clock Hamiltonian (see also \cite{Thiemann:2004wk,chataignier2020relational}).

\subsection{Deparametrization through classical phase space reduction}\label{sec_redQT2}

Owing to the degeneracy reflected in Eq.~\eqref{Cdecomp}, we have to construct two reduced phase spaces, one each for the positive and negative frequency sectors. Indeed, reduction involves solving the constraint in Eq.~\eqref{constraint2} for the redundant (here temporal reference system, i.e.\ clock) variables and for each value of $s\,H_S<0$, we have two solutions for $p_t$. We gauge fix the clock variable to $T=0$. We are free to do so without discarding the information about the relational dynamics, as the relational observables in Eq.~\eqref{F2} are constant along the flow of $C_H$. We are thus free to evaluate them anywhere on the dynamical orbit and it is in any case the evolution parameter $\tau$ which keeps track of time evolution.
 $T=0$ is a good gauge fixing everywhere, except on $\cc_+\cap \cc_-$ where this condition is not defined.\footnote{There are in any case no good gauge fixing conditions (in the sense of intersecting every orbit once and only once) on $\cc_+\cap\cc_-$ involving only the (to be chosen as redundant) clock variables. Only functions with non-trivial dependence on $t$ can be used to fix the flow of the $s\, p_t^2/2$ term, but on $p_t=0$ $t$ does not evolve. Hence, any condition $g(p_t,t)=0$ could for $p_t=0$ only be solved (at most) for a single value of $t$ and would thereby miss all orbits on $p_t=0$ with differing $t$-values. (See also \cite{hoehnHowSwitchRelational2018,Hoehn:2018whn}.)} The gauge fixed reduced phase spaces, which we construct for our purposes are $\cp_\sigma\simeq (\cc_\sigma\setminus(\cc_+\cap\cc_-))\cap\cs_{T=0}$, where $\cs_{T=0}$ is the set in $\cp_{\rm kin}$ where $T=0$ is defined.  These two reduced phase spaces will look ``exactly the same", however.
 
The Dirac bracket, defining the symplectic structure on each $\cp_\sigma$, reads in this case
\ba
\{F,G\}_D:=\{F,G\}-\{F,C_H\}\{T,G\}+\{F,T\}\{C_H,G\}\nn
\ea
$\forall\,F,G\text{ on }\cc$. Restricting to functions $f_S,g_S$, depending on only the system variables, we have $\{f_S,g_S\}_D\equiv \{f_S,g_S\}$. 
Moreover, we can drop the redundant (and fixed) clock variables $(T,p_t)$, which satisfy $\{T,p_t\}_D=0$, using only the system variables to parametrize $\cp_\sigma$. 

To remember that the functions corresponding to system degrees of freedom now live on the phase spaces $\cp_\sigma$, we equip them with the label $\sigma$, although as functions of the basic system variables they will be the same as on $\cp_S$. Since the constraint Eq.~\eqref{constraint2} requires $s\,H^\sigma_S\leq0$ for the reduced system Hamiltonian, this may impose restrictions on the range of system variables compared to the original $\cp_S$ \cite{Ashtekar:1982wv,hoehnHowSwitchRelational2018}. In particular, $\cp_\sigma$ and $\cp_S$ need not be isomorphic.

The relational Dirac observables Eq.~(\ref{F2}) reduce under this procedure to
\ba
f_S^\sigma(\tau) = \sum_{n=0}\,\f{\left(-\tau\right)^n}{n!}\,\{H^\sigma_S,f^\sigma_S\}_n\,,\label{redobs2}
\ea
 and satisfy the standard evolution equations on $\cp_\sigma$
\ba
\f{df_S^\sigma}{d\tau}=\{f_S^\sigma,H^\sigma_S\}_D\,.\label{redevol2}
\ea
There is no more redundancy and the theory is deparametrized: there is no more gauge parameter and the clock degrees of freedom have disappeared from among the set of dynamical degrees of freedom. This is consistent as we do not wish to describe the temporal reference system relative to itself. As such, we can interpret this reduced relational dynamics as the dynamics described relative to the clock $T$ \cite{hoehnHowSwitchRelational2018,Hoehn:2018whn}.

The reason the dynamics for both the negative and positive frequency sectors look identical is, of course, that, as noted in Sec.~\ref{sec_classrel}, the clock $T$ runs `forward' with unit speed along the flow generated by $C_H$ on both $\cc_+$ and $\cc_-$, in contrast to $t$.

We note that, since we have ignored $\cc_+\cap\cc_-$ in our reduction, the reduced phase spaces $\cp_\sigma$ do not contain a boundary $H^\sigma_S=0$. (It is possible to regularize the classical theory to also include this boundary, e.g., see \cite{hoehnHowSwitchRelational2018,Hoehn:2018whn}, however, here we shall ignore such subtleties.) Conceptually, these issues are not surprising: Eq.~(\ref{redevol2}) shows that $H^\sigma_S$ is the generator of the system's evolution in the parameter $\tau$, which corresponds to the values that the dynamical clock $T$ takes. Since the latter is ill-defined for $p_t=H_S=0$, it is consistent that the reduced theories do not contain a boundary where $H^\sigma_S=0$.

\subsection{Relational dynamics in reduced quantization}

We proceed with the quantization of the gauge fixed reduced phase spaces $\cp_\sigma$. This amounts to finding a quantum representation of a system observable (sub-)algebra on suitable Hilbert spaces  $\tilde\ch_S^\sigma$. In this appendix, we denote the objects of reduced quantization with a tilde in order to distinguish them from the corresponding objects in the quantum symmetry reduced theory in Sec.~\ref{sec_trinity} which look structurally similar. In particular, we denote the quantization of the classical reduced Hamiltonian $H_S^\sigma$ by $\hat{\tilde{H}}_S^\sigma$.
Using the eigenbasis of the quantum Hamiltonian $\hat{\tilde{H}}_S^\sigma$, reduced states take the form\footnote{Should the spectrum be degenerate, we would have to add additional degeneracy labels.}
\ba
\ket{\tilde\psi_S^\sigma}= {\ \,{\intsum}_{E\in\spec(\hat{\tilde{H}}_S^\sigma)}} \,\tilde\psi_S^\sigma(E)\,\ket{E}_S\,.\label{redstatespm}
\ea

Assuming 
${\ \,{\intsum}_{E\in\spec(\hat{\tilde{H}}_S^\sigma)}}\,\braket{E'|E}_S\,f(E)=f(E')$ for an arbitrary complex function $f$, the inner product on $\tilde\ch_S^\sigma$  reads
\ba
\braket{\tilde\psi_S^\sigma|\tilde\phi_S^\sigma}={\ \,{\intsum}_{E\in\spec(\hat{\tilde{H}}_S^\sigma)}} \,\tilde\psi_S^\sigma(E)^*\,\tilde\phi_S^\sigma(E)\,.\label{redIP2}
\ea
Note that the precise representation (in particular, the measure and normalization of $\ket{E}_S$) will depend on the details of the system. For instance, if the system Hamiltonian was given by $\hat H_S =\hat p$, then the reduced Hilbert spaces would correspond to affinely quantized theories \cite{hoehnHowSwitchRelational2018}.

Finally, we quantize the evolving reduced observables Eq.~(\ref{redobs2}) as
\ba
\hat{\tilde f}^\sigma_S(\tau) &=& \sum_{n=0}^{\infty}\, \frac{(i\tau)^n}{n!}  \left[ \hat{\tilde H}^\sigma_S,\hat{\tilde f}^\sigma_S \right]_n \nn \\
&=& e^{i \tau \hat{\tilde H}^\sigma_S} \hat{\tilde f}^\sigma_S e^{-i \tau \hat{\tilde H}^\sigma_S}\,,\label{RO3}
\ea
which satisfy the relational Heisenberg equations
\ba
\f{d\hat{\tilde f}_S^\sigma}{d\tau}=i\,[\hat{\tilde H}^\sigma_S,\hat{\tilde f}^\sigma_S]\,.\nn
\ea

Structurally, this quantization of the classically reduced theory, incl.\ the dynamics, looks very similar to the quantum theory obtained through quantum symmetry reduction of $\ch_{\rm phys}$ in Sec.~\ref{sssec_qdepar}. The former is obtained through reduced, the latter through Dirac quantization.
However, given the possibly different value sets of $s\,H_S$ on $\cp_{\rm kin}$ and $s\,H_S^\sigma<0$ on $\cp_\sigma$, we emphasize that one need not in general expect that the spectrum $\sigma_{SC}$ in Eq.~\eqref{spectrum} of the system Hamiltonian on $\ch_{\rm phys}$ coincides with $\spec(\hat{\tilde H}_S^\sigma)$. Classical value restrictions may severely impact the domain where $\hat{\tilde H}_S^\sigma$ is self-adjoint and thereby its spectrum \cite{isham2}. Thus, Dirac and reduced quantization will not always be equivalent \cite{Ashtekar:1982wv,Kuchar:1986jj,Schleich:1990gd,Romano:1989zb,Loll:1990rx,Kunstatter:1991ds,Hoehn:2019owq}. There are, however, models, where $\sigma_{SC}=\spec(\hat{\tilde H}_S^\sigma)$, e.g., if $H_S$ is (minus) the Hamiltonian of a harmonic oscillator or a free particle, or if $H_S=p$ for some canonical momentum $p$. In this case, Dirac and reduced quantization are equivalent and yield two faces of the same relational quantum dynamics, as shown in \cite{hoehnHowSwitchRelational2018,Hoehn:2018whn}.

\end{document}